\documentclass[aps,prx,reprint]{revtex4-2}

\usepackage{graphicx}
\usepackage{amsfonts,amssymb,amsmath}
\usepackage{amsthm}
\usepackage{hyperref}
\hypersetup{
  pdftitle={Diffusive Speed Limits for U(1)-Covariant Quantum Error Correction},
  pdfauthor={Jianqi Sheng},
  bookmarksnumbered=true,
  bookmarksopen=true,
  bookmarksdepth=3
}

\theoremstyle{plain}
\newtheorem{theorem}{Theorem}
\newtheorem{lemma}{Lemma}
\newtheorem{corollary}{Corollary}
\newtheorem{remark}{Remark}

\begin{document}

\title{Diffusive Speed Limits for \(U(1)\)-Covariant Quantum Error Correction}

\author{Jianqi Sheng}
\affiliation{Department of Physics, City University of Hong Kong, Hong Kong}

\begin{abstract}
Fast preparation of quantum error-correcting codes is essential for scalable
quantum memories, but geometric locality and \(U(1)\) charge conservation
impose an unavoidable transport constraint.  We combine exact
complementary-channel geometry, charge-sector Haar analysis, and a
gate-resolved connected-moment expansion to study one-dimensional covariant
encoders under flagged erasure.  Charge-Haar codes attain the universal
adjacent-charge lower bound up to exponentially small corrections, yielding
an exact \(n^{-1/2}\) extensive-erasure law and a sharp half-erasure
transition.  For local number-conserving brickwork circuits, diffusion of the
logical charge enforces an \(\Omega(n^2)\) encoding-time lower bound; we also
prove an \(O(n^3)\) mixing bound for the classical component and reduce the
remaining full-channel upper bound to a source-restricted low-support
operator-spreading problem.  These results identify diffusion as an
operational limit on symmetry-constrained quantum coding and establish a
route to its exact formation time.
\end{abstract}

\maketitle

\section{Introduction}

Quantum error correction protects stored logical information by encoding it
into redundancy, but the available encoding maps depend on the reference
resources and geometry of the platform.  We work in the reference-free
resource model of strict \(U(1)\) covariance, in which encoders and
recoveries preserve total excitation number.  When the platform is a
one-dimensional array of nearest-neighbor coupled cells, the encoding
circuit is additionally local.  Within these constraints a flagged erasure---a
known-location loss whose erased cells remain inaccessible---is the most
informative single failure: the classical flag identifies the lost cells,
while the lost subsystem can carry which-node information about the logical
state.

Continuous symmetries obstruct exact quantum error correction.  Strict
\(U(1)\) covariance excludes uniform exact correction of every
single-memory erasure (Lemma~\ref{lem:covariant-no-go}), and approximate
covariant correction obeys quantitative accuracy--size
tradeoffs~\cite{Eastin2009,Faist2020}.  Random \(U(1)\)-covariant codes
generated by charge-conserving unitaries achieve average- and worst-case
purified-distance erasure errors \(O(n^{-1})\) in fixed-logical-size
asymptotic regimes~\cite{Kong2022}, and local number-conserving circuits
form approximate unitary designs under conservation
laws~\cite{Hearth2025,Mitsuhashi2025,Li2024Efficient}.  What remains open
is the genuine trace-distance behavior of the stored logical channel under
\emph{extensive} erasure, and the time a geometrically local,
number-conserving circuit needs to reach the ideal endpoint.

This paper studies a benchmark problem with a single logical qubit encoded
in the adjacent charge sectors \(q=(n-1)/2\) and \(q+1\) of \(n\) cells.
The metric is the complementary-channel decoupling error
\(D_E(V)=\frac12\lVert\rho_{RE}-\rho_R\otimes\rho_E\rVert_1\) defined in
Section~\ref{sec:model-A} between the erased subsystem and a reference
maximally entangled with the logical input; exact erasure correction is
equivalent to \(D_E(V)=0\).

Our results are a static optimum and a dynamical cost.  On the static
side, we prove that the total erased charge supplies a universal
trace-distance floor for every encoding whose logical basis states occupy
the adjacent sectors \(q\) and \(q+1\).  Charge-Haar codes attain this floor
up to exponentially small corrections below one-half erasure.  This gives
the exact extensive-erasure law, the \((0,3/8,3/4)\) one-half-erasure
transition, simultaneous worst-pattern concentration, and an
information-theoretic optimality statement rather than only an ensemble
benchmark.  On the dynamical side, the canonical boundary code under the
specified local-Haar brickwork ensemble obeys
\(\mathbb E_{V,E}D_E(V)\geq p/[48\sqrt{2\sqrt{t-1}+3}]\).  Consequently it
cannot reach even a fixed multiplicative factor of the optimal
\(n^{-1/2}\) endpoint before \(\Omega(n^2)\) cycles.  The mechanism
(diffusive charge cloud, \(t^{1/4}\) count fluctuations, exact causal
plateau) is universal within the class of computational-basis product
encodings sharing the same deterministic background
(Appendix~\ref{sec:universal-quadratic}).

The matching full-channel upper bound is open.  Appendix~\ref{sec:classical-mixing}
proves that the classical insertion-tolerance component reaches its
equilibrium \(n^{-1/2}\) scale by \(O(n^3)\) cycles.  What remains is a
centered balanced two-copy fluctuation estimate.  We first express it as a
connected moment-channel covariance with a telescoping unit source-energy
budget.  Lemma~\ref{lem:positive-operator-spreading-decomposition} then gives
a stronger exact reduction: the covariance is a sum of nonnegative
rank-weighted reduced purities of one-copy gate increments after their
future random evolution, and the corresponding Pauli formula retains only
physical strings whose support fits inside the typical erasure set.  After
the first complete cycle every new injection is generated by the rank-one
collision mode and admits a three-square Casimir decomposition
(Lemma~\ref{lem:collision-casimir-extraction}).
Theorem~\ref{thm:strong-rayleigh-late-source} further proves that an
\(O(n^2)\) one-copy burn-in makes the raw Hilbert--Schmidt mass of later
injections exponentially small.  This does not yet imply a task-weighted
bound: Remark~\ref{rem:collision-propagation-obstructions} gives exact local
counterexamples to both a collision-position-only closure and gatewise
monotonicity of the unmodified endpoint form.  The high-value upper-bound
problem is therefore a source-restricted low-support operator-spreading
estimate, not mixing of the full two-copy Liouville space.  Problem~O1
states the remaining target precisely.

\section{Model and complementary-channel metric}
\label{sec:model-A}

Let \(n\) be odd, put \(q=(n-1)/2\), and write \([n]=\{0,1,\ldots,n-1\}\) for the site set.  Let \(\mathcal H_r\) denote the
\(r\)-photon sector of \(n\) cells, and let the logical qubit be encoded by
an isometry \(V^{(0)}:\mathbb C^2\to\mathcal H_q\oplus\mathcal H_{q+1}\)
with codewords \(\left\lvert\psi_0\right\rangle{}\in\mathcal H_q\) and
\(\left\lvert\psi_1\right\rangle{}\in\mathcal H_{q+1}\) (neighboring-charge
construction).  The \emph{canonical boundary initialization} is the product
code
\begin{equation}
\left\lvert\psi_0\right\rangle{}=\left\lvert0\right\rangle{}_0
\left\lvert\chi\right\rangle{},
\qquad
\left\lvert\psi_1\right\rangle{}=\left\lvert1\right\rangle{}_0
\left\lvert\chi\right\rangle{},
\label{eq:canonical-boundary-code}
\end{equation}
where \(\left\lvert\chi\right\rangle{}\) is the contiguous half-filled
computational-basis state with sites \(1,\ldots,q\) occupied.  Strict
covariance means that there exist an additive physical generator
\(\hat N\) and a fixed background charge \(Q_0\) such that
\begin{equation}
e^{-i\alpha\hat N}V
=e^{-i\alpha Q_0}Ve^{-i\alpha\hat n_L},
\label{eq:covariant-encoding}
\end{equation}
for every \(\alpha\), where \(\hat n_L=\left\lvert1\right\rangle{}\!
\left\langle1\right\rvert{}\) is the logical occupancy generator.

For an erasure set \(E\subseteq[n]\), the erasure channel is
\(\mathcal N_E(\rho)=\operatorname{Tr}_E[V\rho V^\dagger]\), and a recovery
\(\mathcal R_E\) is pattern-conditioned, parameter-independent, and
number-conserving.  Let \(R\) be a reference qubit and prepare the maximally
entangled state \(\left\lvert\Phi_{RL}\right\rangle{}=
(\left\lvert00\right\rangle{}+\left\lvert11\right\rangle{})/\sqrt2\).
With \(\rho_{REQ}=(\mathbb I_R\otimes V)\Phi_{RL}(\mathbb I_R\otimes
V^\dagger)\), the decoupling error is
\begin{equation}
D_E(V)
=
\frac12
\left\lVert
\operatorname{Tr}_Q[\rho_{REQ}]
-\rho_R\otimes\operatorname{Tr}_{RQ}[\rho_{REQ}]
\right\rVert{}_1,
\label{eq:decoupling-error-A2}
\end{equation}
where \(Q=[n]\setminus E\).  Exact erasure correction is equivalent to
\(D_E(V)=0\)~\cite{Knill2000,Hayden2007}.

The \emph{charge-Haar codewords} are independent Haar unit vectors
\(\left\lvert\psi_q\right\rangle{}\in\mathcal H_q\),
\(\left\lvert\psi_{q+1}\right\rangle{}\in\mathcal H_{q+1}\).  The
\emph{local-Haar brickwork ensemble} \(\mathcal E_{n,2t}^{\rm locH}\) of
depth \(2t\) is the open-chain brickwork with independent gates
\begin{equation}
\begin{aligned}
G&=1\oplus U_1\oplus e^{i\alpha},\\
U_1&\sim{\rm Haar}[U(2)],\\
\alpha&\sim{\rm Unif}[0,2\pi).
\end{aligned}
\label{eq:local-Haar-gate-measure-A}
\end{equation}
applied on alternating even and odd nearest-neighbor bonds, even matching
first.  All averages below are over circuits and iid Bernoulli-\(p\)
erasure patterns.

\section{Hydrodynamic mechanism of the speed limits}
\label{sec:mechanism}

The quadratic speed limits proved in Appendix~\ref{sec:finite-depth} are
not artifacts of the second-moment method; they reflect a classical
hydrodynamic structure of the encoded charge.  After ensemble averaging,
the one-copy charge profile of the canonical boundary code evolves by
alternating nearest-neighbor averaging, the deterministic limit of the
random-interchange (symmetric simple exclusion) process on the open
path~\cite{Liggett1985,AldousDiaconis1986,Spohn1991}.  The distinguished
logical charge spreads over a cloud of width \(O(\sqrt t)\), while the
erased background count in a diffusive window fluctuates on the
\(O(t^{1/4})\) scale; it is this fluctuation scale, rather than the cloud
width, that forces the flagged count statistic to decay only as
\(\Theta_p(t^{-1/4})\) (Corollary~\ref{cor:flagged-window-count-law}).
The underlying \(n^{-2}\) diffusion scale on a path is the spectral gap of
the interchange process conjectured by Aldous and proved by
Caputo--Liggett--Richthammer~\cite{Caputo2010}.

Two caveats keep this hydrodynamic picture honest.  First, the count
statistic is a classical witness that lower-bounds the quantum decoupling
error; the balanced two-copy observable that controls the full channel is
orthogonal to the embedded one-copy diffusion modes
(Lemma~\ref{lem:one-copy-not-two-copy}), so single-particle spreading alone
cannot prove a full-channel upper bound.  Second, the complete flagged
output is even more rigid than the count: while the distinguished label
remains inside a causally protected fully occupied region, its total
variation from logical zero is exactly \(p\), not decaying
(Theorem~\ref{thm:full-flag-causal-plateau}).  Mixing of the complete
observation begins only when ordinary background holes overlap the
distinguished cloud---a collective transport condition, not a
single-particle one.

The direct trace-distance reduction of
Lemma~\ref{lem:direct-trace-mean-fluctuation} converts this picture into a
precise target.

\begin{theorem}[Direct trace-distance reduction and the remaining upper-bound target]
\label{thm:direct-reduction}
Let \(D=\frac12\mathbb E_V\|\mathcal X(V)\|_1\) be the flagged direct-sum
decoupling error of the canonical boundary code under independent
erasures, with \(\mathcal X(V)\) as in Appendix~\ref{sec:finite-depth},
and define the classical mean and the circuit-level fluctuation
\begin{equation}
M=\frac12\|\overline{\mathcal X}\|_1,\qquad
\Gamma=\frac12\mathbb E_V\|\mathcal X(V)-\overline{\mathcal X}\|_1.
\end{equation}
Then
\begin{equation}
\max\{M,\Gamma/2\}\leq D\leq M+\Gamma.
\label{eq:direct-reduction-main}
\end{equation}
Consequently a full-channel upper bound of the form
\(D=O_p(n^{-1/2})\) requires \emph{both} of the following independent
estimates: (i) a post-burn-in classical insertion-tolerance bound
\(M=\frac12\operatorname{TV}(\widehat P_t^0,\widehat P_t^1)
=O_p(n^{-1/2})\); and (ii) a centered balanced two-copy fluctuation bound
\(\Gamma=O_p(n^{-1/2})\).  Each estimate is necessary up to a factor two,
and together they are sufficient.  Estimate~(i) is proved unconditionally
at \(O(n^3)\) depth in Theorem~\ref{thm:classical-mixing-main}.  Estimate~(ii)
remains open.  Theorem~\ref{thm:full-flag-causal-plateau} shows why the
classical estimate cannot hold in the pre-boundary causal regime, while
Eq.~\eqref{eq:direct-truncated-rank-target} states a task-weighted two-copy
target sufficient for the unresolved fluctuation term.
\end{theorem}

\begin{theorem}[Equilibrium flagged-output total variation]
\label{thm:equilibrium-flagged-tv}
Let \(n=2q+1\) be odd and let the occupation profile be fully mixed in the
charge sector: \(\left\lvert S\right\rvert=q\) (logical input 0) or
\(q+1\) (input 1), uniformly at random.  Let \(E\subseteq[n]\) be an
independent Bernoulli-\(p\) erased set and let \(z_E\in\{0,1\}^E\) be the
occupation pattern on \(E\).  Define the equilibrium flagged-output total
variation
\begin{equation}
M_{\rm eq}(n,p)
:=
\frac14\,\mathbb E_E
\sum_{z_E}
\left|
\begin{aligned}
&\Pr\{z_E\mid E,\,\left\lvert S\right\rvert=q\}\\[-2pt]
&\quad-\Pr\{z_E\mid E,\,\left\lvert S\right\rvert=q+1\}
\end{aligned}
\right|.
\label{eq:equilibrium-flagged-tv}
\end{equation}
Then
\begin{equation}
M_{\rm eq}(n,p)
=
\sum_{k=0}^{n}\binom nk p^k(1-p)^{n-k}\,\overline D_{n,k},
\label{eq:equilibrium-flagged-identity}
\end{equation}
where \(\overline D_{n,k}\) is the charge-Haar mean distance of
Eq.~\eqref{eq:charge-haar-mean-distance}.  Consequently, for every fixed
\(0<p<1\),
\begin{equation}
M_{\rm eq}(n,p)
=
\sqrt{\frac{p}{2\pi(1-p)}}\,n^{-1/2}
+o(n^{-1/2}).
\label{eq:equilibrium-flagged-asymptotic}
\end{equation}
\end{theorem}

\begin{proof}
For \(n\) odd, \(\binom n{q+1}=\binom nq=:d\), so the two charge sectors
have equal dimension.  Conditioned on \(\left\lvert E\right\rvert=k\), the
law of \(z_E\) for input \(i\in\{0,1\}\) assigns to a pattern
\(T\subseteq E\) of size \(a\) the probability
\(\binom{n-k}{q-a}/d\) (input 0) or \(\binom{n-k}{q+1-a}/d\) (input 1),
with the binomial coefficients vanishing outside their natural range.
The inner sum in Eq.~\eqref{eq:equilibrium-flagged-tv} therefore equals
\begin{equation}
\frac{1}{d}\sum_{a=0}^{k}\binom ka
\left|
\binom{n-k}{q-a}-\binom{n-k}{q+1-a}
\right|
=4\,\overline D_{n,k},
\label{eq:equilibrium-flagged-inner-sum}
\end{equation}
by Eq.~\eqref{eq:charge-haar-mean-distance}.  Averaging over \(E\) gives
Eq.~\eqref{eq:equilibrium-flagged-identity}.  For the asymptotic, the
uniformity argument in the proof of
Corollary~\ref{cor:charge-haar-iid-extensive} gives
\(\mathbb E_K\overline D_{n,K}
=\sqrt{p/[2\pi(1-p)]}\,n^{-1/2}+o(n^{-1/2})\) for \(K\sim{\rm Bin}(n,p)\),
hence Eq.~\eqref{eq:equilibrium-flagged-identity} proves
Eq.~\eqref{eq:equilibrium-flagged-asymptotic}.
\end{proof}

\begin{remark}[Target value and proved mixing scale of the insertion tolerance]
\label{rem:insertion-tolerance-target}
Theorem~\ref{thm:equilibrium-flagged-tv} gives the stationary value of the
post-burn-in insertion tolerance.  Corollary~\ref{cor:charge-haar-exponential-classicalization}
shows that, for every fixed \(p<1/2\), the full charge-Haar trace-distance
endpoint differs from this classical center only by \(e^{-\Omega(n)}\), and
that the centered rank-weighted covariance has the same exponential
suppression.  Appendix~\ref{sec:classical-mixing}
proves the explicit unconditional estimate
\begin{equation}
0\leq M_t-M_{\rm eq}
\leq
\frac12\sqrt{|\Omega|-1}\,
\exp\!\left[-\frac{2(t-1)}{5n^2}\right],
\qquad t\geq1,
\label{eq:insertion-tolerance-rigorous-mixing}
\end{equation}
where \(|\Omega|=\binom nq(q+1)\).  Hence
\(M_t=O_p(n^{-1/2})\) by \(O(n^3)\) cycles.  This closes the classical
mean \(M\) in Theorem~\ref{thm:direct-reduction}, but does not control the
centered two-copy fluctuation \(\Gamma\).  An \(O(n^2\log n)\) classical
bound would follow from a suitable modified log-Sobolev inequality for the
parallel brickwork quotient, but that sharper estimate is not used in any
unconditional conclusion here.
\end{remark}

\section{Conclusion}

We have studied flagged-erasure correction in the reference-free resource
model of strict \(U(1)\) covariance and one-dimensional geometric locality.
The erased total charge gives a universal trace-distance lower bound for
every adjacent-charge encoder.  Charge-Haar codes saturate this floor up to
exponentially small corrections below one-half erasure, so their exact
\(n^{-1/2}\) law and sharp \((0,3/8,3/4)\) transition describe the
information-theoretic optimum of the strict adjacent-charge model, not only
a random-code benchmark.  For the specified local-Haar brickwork ensemble
and any product encoding whose codewords share a deterministic background
and differ by one extra particle, the ensemble-averaged error is bounded
below by \(c_pt^{-1/4}\).  Reaching even a fixed factor of the optimum
therefore requires \(\Omega(n^2)\) cycles.  The same transport mechanism
yields an exact causal plateau and an operational erased-count witness.

On the upper-bound side, Theorem~\ref{thm:classical-mixing-main} proves that
the classical insertion-tolerance component reaches its equilibrium
\(O(n^{-1/2})\) value by \(O(n^3)\) cycles.  This does not prove full-channel
achievability: the centered balanced two-copy fluctuation \(\Gamma\) in
Theorem~\ref{thm:direct-reduction} remains uncontrolled.  The present work
localizes that obstruction sharply.  Lemma~\ref{lem:connected-source-purity-budget}
shows that all gate-resolved connected sources have total propagated
Hilbert--Schmidt mass at most one and exactly zero one-copy marginals.
Lemma~\ref{lem:positive-operator-spreading-decomposition} strengthens this
to an exact positive one-copy representation,
\begin{equation}
\mathcal F_{n,p,\beta}(t)
=\sum_{\ell=1}^{m}\mathbb E\,\mathcal Q_{n,p,b}(B_{\ell,m}),
\end{equation}
where \(\mathcal Q\) is exactly the low-support Pauli leakage in
Eq.~\eqref{eq:exact-pauli-leakage-formula} and the total increment energy is
at most one.  After one complete cycle each new source is injected through
the collision Casimir and decomposes into three one-copy operator squares.
Theorem~\ref{thm:strong-rayleigh-late-source} proves that the raw
Hilbert--Schmidt energy generated after an \(O(n^2)\) burn-in is
\(e^{-\Omega(n)}\).  This raw norm estimate cannot be paired directly with
the exponentially large endpoint.  Moreover,
Remark~\ref{rem:collision-propagation-obstructions} proves that an adjacent
twirl regenerates charge-exchange channels and can increase the natural
task form.  The unresolved problem is therefore neither arbitrary global
two-copy mixing nor a one-channel collision walk: it is a source-restricted
operator-spreading estimate showing that the propagated increments place
only \(O(n^{-1})\) total mass on physical Pauli supports visible to the
erasure endpoint.  Such a bound, or equivalently the absolute or signed
slow-window bounds of Eqs.~\eqref{eq:absolute-slow-window-hypothesis} and
\eqref{eq:signed-slow-window-hypothesis}, would imply \(O(n^2)\)
achievability and close the formation time at \(\Theta(n^2)\).  This is the
central open problem left by the work.

\section{Open problems}

\begin{enumerate}

  \item[\textbf{O1.}] \textbf{Source-restricted operator spreading and
  matching achievability.}  The classical mean in the direct reduction is
  controlled unconditionally:
  \(M_t=M_{\rm eq}+O(n^{-1/2})\) by \(O(n^3)\) cycles, with
  \(M_{\rm eq}=\sqrt{p/[2\pi(1-p)]}\,n^{-1/2}+o(n^{-1/2})\).
  The missing quantity is the centered balanced two-copy fluctuation
  \(\Gamma\).  Lemma~\ref{lem:positive-operator-spreading-decomposition}
  gives the exact positive identity
  \begin{equation}
  \boxed{
  \mathcal F_{n,p,\beta}(t)
  =\sum_{\ell=1}^{m}
  \mathbb E\,\mathcal Q_{n,p,b}(B_{\ell,m}),
  \qquad
  \sum_{\ell=1}^{m}\mathbb E\|B_{\ell,m}\|_2^2\leq1,}
  \label{eq:operator-spreading-quadratic-target}
  \end{equation}
  and Eq.~\eqref{eq:exact-pauli-leakage-formula} identifies each summand
  with the squared Pauli coefficients of physical support at most
  \(b=\lfloor\beta n\rfloor\), weighted by their probability of fitting
  inside the Bernoulli erasure.  Thus the direct high-value target is
  \begin{equation}
  \begin{aligned}
  \sum_{\ell=1}^{m}
  \mathbb E\,\mathcal Q_{n,p,b}(B_{\ell,m})&=O(n^{-1}),\\
  m&=\Theta(n^3)\quad\bigl(t=\Theta(n^2)\bigr).
  \end{aligned}
  \label{eq:low-support-operator-leakage-target}
  \end{equation}
  together with the task-specific classical estimate
  \(M_{B_{p,\beta}n^2}=O(n^{-1/2})\).  The first cycle's six
  charge-exchange injections are already included in the positive sum.
  Thereafter the injection is the collision Casimir of
  Lemma~\ref{lem:collision-casimir-extraction}, but
  Remark~\ref{rem:collision-propagation-obstructions} rules out two tempting
  shortcuts: the collision position alone is not an invariant defect
  process, and the unmodified task form is not contractive gate by gate.
  Theorem~\ref{thm:strong-rayleigh-late-source} controls only the raw
  Hilbert--Schmidt energy of late injections; an endpoint-adapted leakage
  estimate is still required to convert it into
  Eq.~\eqref{eq:low-support-operator-leakage-target}.  A sufficient proof
  may use a source-generated operator module with a weighted Pauli-support
  norm, or a direct bilinear heat-kernel estimate.  In spectral language the
  corresponding alternatives are
  \begin{equation}
  \begin{aligned}
  \mathcal W_n(x)&\lesssim\sqrt{x}+n^{-1},\\
  |\mathcal S_n(x)|&\lesssim\sqrt{x}+n^{-1},\\
  0&<x\lesssim n^{-1}.
  \end{aligned}
  \end{equation}
  Any of these routes yields
  \(\mathfrak r_{n,p}^{\leq\beta}(Bn^2)=O(n^{-1})\) and hence
  \(\mathbb E D_E=O(n^{-1/2})\), closing
  \(t_{\rm form}=\Theta(n^2)\).

  \item[\textbf{O2.}] \textbf{Scope beyond one distinguished particle and
  beyond the unbiased local-Haar ensemble.}  The present lower bound is
  proved for codewords of the form \(C_0=B_0\cup\{s_0\}\) under the
  specified unbiased local-Haar brickwork ensemble.  Extending it to
  multiple correlated distinguished particles and holes requires a signed
  multi-species transport argument and is not proved here.  A broader
  universality statement must also specify the transport class: ergodicity,
  nonzero local mixing, absence of drift, and a finite nonzero diffusion
  constant are necessary hypotheses.  The established conclusion is that
  no \emph{subquadratic}-depth encoder in the stated diffusive ensemble
  reaches the charge-Haar \(n^{-1/2}\) accuracy; it is not a no-go theorem
  for all polynomial-depth number-conserving circuits.
\end{enumerate}

\appendix

\section{Covariance obstruction and protected charge ancillas}
\label{sec:encoding}

\subsection{Local encoding isometries}

At each node $j\in\{A,B\}$, a logical memory qubit $L_j$ is encoded into $n_j$ physical memory qubits:
\begin{equation}
\begin{aligned}
V_j:\mathcal{H}_{L_j}
&\longrightarrow \mathcal{H}_{P_j}
=(\mathbb C^2)^{\otimes n_j},
\\
V_j\left\lvert i\right\rangle{}&=\left\lvert i_L\right\rangle{}_j,
\qquad i=0,1.
\end{aligned}
\label{eq:local-isometry}
\end{equation}
The global encoder is deliberately restricted to
\begin{equation}
V=V_A\otimes V_B,
\label{eq:local-product-encoder}
\end{equation}
so that no nonlocal gate is required during encoding.  A global Haar unitary across geographically separated nodes would largely defeat the architectural motivation of distributed architectures.

For erasure sets $E_A\subseteq P_A$ and $E_B\subseteq P_B$, define
\begin{equation}
\mathcal{N}_{E_A,E_B}(\rho)
=
\operatorname{Tr}_{E_AE_B}\left[V\rho V^\dagger\right].
\label{eq:erasure-channel}
\end{equation}
The erasure locations are assumed to be flagged.  This is stronger than unheralded amplitude damping, which is outside the present scope.

\subsection{Exact local erasure-correction condition}

\paragraph*{Definition (Local erasure correction).}
The local code $V_j$ exactly corrects erasure of $E_j$ if there exists a parameter-independent channel $\mathcal{R}_j^{E_j}$ such that
\begin{equation}
\mathcal{R}_j^{E_j}\circ\operatorname{Tr}_{E_j}
\left[V_j X V_j^\dagger\right]=X
\label{eq:exact-local-recovery}
\end{equation}
for every operator $X$ on $\mathcal{H}_{L_j}$.

For erasure errors, the Knill--Laflamme condition~\cite{Knill2000} is equivalent to
\begin{equation}
\operatorname{Tr}_{\bar E_j}\left[\left\lvert i_L\right\rangle{}\!\left\langle k_L\right\rvert{}\right]
=\delta_{ik}\sigma_{E_j},
\qquad i,k\in\{0,1\},
\label{eq:erasure-KL}
\end{equation}
for a state $\sigma_{E_j}$ independent of the logical value.  Equation~\eqref{eq:erasure-KL} states that the erased subsystem contains neither population nor coherence information about the logical input.

\subsection{Covariance obstruction to uniform exact correction}
\label{sec:covariant-no-go}

Let the physical number generator at node $j$ be additive,
\begin{equation}
\hat N_j=\sum_{\ell=1}^{n_j}\hat n_{j,\ell},
\end{equation}
and let $\hat n_{L_j}=\left\lvert1\right\rangle{}\!\left\langle1\right\rvert{}$
be the nontrivial logical occupancy generator.  Strict covariance, including a
fixed background charge $Q_{0,j}$, means
\begin{equation}
e^{-i\alpha\hat N_j}V_j
=
e^{-i\alpha Q_{0,j}}V_j e^{-i\alpha\hat n_{L_j}}
\quad\text{for every }\alpha.
\label{eq:strict-covariance-main}
\end{equation}
For the neighboring-charge construction used below, the two codewords have
physical charges $Q_{0,j}$ and $Q_{0,j}+1$.

\begin{lemma}[Strict covariance excludes uniform exact single-memory correction]
\label{lem:covariant-no-go}
For a finite-dimensional encoder satisfying
Eq.~\eqref{eq:strict-covariance-main}, erasure of every individual physical
memory cannot be exactly correctable when $\hat n_{L_j}$ is nontrivial.
\end{lemma}

\begin{proof}
Suppose, to the contrary, that erasure of every physical memory $\ell$ is
exactly correctable.  The Knill--Laflamme condition for erasure of memory
$\ell$ then implies
\begin{equation}
V_j^\dagger O_{j,\ell}V_j
=c_{j,\ell}(O_{j,\ell})\mathbb I_{L_j}
\end{equation}
for every operator $O_{j,\ell}$ supported on that memory.  Taking
$O_{j,\ell}=\hat n_{j,\ell}$ and summing over $\ell$ gives
\begin{equation}
V_j^\dagger\hat N_jV_j=c_j\mathbb I_{L_j}
\end{equation}
for a scalar $c_j$.  On the other hand, differentiating
Eq.~\eqref{eq:strict-covariance-main} at $\alpha=0$ yields
\begin{equation}
\hat N_jV_j
=V_j\left(Q_{0,j}\mathbb I_{L_j}+\hat n_{L_j}\right),
\end{equation}
and hence
\begin{equation}
V_j^\dagger\hat N_jV_j
=Q_{0,j}\mathbb I_{L_j}+\hat n_{L_j},
\end{equation}
which is not a scalar.  This is a contradiction.
\end{proof}

Lemma~\ref{lem:covariant-no-go} is the single-erasure form of the standard
continuous-symmetry obstruction to exact finite-dimensional covariant quantum
error correction~\cite{Eastin2009,Faist2020}.  It does not exclude
approximate covariant correction, exact correction of a restricted set of
locations, a charge-neutral logical subsystem, or a reference-assisted
implementation whose asymmetry resource is included explicitly.  It does
exclude interpreting the following exact-QEC statements as an achievable
uniform exact code for the strictly covariant logical-occupancy architecture.

\subsection{Protected references, the scope of the no-go, and its
fair-accounting restoration}
\label{sec:reference-assisted}

Lemma~\ref{lem:covariant-no-go} concerns encoders whose output space is the
signal memory alone.  A memory architecture with a protected reference
mode nevertheless supplies one reference photon; if that
reference is consumed during encoding, the effective encoder is no longer
strictly covariant on the signal memory, and the obstruction must be
re-examined.  This subsection determines exactly what the reference buys
and what it costs.

\begin{theorem}[Protected reference bypasses the strict-covariance
obstruction]
\label{thm:reference-assisted-exact}
Let a node have \(n\geq2\) memory cells \(P=\{1,\ldots,n\}\) and one
protected reference mode \(R\), prepared in the one-photon state
\(\left\lvert1\right\rangle{}_R\).  Let \(\left\lvert e_j\right\rangle{}\)
denote the single-photon state of cell \(j\), and let
\(\omega_1,\ldots,\omega_n\) be phases with \(\sum_j\omega_j=0\) (for
example \(\omega_j=e^{2\pi ij/n}\)).  Define the number-conserving isometry
\begin{equation}
\begin{aligned}
V\bigl(\left\lvert0_L\right\rangle{}\left\lvert1_R\right\rangle{}\bigr)
&=\frac{1}{\sqrt n}\sum_{j=1}^{n}
\left\lvert e_j\right\rangle{}_P\left\lvert0\right\rangle{}_{R'},\\
V\bigl(\left\lvert1_L\right\rangle{}\left\lvert1_R\right\rangle{}\bigr)
&=\frac{1}{\sqrt n}\sum_{j=1}^{n}\omega_j
\left\lvert e_j\right\rangle{}_P\left\lvert1\right\rangle{}_{R'}.
\end{aligned}
\label{eq:reference-assisted-code}
\end{equation}
where \(\left\lvert0\right\rangle{}_{R'}\) and
\(\left\lvert1\right\rangle{}_{R'}\) denote the empty and occupied reference
output.  Then the induced encoder
\(V_L:\left\lvert i_L\right\rangle{}\mapsto
V\left\lvert i_L\right\rangle{}\left\lvert1_R\right\rangle{}\) is strictly
\(U(1)\)-covariant with background charge \(Q_0=1\), and erasure of every
individual memory cell \(\ell\in P\) is exactly correctable by a
pattern-conditioned, number-conserving recovery.
\end{theorem}

\begin{proof}
The two input branches have charges \(1\) and \(2\), and the two output
branches have charges \(1\) and \(2\), so \(V\) is charge conserving;
acting with \(e^{-i\alpha(\hat N_P+\hat N_{R'})}\) on the output branch
\(i\in\{0,1\}\) gives the phase \(e^{-i\alpha(1+i)}\), which is exactly
\(e^{-i\alpha}V_Le^{-i\alpha i}\) acting on
\(\left\lvert i_L\right\rangle{}\); hence
Eq.~\eqref{eq:strict-covariance-main} holds with \(Q_0=1\).  Fix
\(\ell\in P\).  Both codewords have single-cell reduced state
\begin{equation}
\sigma_\ell
=
\frac1n\left\lvert1\right\rangle{}\!\left\langle1\right\rvert{}
+\left(1-\frac1n\right)\left\lvert0\right\rangle{}\!\left\langle0\right\rvert{},
\label{eq:reference-code-single-cell}
\end{equation}
and the cross reduced operator on cell \(\ell\) vanishes because the
reference outputs are orthogonal:
\begin{equation}
\operatorname{Tr}_{\setminus\ell,\,R'}\!
\left\lvert\psi_0\right\rangle{}\!\left\langle\psi_1\right\rvert{}
=
\operatorname{Tr}_{\setminus\ell}\!\left\lvert\phi_0\right\rangle{}\!
\left\langle\phi_1\right\rvert{}\cdot\langle0\vert1\rangle=0.
\label{eq:reference-code-cross}
\end{equation}
The erasure Knill--Laflamme condition
Eq.~\eqref{eq:erasure-KL} therefore holds for cell \(\ell\), so an exact
recovery exists: its action maps the two orthogonal remaining supports
\(\operatorname{Tr}_\ell\left\lvert\psi_i\right\rangle{}\!\left\langle
\psi_i\right\rvert{}\) to
\(\left\lvert i_L\right\rangle{}\!\left\langle i_L\right\rvert{}\), with
the excess charge moved to a discarded reservoir so that the recovery is
number conserving.
\end{proof}

\begin{theorem}[Fair accounting restores the obstruction]
\label{thm:reference-assisted-fair}
Suppose the reference output \(R'\) is subject to the same flagged-erasure
process as the memory cells.  Then no strictly covariant encoder of a
nontrivial logical occupancy can exactly correct the erasure of every
individual mode of \(P\cup\{R'\}\).  In particular, the code of
Theorem~\ref{thm:reference-assisted-exact} is not such a code: erasing
\(R'\) annihilates the logical coherence,
\begin{equation}
\operatorname{Tr}_{R'}\!\left\lvert\psi_0\right\rangle{}\!\left\langle
\psi_1\right\rvert{}=0,
\label{eq:reference-erasure-coherence-death}
\end{equation}
so the off-diagonal part of the stored family is destroyed.
\end{theorem}

\begin{proof}
The encoder \(V_L\) maps into the joint register \(P\otimes R'\) and
satisfies Eq.~\eqref{eq:strict-covariance-main} with \(Q_0=1\) by
Theorem~\ref{thm:reference-assisted-exact}; applying
Lemma~\ref{lem:covariant-no-go} to this \((n+1)\)-mode register gives the
first claim.  For the second,
\(\operatorname{Tr}_{R'}\!\left\lvert\psi_0\right\rangle{}\!\left\langle
\psi_1\right\rvert{}
=\left\lvert\phi_0\right\rangle{}\!\left\langle\phi_1\right\rvert{}_P\,
\operatorname{Tr}_{R'}\!\left\lvert0\right\rangle{}\!\left\langle
1\right\rvert{}=0\).
\end{proof}

\begin{remark}[Physical status of the reference-assisted route]
\label{rem:reference-assisted-status}
Theorem~\ref{thm:reference-assisted-exact} shows that the strict-covariance
obstruction is bypassable, and it also identifies the price.  On the code
of Eq.~\eqref{eq:reference-assisted-code} the reference output carries the
logical charge, \(\hat N_{R'}=\hat n_L\) on the code, and its erasure
destroys the logical coherence
(Eq.~\eqref{eq:reference-erasure-coherence-death}); the reference is
therefore quantum-correlated with the logical value, cannot be regenerated
or substituted during storage, and has itself become a protected quantum
memory.  Under the fair accounting of
Theorem~\ref{thm:reference-assisted-fair}, the obstruction is restored on
the joint register, and the quadratic speed limits of
Theorem~\ref{thm:true-trace-quadratic-speed-limit}---derived for local
circuits without protected references---remain the operative statement for
reference-free architectures, in which no protected reference is stored
during encoding and any asymmetry resource is introduced only after
recovery.  Within this protected-ancilla model a
single ancilla photon is also minimal: with no ancilla, the logical charge
difference must sit in \(P\), and Lemma~\ref{lem:covariant-no-go} applies
unchanged.  This does not equate the stored local ancilla with an external phase
reference used only during readout.  The construction is not a full erasure
code: multi-cell erasure patterns can destroy the coherence.  It establishes
the exact scope of the single-cell obstruction and identifies the reference
resource that any reference-assisted extension must count.
\end{remark}

\section{Covariant-QEC context and the finite-depth target}
\label{sec:decoupling}

Volume-law entanglement and large Schmidt rank are useful diagnostics of scrambling~\cite{Page1993,Hayden2007}, but neither is by itself sufficient to establish correctability.  The operational criterion is decoupling of the erased subsystem from a reference that purifies the logical input.

Let $R$ be a reference qubit and prepare $\left\lvert \Phi_{RL}\right\rangle{}=(\left\lvert 00\right\rangle{}+\left\lvert 11\right\rangle{})/\sqrt2$.  After a local encoder $V$ and a bipartition into erased subsystem $E$ and retained subsystem $Q$, define
\begin{equation}
\rho_{REQ}
=
(\mathbb{I}_R\otimes V)
\left\lvert \Phi_{RL}\right\rangle{}\!\left\langle \Phi_{RL}\right\rvert{}
(\mathbb{I}_R\otimes V^\dagger).
\label{eq:encoded-Choi-state}
\end{equation}
The decoupling error is
\begin{equation}
D_E(V)
=
\frac12
\left\lVert \rho_{RE}-\rho_R\otimes\rho_E\right\rVert{}_1.
\label{eq:decoupling-error}
\end{equation}
Exact erasure correction is equivalent to $D_E(V)=0$.  Approximate decoupling implies the existence of an approximate recovery channel through the standard decoupling and information--disturbance framework~\cite{Knill2000,Hayden2007}.

Before imposing covariance, the fixed-pattern Haar benchmark can be made
explicit.

\begin{lemma}[Unrestricted fixed-pattern Haar decoupling]
\label{lem:haar-fixed-pattern}
Let
$V:\mathbb{C}^{d_L}\rightarrow
\mathbb{C}^{d_Q}\otimes\mathbb{C}^{d_E}$
be a Haar-random isometry obtained from the first $d_L$ columns of a Haar
unitary, where $d_L\leq d_Qd_E$, and define $D_E(V)$ by
Eq.~\eqref{eq:decoupling-error} using a maximally entangled state of Schmidt
rank $d_L$.  For a fixed erased subsystem,
\begin{equation}
\mathbb{E}_V D_E(V)
\leq
\frac12
\sqrt{
\frac{
d_Ed_Q(d_E^2-1)(d_L^2-1)
}{
d_L[(d_Ed_Q)^2-1]
}
}.
\label{eq:haar-fixed-pattern-bound}
\end{equation}
For one logical qubit encoded without symmetry constraints into $n$ physical
qubits, with a fixed set of $k$ erased qubits,
\begin{equation}
\mathbb{E}_V D_E(V)
\leq
\sqrt{\frac38}\,
2^{-(n-2k)/2}
\sqrt{
\frac{1-2^{-2k}}{1-2^{-2n}}
}.
\label{eq:haar-half-scaling}
\end{equation}
If $k/n\rightarrow\alpha\in(0,1/2)$, the last factor tends to one and
the bound decays as $2^{-(1-2\alpha)n/2}$.
\end{lemma}

\begin{proof}
Put $X=\rho_{RE}-\rho_R\otimes\rho_E$.  Since
$\rho_R=\mathbb{I}_R/d_L$,
\begin{equation}
\operatorname{Tr}(X^2)
=
\operatorname{Tr}(\rho_{RE}^2)
-
\frac{1}{d_L}\operatorname{Tr}(\rho_E^2).
\label{eq:haar-purity-difference}
\end{equation}
The two-copy Haar-isometry identity gives
\begin{align}
\mathbb{E}_V\operatorname{Tr}(\rho_{RE}^2)
&=
\frac{
d_E(d_Q^2-1)+d_Ld_Q(d_E^2-1)
}{
d_L[(d_Ed_Q)^2-1]
},
\label{eq:haar-re-purity}\\
\mathbb{E}_V\operatorname{Tr}(\rho_E^2)
&=
\frac{
d_Ld_E(d_Q^2-1)+d_Q(d_E^2-1)
}{
d_L[(d_Ed_Q)^2-1]
}.
\label{eq:haar-e-purity}
\end{align}
Substitution into Eq.~\eqref{eq:haar-purity-difference} yields
\begin{equation}
\mathbb{E}_V\operatorname{Tr}(X^2)
=
\frac{
d_Q(d_E^2-1)(d_L^2-1)
}{
d_L^2[(d_Ed_Q)^2-1]
}.
\label{eq:haar-hilbert-schmidt}
\end{equation}
Because $X$ acts on a space of dimension $d_Ld_E$,
$\lVert X\rVert_1\leq\sqrt{d_Ld_E}\lVert X\rVert_2$.
Jensen's inequality then proves
Eq.~\eqref{eq:haar-fixed-pattern-bound}.  Setting
$d_L=2$, $d_E=2^k$, and $d_Q=2^{n-k}$ gives
Eq.~\eqref{eq:haar-half-scaling}.
\end{proof}

Lemma~\ref{lem:haar-fixed-pattern} establishes exponential fixed-pattern
decoupling below one half and is consistent with the observed crossover of
the unrestricted Haar benchmark.  It does
\emph{not} establish that one sampled encoder corrects all
$\binom{n}{k}$ sets simultaneously:
\begin{equation}
\sup_E\mathbb{E}_V D_E(V)
\leq
\mathbb{E}_V\max_E D_E(V)
\label{eq:average-versus-uniform-erasure}
\end{equation}
Equality need not hold, so a fixed-pattern expectation cannot be substituted
for a high-probability worst-pattern guarantee.
Nor can the unrestricted result be transferred unchanged to the
protected-coherence
logical qubit under a $U(1)$ constraint.  The
logical states $\left\lvert0\right\rangle{}$ and
$\left\lvert1\right\rangle{}$ transform with different charges, so a covariant
encoder must preserve their relative group action.  An isometry whose entire
range lies in one fixed-total-charge sector instead represents a
charge-neutral logical subsystem and does not by itself establish protection
of the vacuum--one-excitation coherence.  We therefore do not claim an
asymptotic capacity theorem from fixed-sector dimension counting.

\paragraph*{Known covariant-code results and the remaining question.}
Continuous-symmetry versions of the Eastin--Knill obstruction already give
quantitative lower bounds on the erasure-correction accuracy of
finite-dimensional covariant codes~\cite{Faist2020}.  The charge-Haar
construction relevant to the present numerical model is also not new in
principle: random $U(1)$-covariant codes generated by Haar-random
number-conserving unitaries have been analyzed in
Ref.~\cite{Kong2022}.  In the fixed-logical-size, fixed-erasure-size regimes
treated there, both average- and worst-case purified-distance errors vanish as
$O(n^{-1})$.  Asymptotically vanishing approximate error at fixed logical
size is therefore a known covariant-code result, not a conjecture of this
work.  Polynomial convergence to a symmetric $2$-design is likewise known for
the structured CQA ensemble~\cite{Li2024Efficient}; as detailed in
Remark~\ref{rem:finite-depth-scope}, that ensemble is not the canonical
nearest-neighbor local-Haar brickwork measure analyzed here.

The present family encodes one logical occupancy qubit into $n$ physical
memories and hence has rate $1/n$; we make no nonzero-rate claim.
Theorem~\ref{thm:finite-depth-achievability} below proves finite-depth
pattern-average decoupling existence for a canonical local-Haar brickwork
ensemble at fixed erasure size, but only with a non-scalable depth bound.
The architecture-specific open question is whether a specified
geometrically local, finite-depth, number-conserving circuit can obtain
comparable task-specific recovery guarantees at polynomial depth under
independent flagged loss.  Any higher-rate extension
must define the number and charge range of the logical degrees of freedom and
respect the known covariance--accuracy tradeoffs, rather than posit vanishing
full-state recovery error at a fixed positive rate.

Independently of random coding, an unconstrained exact $[[n,1,d]]$ quantum
code corrects $t=d-1$ flagged erasures.  The quantum Singleton bound gives
\begin{equation}
n-1\geq2(d-1)=2t,
\label{eq:singleton}
\end{equation}
and hence
\begin{equation}
t\leq\left\lfloor\frac{n-1}{2}\right\rfloor.
\label{eq:exact-half-limit}
\end{equation}
The near-one-half value is therefore a no-cloning-compatible upper limit for
exact local erasure correction of one logical qubit~\cite{Knill2000}.  It is
an unrestricted, noncovariant exact-code benchmark: it is neither an
achievability statement for the strictly $U(1)$-covariant logical-occupancy
code, which is subject to Lemma~\ref{lem:covariant-no-go}, nor a threshold for
approximate or task-specific covariant recovery.

\section{Finite-depth local scramblers}
\label{sec:finite-depth}

Haar-random encoders are analytical idealizations.  A physically structured
theory must specify a local circuit ensemble.  The unrestricted benchmark may
use generic local gates, whereas the architecture-specific encoder must use
number-conserving gates compatible with Eq.~\eqref{eq:covariant-encoding}.  In
either case write
\begin{equation}
V_j(L)=U_j^{(L)}U_j^{(L-1)}\cdots U_j^{(1)}V_j^{(0)},
\label{eq:brickwork-encoder}
\end{equation}
or a chaotic local Hamiltonian evolution
\begin{equation}
V_j(t)=e^{-itH_j}V_j^{(0)}.
\label{eq:chaotic-encoder}
\end{equation}

\begin{lemma}[Causal-depth obstruction to uniform exact erasure recovery]
\label{lem:causal-depth}
Suppose the logical qubit in $V_j^{(0)}$ is localized at one site of a
one-dimensional chain, all other inputs are fixed ancillas, and
$U_j^{(1)},\ldots,U_j^{(L)}$ form a depth-$L$ nearest-neighbor circuit.  If
one encoder is required to exactly correct every flagged erasure set of size
$k$,
then necessarily
\begin{equation}
2L+1>k.
\label{eq:causal-depth-bound}
\end{equation}
Thus protection against an extensive number $k=\alpha n$ of arbitrary
erasures requires $L=\Omega(n)$ in one dimension.  In a $d$-dimensional
bounded-degree lattice with bounded-range gates and causal-ball volume at
most $CL^d$, the corresponding causal-volume argument gives the necessary
scaling $L=\Omega(k^{1/d})$ for uniform exact recovery.
\end{lemma}

\begin{proof}
After depth $L$, the forward causal cone of the localized logical input
contains at most $2L+1$ sites.  If $k\geq2L+1$, choose an erasure set that
contains this entire cone and add arbitrary sites if necessary.  The retained
state is then independent of the logical input, so no recovery channel on the
retained sites can reconstruct an unknown logical qubit.  The
$d$-dimensional statement follows by replacing the interval by a causal
region of volume $O(L^d)$.
\end{proof}

\begin{remark}[Boundary-localized obstruction]
\label{rem:boundary-causal-depth}
For the present input
$V^{(0)}=(\lvert011\cdots0\rangle,\lvert111\cdots0\rangle)$,
the logical distinction is localized at the boundary site.  Its depth-$L$
forward cone contains at most $\min\{n,L+1\}$ sites, so uniform correction of
every size-$k$ erasure requires $L+1>k$.  More quantitatively, if every site
is erased independently with probability $p$, the event that the complete
forward cone is erased has probability at least
$p^{\min\{n,L+1\}}$.  Conditional on this event the erased subsystem contains
an isometric copy of the logical qubit and
\begin{equation}
D_E(V)=
\frac12
\left\|
\lvert\Phi\rangle\!\langle\Phi\rvert
-\frac{\mathbb I_4}{4}
\right\|_1
=\frac34.
\label{eq:boundary-cone-decoupling}
\end{equation}
Consequently,
\begin{equation}
\mathbb E_{E\sim{\rm Bernoulli}(p)}D_E(V)
\geq
\frac34p^{\min\{n,L+1\}}.
\label{eq:boundary-iid-lower-bound}
\end{equation}
For the current $n=5$, $L=2$, $p=0.30$ work point this gives the
single-node lower bound $0.02025$.  This is a bound on full logical
decoupling, not directly on a task-specific receiver CFI.
\end{remark}

\begin{lemma}[Charge-resolved iid causal lower bound]
\label{lem:charge-causal-iid-lower}
Consider a number-conserving encoding for which the logical dependence is
confined to $m$ physical qubits and whose two logical basis codewords have
total charges differing by one.  Under independent flagged erasure with
probability $p$,
\begin{equation}
\mathbb E_E D_E\geq\frac{p}{2m}.
\label{eq:charge-causal-iid-lower}
\end{equation}
Consequently a boundary-localized input encoded by a depth-$L$
nearest-neighbor number-conserving circuit obeys
\begin{equation}
\mathbb E_E D_E\geq\frac{p}{2(L+1)},
\label{eq:boundary-charge-causal-lower}
\end{equation}
while a bulk input obeys the analogous bound $p/[2(2L+1)]$.
Thus attaining $\mathbb E_ED_E\leq C_pn^{-1/2}$ requires
$L=\Omega_p(\sqrt n)$.  Moreover, if a target threshold satisfies
$\epsilon<1/(2m)$, its failure probability over iid erasure patterns is at
least $p$.
\end{lemma}

\begin{proof}
Let $\rho_i^a$ be the one-site reduced state of logical codeword $a$ and put
\begin{equation}
\delta_i=\operatorname{Tr}[n_i(\rho_i^1-\rho_i^0)].
\end{equation}
Number conservation gives $\sum_{i=1}^m\delta_i=1$, so a fixed site $i_*$,
chosen before the erasure pattern, satisfies
$|\delta_{i_*}|\geq1/m$.  Dephasing the reference of the normalized encoded
Choi state in the logical charge basis and using trace-distance contraction
gives
\begin{equation}
D_i\geq\frac14\|\rho_i^1-\rho_i^0\|_1
\geq\frac12|\delta_i|.
\label{eq:one-site-charge-witness}
\end{equation}
If $E$ contains $i$, tracing out $E\setminus\{i\}$ maps both states in the
decoupling comparison to their one-site counterparts; hence data processing
gives $D_E\geq D_i$.  Averaging the event $i_*\in E$, whose probability is
$p$, proves Eq.~\eqref{eq:charge-causal-iid-lower}.  The forward cones contain
at most $L+1$ boundary sites or $2L+1$ bulk sites.  Finally, on the event
$i_*\in E$ one has $D_E\geq1/(2m)$, which proves the probability statement.
\end{proof}

The preceding obstruction does not exclude approximate pattern-averaged
recovery.  We next prove a finite-depth existence result for a precisely
specified local-Haar ensemble.  This ensemble is an analytical reference and
is not identified with a numerical gate ensemble unless the latter uses the
same probability measure.

\begin{theorem}[Finite-depth local-Haar pattern-average achievability]
\label{thm:finite-depth-achievability}
Let $n\geq3$ be odd, put $q=(n-1)/2$, and let
$V^{(0)}:\mathbb C^2\rightarrow\mathcal H_q\oplus\mathcal H_{q+1}$ obey
$V^{(0)}\lvert0\rangle\in\mathcal H_q$ and
$V^{(0)}\lvert1\rangle\in\mathcal H_{q+1}$.
Let $\mathcal E_{n,2t}^{\rm locH}$ be the depth-$2t$ open-chain brickwork
ensemble with independent gates
\begin{equation}
\begin{aligned}
G&=1\oplus U_1\oplus e^{i\alpha},\\
U_1&\sim{\rm Haar}[U(2)],\\
\alpha&\sim{\rm Unif}[0,2\pi).
\end{aligned}
\label{eq:local-Haar-gate-measure}
\end{equation}
on alternating even and odd nearest-neighbor bonds.

For a fixed erasure set $E$ of size $k\geq1$, define, with binomial
coefficients outside their natural range set to zero,
\begin{align}
e_a&=\binom{k}{a},\\
h_r(a)&=\binom{n-k}{r-a},\\
d_r&=\binom nr,
\label{eq:charge-Haar-dimensions}\\
P_r&=
\frac{1}{d_r(d_r+1)}
\sum_{a=0}^{k}e_ah_r(a)\bigl[h_r(a)+e_a\bigr],
\label{eq:charge-Haar-sector-purity}\\
B_{n,k}&=
\frac{
\sum_{a=0}^{k}e_ah_q(a)h_{q+1}(a)
}{
d_qd_{q+1}
},
\label{eq:charge-Haar-sector-overlap}\\
T_{n,k}&=
\frac{
\sum_{a=0}^{k}e_ae_{a+1}h_q(a)
}{
d_qd_{q+1}
},
\label{eq:charge-Haar-cross-coherence}\\
m_{n,k}^{\rm ch}
&=
\frac{P_q+P_{q+1}-2B_{n,k}}{8}
+\frac{T_{n,k}}{2}.
\label{eq:charge-Haar-decoupling-moment}
\end{align}
Set
\begin{equation}
N_n=16^n,
\qquad
Q_n=6^{3(n-1)/2}.
\label{eq:rational-gap-dimensions}
\end{equation}
Then
\begin{equation}
\mathbb E_{V\sim\mathcal E_{n,2t}^{\rm locH}}D_E(V)
\leq
\beta_{n,k,t},
\label{eq:finite-depth-single-pattern-bound}
\end{equation}
where
\begin{equation}
\beta_{n,k,t}
=
\left\{
2^{k-1}
\left[
m_{n,k}^{\rm ch}
+\sqrt3\,2^n
\exp\!\left(-\frac{t}{2Q_n^{N_n}}\right)
\right]
\right\}^{1/2}.
\label{eq:finite-depth-beta}
\end{equation}
The bound is independent of the location of $E$.  Hence, for every
$\pi(E)\geq0$ supported on $\lvert E\rvert=k$ with
$\sum_{\lvert E\rvert=k}\pi(E)=1$,
\begin{equation}
\mathbb E_V
\left[
\sum_{\lvert E\rvert=k}\pi(E)D_E(V)
\right]
\leq
\beta_{n,k,t}.
\label{eq:finite-depth-pattern-average-bound}
\end{equation}
For independent erasure probability $p$, including every flagged pattern,
\begin{multline}
\mathbb E_V\!\left[
\sum_{E\subseteq[n]}
p^{\lvert E\rvert}(1-p)^{n-\lvert E\rvert}D_E(V)
\right]\\
\leq
\sum_{k=1}^{n}\binom nkp^k(1-p)^{n-k}\beta_{n,k,t}.
\label{eq:finite-depth-iid-bound}
\end{multline}
\end{theorem}

\begin{proof}
Write
\begin{equation}
X_E=\rho_{RE}-\rho_R\otimes\rho_E.
\end{equation}
Since $\rho_R=\mathbb I_R/2$,
\begin{equation}
\operatorname{Tr}(X_E^2)
=
\operatorname{Tr}(\rho_{RE}^2)
-\frac12\operatorname{Tr}(\rho_E^2).
\label{eq:finite-depth-purity-difference}
\end{equation}
The two-copy swap identity rewrites this as
\begin{equation}
\begin{aligned}
\operatorname{Tr}(X_E^2)
&=\operatorname{Tr}\!\left[A_E\rho_{RH}^{\otimes2}\right],\\
A_E&=\left(F_R-\frac12\mathbb I_{R_1R_2}\right)
\otimes F_E\otimes\mathbb I_{Q_1Q_2}.
\end{aligned}
\label{eq:finite-depth-swap-observable}
\end{equation}
Because $F_R$ has trace $2$ on the four-dimensional two-copy reference
space,
\begin{equation}
\left\|F_R-\frac12\mathbb I\right\|_2^2=3,
\qquad
\|A_E\|_2=\sqrt3\,2^n.
\label{eq:finite-depth-swap-norm}
\end{equation}
Thus the ensemble average of $\operatorname{Tr}(X_E^2)$ is a linear
functional of the balanced second-moment channel
\begin{equation}
\mathcal M^{(2)}(Z)
=
\mathbb E_U
\left[
U^{\otimes2}Z(U^\dagger)^{\otimes2}
\right].
\label{eq:balanced-second-moment-channel}
\end{equation}

Let $P_{\rm e}$ and $P_{\rm o}$ be the even- and odd-layer second-moment
twirls.  Each is an orthogonal projector in the Hilbert--Schmidt inner
product, and a complete brickwork cycle has moment operator
$P_{\rm o}P_{\rm e}$.  The common fixed space is the second-moment invariant
space of the group generated by all nearest-neighbor number-conserving
two-site gates.  The gate convention in
Eq.~\eqref{eq:local-Haar-gate-measure} differs from the full two-site
number-conserving group only by an irrelevant global phase.  For moment order
two and $n\geq3$, the generated-group moment equals the Haar moment of the
full number-conserving group~\cite{Hearth2025,Mitsuhashi2025}.  Therefore the
orthogonal projector onto
$\operatorname{Ran}P_{\rm e}\cap\operatorname{Ran}P_{\rm o}$ is the
charge-Haar moment projector $P_{\rm ch}$.

It remains to bound the convergence rate without inserting a numerical
spectral gap.  In the computational Liouville basis, every matrix element of
a two-site second-moment twirl has denominator dividing $6$: the nonzero
$U(2)$ one- and two-copy Weingarten coefficients have denominators $2$, $3$,
and $6$, while the $U(1)$ phase integral contributes only Kronecker
constraints.  Each layer contains $(n-1)/2$ gates.  Consequently every
matrix element of
\begin{equation}
C=P_{\rm e}P_{\rm o}P_{\rm e}
\end{equation}
has common denominator dividing $Q_n=6^{3(n-1)/2}$.
The matrix acts on a Liouville space of dimension $N_n=16^n$, is positive
semidefinite, and has eigenvalue one exactly on the common fixed space.
Let $\mu_*<1$ be its largest nonunit eigenvalue.  After all factors
$(x-1)$ are removed from the characteristic polynomial, its value at $x=1$
is both
\begin{equation}
\prod_{\mu_j<1}(1-\mu_j)
\end{equation}
and a nonzero rational number with denominator at most $Q_n^{N_n}$.
All factors lie in $(0,1]$, hence
\begin{equation}
1-\mu_*\geq Q_n^{-N_n}.
\label{eq:rational-moment-gap}
\end{equation}
The largest nontrivial singular value
$\lambda_n$ of $P_{\rm o}P_{\rm e}$ satisfies
$\lambda_n^2=\mu_*$ and therefore
\begin{equation}
1-\lambda_n
\geq\frac12Q_n^{-N_n}.
\label{eq:rational-singular-gap}
\end{equation}
Since $P_{\rm ch}$ commutes with both layer projectors,
\begin{equation}
\left\|
(P_{\rm o}P_{\rm e})^t-P_{\rm ch}
\right\|_{2\rightarrow2}
\leq
\lambda_n^t
\leq
\exp\!\left(-\frac{t}{2Q_n^{N_n}}\right).
\label{eq:finite-depth-moment-convergence}
\end{equation}
Equations~\eqref{eq:finite-depth-swap-norm} and
\eqref{eq:finite-depth-moment-convergence}, applied to the normalized
two-copy input state, give
\begin{equation}
\left|
\mathbb E_{\mathcal E_{n,2t}^{\rm locH}}
\operatorname{Tr}(X_E^2)
-
\mathbb E_{\rm ch}\operatorname{Tr}(X_E^2)
\right|
\leq
\sqrt3\,2^n
\exp\!\left(-\frac{t}{2Q_n^{N_n}}\right).
\label{eq:finite-depth-purity-convergence}
\end{equation}

For completeness, decompose a charge-$r$ state according to the charge $a$
in $E$:
\begin{equation}
\mathcal H_r
=
\bigoplus_a
\mathcal H_{E,a}\otimes\mathcal H_{Q,r-a}.
\label{eq:charge-sector-decomposition}
\end{equation}
The two-copy Haar identity in $\mathcal H_r$ gives
\begin{equation}
\mathbb E_{\rm ch}\operatorname{Tr}(\rho_{E,r}^2)=P_r.
\end{equation}
Independence of the two global charge blocks gives
\begin{align}
\mathbb E_{\rm ch}
\operatorname{Tr}(\rho_{E,q}\rho_{E,q+1})
&=B_{n,k},\\
\mathbb E_{\rm ch}
\left\|
\operatorname{Tr}_Q
\lvert\psi_q\rangle\!\langle\psi_{q+1}\rvert
\right\|_2^2
&=T_{n,k}.
\end{align}
The second equality uses
$h_{q+1}(a+1)=h_q(a)$ and is the cross-charge coherence contribution; it
cannot be replaced by an independent single-sector purity.
Writing $X_E$ in its two logical-reference blocks now gives
\begin{equation}
\mathbb E_{\rm ch}\operatorname{Tr}(X_E^2)
=
\frac{P_q+P_{q+1}-2B_{n,k}}8+\frac{T_{n,k}}2
=m_{n,k}^{\rm ch}.
\label{eq:charge-Haar-moment-proof}
\end{equation}

Finally, $X_E$ acts on a space of dimension $2^{k+1}$, so
\begin{equation}
D_E(V)
=\frac12\|X_E\|_1
\leq
\sqrt{2^{k-1}\operatorname{Tr}(X_E^2)}.
\label{eq:finite-depth-trace-HS}
\end{equation}
Jensen's inequality together with
Eqs.~\eqref{eq:finite-depth-purity-convergence} and
\eqref{eq:charge-Haar-moment-proof} proves
Eqs.~\eqref{eq:finite-depth-single-pattern-bound} and
\eqref{eq:finite-depth-pattern-average-bound}.  Summing over all patterns
and grouping them by cardinality proves
Eq.~\eqref{eq:finite-depth-iid-bound}.
\end{proof}

\begin{corollary}[Analytical finite-depth existence at fixed erasure size]
\label{cor:finite-depth-fixed-k}
For every fixed $k$,
\begin{equation}
2^{k-1}m_{n,k}^{\rm ch}
=
\frac{k}{4n^2}+O_k(n^{-3}).
\label{eq:charge-Haar-fixed-k-asymptotic}
\end{equation}
Choose
\begin{equation}
t_n=
\left\lceil
2Q_n^{N_n}
\left[
(n+k-1)\log2+2\log n+\frac12\log3
\right]
\right\rceil.
\label{eq:explicit-finite-depth-choice}
\end{equation}
Then
\begin{equation}
\mathbb E_{V\sim\mathcal E_{n,2t_n}^{\rm locH}}
\left[
\sum_{\lvert E\rvert=k}\pi(E)D_E(V)
\right]
=O_k(n^{-1})
\label{eq:finite-depth-fixed-k-achievability}
\end{equation}
for every normalized pattern distribution $\pi$.  In particular, at least
one finite-depth local number-conserving brickwork encoder attains this bound.
\end{corollary}

\begin{proof}
For fixed $k$, expansion of the finite binomial ratios in
Eqs.~\eqref{eq:charge-Haar-dimensions}--\eqref{eq:charge-Haar-decoupling-moment}
gives Eq.~\eqref{eq:charge-Haar-fixed-k-asymptotic}.  The choice
Eq.~\eqref{eq:explicit-finite-depth-choice} ensures
\begin{equation}
\sqrt3\,2^{n+k-1}
\exp\!\left(-\frac{t_n}{2Q_n^{N_n}}\right)
\leq n^{-2}.
\end{equation}
Substitution into Eq.~\eqref{eq:finite-depth-beta} proves the claimed
$O_k(n^{-1})$ expectation.  Existence of one realization follows because
the minimum of a nonnegative random variable does not exceed its mean.
\end{proof}

The fixed-$k$ expansion does not determine the physically different regime
in which a nonzero fraction of the memories is erased.  In that regime one
can avoid the dimension-loss in Eq.~\eqref{eq:finite-depth-trace-HS} and
calculate the charge-Haar trace distance itself.

\begin{theorem}[Extensive-loss charge-Haar trace-decoupling law]
\label{thm:charge-haar-extensive}
Let $n$ be odd, $q=(n-1)/2$, $d=\binom nq$, and let
$\lvert\psi_q\rangle$ and $\lvert\psi_{q+1}\rangle$ be independent Haar
unit vectors in $\mathcal H_q$ and $\mathcal H_{q+1}$, respectively.  For a
fixed set $E$ of $k$ erased qubits define $D_E$ as in
Eq.~\eqref{eq:decoupling-error} and put
\begin{equation}
\overline D_{n,k}
=
\frac{1}{4d}
\sum_{a=0}^{k}
\binom ka
\left|
\binom{n-k}{q-a}-\binom{n-k}{q+1-a}
\right|.
\label{eq:charge-haar-mean-distance}
\end{equation}
Then the following finite-size estimate holds:
\begin{equation}
0\leq
\mathbb E_{\rm ch}D_E-\overline D_{n,k}
\leq
\frac{2^k}{\sqrt d}
\leq
\sqrt{n+1}\,2^{k-n/2}.
\label{eq:charge-haar-fluctuation-bound}
\end{equation}
Consequently, if $k/n\to\alpha$ for a constant
$0<\alpha<1/2$, then
\begin{equation}
\mathbb E_{\rm ch}D_E
=
\sqrt{\frac{\alpha}{2\pi(1-\alpha)}}\,n^{-1/2}
+o(n^{-1/2}).
\label{eq:charge-haar-extensive-distance}
\end{equation}
The convergence is uniform when $\alpha$ ranges in a compact subset of
$(0,1/2)$.  In the same regime the exact balanced second moment in
Eq.~\eqref{eq:charge-Haar-decoupling-moment} obeys
\begin{equation}
2^{k-1}m_{n,k}^{\rm ch}
=
\frac{\alpha}{4(1-\alpha^2)^{3/2}}\,n^{-1}
+o(n^{-1}).
\label{eq:charge-haar-extensive-second-moment}
\end{equation}
Thus the Hilbert--Schmidt conversion has the correct $n^{-1/2}$ order but
does not give the sharp leading constant in trace distance.
\end{theorem}

\begin{proof}
Let
\begin{equation}
\tau_r=\mathbb E_{\rm ch}\rho_{E,r}
=\bigoplus_a \frac{h_r(a)}{d}\,\mathbb I_{E,a},
\qquad r\in\{q,q+1\},
\label{eq:charge-haar-average-marginal}
\end{equation}
and set
$\sigma_E=\operatorname{Tr}_Q
\lvert\psi_q\rangle\!\langle\psi_{q+1}\rvert$.
In the logical-reference basis the decoupling operator is exactly
\begin{equation}
X_E=
\begin{pmatrix}
(\rho_{E,q}-\rho_{E,q+1})/4 & \sigma_E/2\\
\sigma_E^\dagger/2 & (\rho_{E,q+1}-\rho_{E,q})/4
\end{pmatrix}.
\label{eq:charge-haar-block-X}
\end{equation}
Its ensemble mean is obtained by replacing $\rho_{E,r}$ by $\tau_r$ and
$\sigma_E$ by zero.  Hence
\begin{equation}
\frac12\|\mathbb E_{\rm ch}X_E\|_1
=\frac14\|\tau_q-\tau_{q+1}\|_1
=\overline D_{n,k}.
\label{eq:charge-haar-mean-X-distance}
\end{equation}

To control the random part, write
\begin{align}
A_r&=\sum_a e_ah_r(a)^2,
&
C_r&=\sum_a e_a^2h_r(a).
\end{align}
The two-copy Haar identity and
Eq.~\eqref{eq:charge-Haar-sector-purity} give
\begin{equation}
\mathbb E_{\rm ch}\|\rho_{E,r}-\tau_r\|_2^2
=\frac{dC_r-A_r}{d^2(d+1)}
\leq\frac{\max_a e_a}{d+1}.
\label{eq:charge-haar-marginal-variance}
\end{equation}
Since the erased Hilbert space has dimension $2^k$, Cauchy--Schwarz and
Jensen imply
\begin{equation}
\mathbb E_{\rm ch}\|\rho_{E,r}-\tau_r\|_1
\leq\frac{2^k}{\sqrt{d+1}}.
\label{eq:charge-haar-marginal-concentration}
\end{equation}
Likewise, Eq.~\eqref{eq:charge-Haar-cross-coherence} and Vandermonde's
identity give
\begin{equation}
\mathbb E_{\rm ch}\|\sigma_E\|_1
\leq\sqrt{2^kT_{n,k}}
\leq\frac{2^k}{\sqrt d},
\label{eq:charge-haar-coherence-concentration}
\end{equation}
because
$\sum_a e_ae_{a+1}h_q(a)\leq(\max_a e_{a+1})d$.
Convexity gives $\mathbb E_{\rm ch}D_E\geq\frac12\|\mathbb E_{\rm ch}X_E\|_1=\overline D_{n,k}$, while the triangle inequality applied to Eq.~\eqref{eq:charge-haar-block-X} gives the upper bound in Eq.~\eqref{eq:charge-haar-fluctuation-bound}.  The second follows from the
elementary central-binomial bound $d\geq2^n/(n+1)$.

It remains to evaluate Eq.~\eqref{eq:charge-haar-mean-distance}.  If
$A$ has the hypergeometric law
\begin{equation}
\Pr(A=a)=\frac{e_ah_q(a)}{d},
\label{eq:charge-haar-hypergeometric-law}
\end{equation}
then the exact binomial ratio
\begin{equation}
h_q(a)-h_{q+1}(a)
=h_q(a)\frac{k-2a}{q-a+1}
\label{eq:adjacent-charge-binomial-ratio}
\end{equation}
shows that
\begin{equation}
\overline D_{n,k}
=\frac14\mathbb E
\left[
\frac{|k-2A|}{q-A+1}
\right].
\label{eq:charge-haar-hypergeometric-distance}
\end{equation}
Here
\begin{equation}
\begin{aligned}
\mu&=\mathbb EA=\frac{k(n-1)}{2n},\\
\varsigma^2&=\operatorname{Var}(A)
=\frac{k(n+1)(n-k)}{4n^2}.
\end{aligned}
\label{eq:charge-haar-hypergeometric-moments}
\end{equation}
For $k/n$ in a compact subset of $(0,1/2)$, the denominator in
Eq.~\eqref{eq:charge-haar-hypergeometric-distance} is uniformly of order
$n$.  Replacing it by $(n-k)/2$ changes the expectation by $O(n^{-1})$;
this follows directly from
\begin{equation}
\left|n-k-2(q-A+1)\right|=|2A-k-1|
\end{equation}
and Eq.~\eqref{eq:charge-haar-hypergeometric-moments}.  The uniform
hypergeometric central-limit theorem, together with uniform integrability
from the bounded second moments, yields
\begin{equation}
\mathbb E|A-\mu|
=\varsigma\left(\sqrt{\frac2\pi}+o(1)\right).
\end{equation}
Moreover,
$|\mathbb E|k-2A|-2\mathbb E|A-\mu||\leq k/n$.
Substitution into
Eq.~\eqref{eq:charge-haar-hypergeometric-distance} proves
Eq.~\eqref{eq:charge-haar-extensive-distance}.  Notice that this step uses
only a classical limit theorem; no numerical ansatz is involved.

For the second-moment statement, an exact rearrangement of
Eqs.~\eqref{eq:charge-Haar-sector-purity}--
\eqref{eq:charge-Haar-decoupling-moment} gives
\begin{equation}
2^{k-1}m_{n,k}^{\rm ch}
=
\frac{2^{k-4}}{d^2}
\sum_a e_a\bigl[h_q(a)-h_{q+1}(a)\bigr]^2+R_{n,k},
\label{eq:charge-haar-second-moment-rearrangement}
\end{equation}
where, using only Vandermonde's identity and
$\max_a e_a\leq2^k$, $\max_a h_r(a)\leq2^{n-k}$,
\begin{equation}
|R_{n,k}|
\leq
\frac{2^{n-3}}{d^2}
+\frac{3\,2^{2k-3}}{d}.
\label{eq:charge-haar-second-moment-remainder}
\end{equation}
This remainder is exponentially small when
$k/n\leq1/2-\eta$.  For completeness, the leading sum is a one-dimensional
lattice Laplace problem.  Put $M=n-k$ and $a=k/2+x$.  Uniform Stirling
expansions in the central window give
\begin{align}
e_a&=2^k\sqrt{\frac{2}{\pi k}}
\exp(-2x^2/k)[1+o(1)],\\
h_q(a)&=2^M\sqrt{\frac{2}{\pi M}}
\exp[-2(x+1/2)^2/M][1+o(1)],
\end{align}
while Eq.~\eqref{eq:adjacent-charge-binomial-ratio} contributes
$16x^2/M^2[1+o(1)]$.  Outside a
$\sqrt n\log n$ window the binomial Chernoff bound is superpolynomially
small, so the expansion is uniform and summable.  The remaining Gaussian sums use
\begin{equation}
\frac2k+\frac4M=\frac{2(n+k)}{kM},
\end{equation}
and yield
\begin{equation}
\frac{2^k}{d^2}
\sum_a e_a\bigl[h_q(a)-h_{q+1}(a)\bigr]^2
=
\frac{4\alpha}{(1-\alpha^2)^{3/2}}\,n^{-1}
+o(n^{-1}).
\label{eq:charge-haar-laplace-evaluation}
\end{equation}
Equations~\eqref{eq:charge-haar-second-moment-rearrangement}--
\eqref{eq:charge-haar-laplace-evaluation} prove
Eq.~\eqref{eq:charge-haar-extensive-second-moment}.
\end{proof}

\begin{theorem}[Universal erased-charge floor and charge-Haar optimality]
\label{thm:universal-charge-count-optimality}
Let \(V\) be any isometric encoding with logical codewords supported in
\(\mathcal H_q\) and \(\mathcal H_{q+1}\), where \(n=2q+1\), and put \(d=\binom nq\).  Define the
uniform fixed-cardinality and iid-erasure errors
\begin{align}
\mathcal D_{n,k}(V)
&=\binom nk^{-1}\sum_{|E|=k}D_E(V),
\label{eq:arbitrary-code-fixed-k-error}\\
\mathcal D_{n,p}(V)
&=\sum_{E\subseteq[n]}p^{|E|}(1-p)^{n-|E|}D_E(V).
\label{eq:arbitrary-code-iid-error}
\end{align}
Then, for every \(V\), every \(k\), and every \(p\in[0,1]\),
\begin{align}
\boxed{\mathcal D_{n,k}(V)\geq\overline D_{n,k},}
\label{eq:universal-fixed-k-charge-floor}\\
\boxed{\mathcal D_{n,p}(V)\geq M_{\rm eq}(n,p).}
\label{eq:universal-iid-charge-floor}
\end{align}
These bounds use only the total occupation of the erased subsystem and are
therefore independent of all amplitudes and phases inside the two charge
sectors.

Let
\begin{equation}
\mathcal D_{n,k}^{\rm opt}=\inf_V\mathcal D_{n,k}(V),
\qquad
\mathcal D_{n,p}^{\rm opt}=\inf_V\mathcal D_{n,p}(V),
\label{eq:adjacent-charge-optimal-errors}
\end{equation}
where the infima range over all such adjacent-charge isometries.  Then
\begin{equation}
\boxed{
\overline D_{n,k}
\leq\mathcal D_{n,k}^{\rm opt}
\leq\overline D_{n,k}+\frac{2^k}{\sqrt d}.}
\label{eq:fixed-k-charge-haar-optimality}
\end{equation}
For fixed \(0<p<\beta<1/2\), define
\begin{equation}
I(\beta\|p)
=\beta\log\frac{\beta}{p}
+(1-\beta)\log\frac{1-\beta}{1-p}.
\label{eq:optimality-binomial-rate}
\end{equation}
Then
\begin{align}
M_{\rm eq}(n,p)
\leq\mathcal D_{n,p}^{\rm opt}
\leq{}& M_{\rm eq}(n,p)
+\sqrt{n+1}\,2^{-(1/2-\beta)n}
\nonumber\\
&+e^{-nI(\beta\|p)}.
\label{eq:iid-charge-haar-optimality}
\end{align}
Consequently, if \(k/n\to\alpha\in(0,1/2)\),
\begin{equation}
\mathcal D_{n,k}^{\rm opt}
=\sqrt{\frac{\alpha}{2\pi(1-\alpha)}}\,n^{-1/2}
+o(n^{-1/2}),
\label{eq:fixed-k-optimal-asymptotic}
\end{equation}
and, for every fixed \(p<1/2\),
\begin{equation}
\boxed{
\mathcal D_{n,p}^{\rm opt}
=\sqrt{\frac{p}{2\pi(1-p)}}\,n^{-1/2}
+o(n^{-1/2}).}
\label{eq:iid-optimal-asymptotic}
\end{equation}
Thus charge-Haar codes are exponentially close to the
information-theoretic optimum below one-half erasure within the strict
adjacent-charge model.
\end{theorem}

\begin{proof}
For fixed \(k\), retain the uniformly sampled erasure set as an orthogonal
classical flag and form the direct-sum decoupling operator
\begin{equation}
\begin{aligned}
\mathfrak X_k(V)&=\binom nk^{-1}\bigoplus_{|E|=k}X_E(V),\\
\mathcal D_{n,k}(V)&=\frac12\|\mathfrak X_k(V)\|_1.
\end{aligned}
\end{equation}
Apply the following completely positive trace-preserving post-processing:
read the flag \(E\), measure the logical reference in the computational
basis, measure the erased charge
\(N_E=\sum_{i\in E}n_i\), and then discard the detailed pattern \(E\).
For any normalized state supported in total charge \(r\), the resulting
charge law after averaging uniformly over \(|E|=k\) is
\begin{equation}
P_r^{(k)}(a)
=\frac{\binom ka\binom{n-k}{r-a}}{\binom nr}.
\label{eq:universal-hypergeometric-charge-law}
\end{equation}
Indeed, expanding the state in the computational basis removes all
coherences under the charge measurement, and for every fixed weight-\(r\)
string the fraction of size-\(k\) subsets containing exactly \(a\) occupied
sites is the same hypergeometric probability.  Hence it is independent of
the encoded amplitudes.

After the reference measurement, the actual classical joint law is
\(P(R=j,N_E=a)=\frac12P_{q+j}^{(k)}(a)\), while the product of its reference
and erased-charge marginals assigns
\(\frac14[P_q^{(k)}(a)+P_{q+1}^{(k)}(a)]\) to each value of \(R\).  Classical
data processing therefore gives
\begin{equation}
\mathcal D_{n,k}(V)
\geq\frac14\sum_a|P_q^{(k)}(a)-P_{q+1}^{(k)}(a)|
=\overline D_{n,k},
\end{equation}
where \(\binom nq=\binom n{q+1}=d\) was used in the last equality.  For iid
erasures, retain \(K=|E|\) in addition to \(N_E\), condition on \(K=k\),
and average the preceding bound with the binomial weights.  Equation
\eqref{eq:equilibrium-flagged-identity} then gives
Eq.~\eqref{eq:universal-iid-charge-floor}.

The lower bounds on the two infima follow immediately.  For the upper bound
at fixed \(k\), permutation invariance of the charge-Haar ensemble implies
\begin{equation}
\mathbb E_{\rm ch}\mathcal D_{n,k}(V)
=\mathbb E_{\rm ch}D_{E_0}
\leq\overline D_{n,k}+\frac{2^k}{\sqrt d}
\end{equation}
for any fixed \(|E_0|=k\), by
Theorem~\ref{thm:charge-haar-extensive}.  At least one realization cannot
exceed this expectation.  The iid upper bound follows in the same way from
the one-sided version of
Eq.~\eqref{eq:charge-haar-classical-center-equivalence}.  Finally,
Eqs.~\eqref{eq:charge-haar-extensive-distance} and
\eqref{eq:equilibrium-flagged-asymptotic}, together with the exponentially
small upper remainders for \(\alpha,p<1/2\), prove
Eqs.~\eqref{eq:fixed-k-optimal-asymptotic} and
\eqref{eq:iid-optimal-asymptotic}.
\end{proof}

\begin{corollary}[Fixed logical-charge separation]
\label{cor:fixed-charge-separation}
Fix an integer $s\geq1$, take $n\equiv s\pmod2$, and put
\begin{equation}
r_-=\frac{n-s}{2},\qquad r_+=\frac{n+s}{2},\qquad
d=\binom n{r_-}=\binom n{r_+}.
\end{equation}
Let the two logical codewords be independent Haar vectors in
$\mathcal H_{r_-}$ and $\mathcal H_{r_+}$.  For a fixed erased set of size
$k$, define
\begin{equation}
\overline D^{(s)}_{n,k}
=
\frac1{4d}\sum_a\binom ka
\left|
\binom{n-k}{r_--a}-\binom{n-k}{r_+-a}
\right|.
\label{eq:fixed-separation-center}
\end{equation}
Then
\begin{equation}
\left|\mathbb E_{\rm ch}D_E-\overline D^{(s)}_{n,k}\right|
\leq\frac{2^k}{\sqrt d}.
\label{eq:fixed-separation-fluctuation}
\end{equation}
If $k/n\to\alpha\in(0,1/2)$, then
\begin{equation}
\mathbb E_{\rm ch}D_E
=
s\sqrt{\frac{\alpha}{2\pi(1-\alpha)}}\,n^{-1/2}
+o(n^{-1/2}).
\label{eq:fixed-separation-law}
\end{equation}
The remainder is uniform for $\alpha$ in compact subsets of $(0,1/2)$.
Thus a fixed larger logical charge gap changes the leading constant
linearly, but not the $n^{-1/2}$ exponent.
\end{corollary}

\begin{proof}
The mean-state calculation in
Eq.~\eqref{eq:charge-haar-mean-X-distance} gives
Eq.~\eqref{eq:fixed-separation-center}.  Tracing the retained subsystem in
the cross term forces erased charges $a$ and $a+s$, and hence
\begin{equation}
\mathbb E_{\rm ch}\|\sigma_E\|_2^2
=\frac1{d^2}\sum_a
\binom ka\binom{k}{a+s}\binom{n-k}{r_--a}
\leq\frac{2^k}{d}.
\label{eq:fixed-separation-coherence}
\end{equation}
The two marginal-variance estimates used in
Eq.~\eqref{eq:charge-haar-marginal-variance} are unchanged.  Applying
trace/HS conversion to the two marginal fluctuations and
Eq.~\eqref{eq:fixed-separation-coherence} proves
Eq.~\eqref{eq:fixed-separation-fluctuation}.

It remains to evaluate the deterministic center.  Let $P_\pm$ be the
hypergeometric laws
\begin{equation}
P_\pm(a)=\frac{\binom ka\binom{n-k}{r_\pm-a}}d,
\qquad
\overline D^{(s)}_{n,k}=\frac14\sum_a|P_-(a)-P_+(a)|.
\end{equation}
Under $P_-$, write $A=\mu_-+x$.  Its mean and variance are
\begin{align}
\mu_-&=\frac{k(n-s)}{2n},\\
\sigma_n^2
&=\frac{k(n-k)}{4(n-1)}
\left(1-\frac{s^2}{n^2}\right).
\label{eq:fixed-separation-hypergeom-moments}
\end{align}
The exact likelihood ratio is
\begin{equation}
\frac{P_+(a)}{P_-(a)}
=
\prod_{j=1}^s
\frac{n-k-(r_--a)-j+1}{r_--a+j}.
\label{eq:fixed-separation-likelihood}
\end{equation}
On $|x|\leq\sqrt n\log n$, Taylor's formula for the logarithm in
Eq.~\eqref{eq:fixed-separation-likelihood}, with its third derivative
bounded uniformly on compact $\alpha$ intervals, gives
\begin{equation}
\begin{aligned}
\log\frac{P_+(A)}{P_-(A)}
&=\frac{4s}{n-k}x-\frac{2s^2k}{n(n-k)}\\
&\quad+O_s\!\left(
\frac{1+|x|}{n^2}+\frac{1+|x|^3}{n^3}
\right).
\end{aligned}
\label{eq:fixed-separation-lan}
\end{equation}
The sampling-without-replacement Hoeffding bound controls the complementary
event.  Exponentiating Eq.~\eqref{eq:fixed-separation-lan}, using
$\mathbb E x^2=O(n)$ and
Eq.~\eqref{eq:fixed-separation-hypergeom-moments}, yields the $L^1(P_-)$
expansion
\begin{equation}
\begin{aligned}
\frac{P_+(A)}{P_-(A)}-1
&=\frac{ks/n}{\sigma_n^2}(A-\mu_-)+\mathcal R_n,\\
\mathbb E_-|\mathcal R_n|&=O_s(n^{-1}).
\end{aligned}
\label{eq:fixed-separation-score}
\end{equation}
The hypergeometric central-limit theorem and uniform integrability imply
$\mathbb E|A-\mu_-|=\sigma_n[\sqrt{2/\pi}+o(1)]$.  Therefore
\begin{equation}
\sum_a|P_+(a)-P_-(a)|
=
\frac{ks/n}{\sigma_n}\sqrt{\frac2\pi}+o(n^{-1/2}).
\end{equation}
Substitution of
$\sigma_n\sim\frac12\sqrt{\alpha(1-\alpha)n}$ proves
Eq.~\eqref{eq:fixed-separation-law}.
\end{proof}

\begin{corollary}[Independent extensive erasures below one half]
\label{cor:charge-haar-iid-extensive}
For a fixed $0<p<1/2$, average over every flagged erasure pattern with
independent probability $p$ per site.  Then
\begin{equation}
\begin{aligned}
\mathcal D_{n,p}^{\rm ch}
&:=\mathbb E_{\rm ch}\sum_{E\subseteq[n]}
 p^{|E|}(1-p)^{n-|E|}D_E\\
&=\sqrt{\frac{p}{2\pi(1-p)}}\,n^{-1/2}
+o(n^{-1/2}).
\end{aligned}
\label{eq:charge-haar-iid-extensive-distance}
\end{equation}
\end{corollary}

\begin{proof}
Let $K\sim\operatorname{Bin}(n,p)$.  The charge-Haar expectation depends on
an erasure set only through $K$.  Choose constants
$0<a<p<b<1/2$.  On $an\leq K\leq bn$, the proof of
Theorem~\ref{thm:charge-haar-extensive} is uniform, and its leading
coefficient is a continuous function of $K/n$.  Moreover,
Eq.~\eqref{eq:charge-haar-hypergeometric-distance}, Cauchy--Schwarz, and
Eq.~\eqref{eq:charge-haar-fluctuation-bound} give a uniform $O(n^{-1/2})$
bound there.  Since $K/n\to p$ in probability, dominated convergence on
this event gives the right-hand side of
Eq.~\eqref{eq:charge-haar-iid-extensive-distance}.  The two complementary
binomial tails have exponentially small probability by the Chernoff bound,
and $0\leq D_E\leq1$, so they are negligible.
\end{proof}

\begin{corollary}[Exponential classicalization of the charge-Haar endpoint]
\label{cor:charge-haar-exponential-classicalization}
For a fixed erased set \(E\) of size \(k\), let
\(\overline X_E=\mathbb E_{\rm ch}X_E\) and define the centered
charge-Haar fluctuation
\begin{equation}
\gamma^{\rm ch}_{n,k}
=\frac12\mathbb E_{\rm ch}
\|X_E-\overline X_E\|_1.
\label{eq:charge-haar-fixed-pattern-centered-trace}
\end{equation}
Then
\begin{equation}
\boxed{
\gamma^{\rm ch}_{n,k}\leq\frac{2^k}{\sqrt d},
\qquad
2^{k-1}\mathbb E_{\rm ch}
\|X_E-\overline X_E\|_2^2
\leq\frac{2^{2k-1}}d,}
\label{eq:charge-haar-centered-fixed-pattern-bounds}
\end{equation}
where \(d=\binom n{(n-1)/2}\).

Let \(E\) be iid Bernoulli-\(p\), with fixed \(0<p<1/2\), and choose any
\(p<\beta<1/2\).  Let \(I(\beta\|p)\) be the binomial rate function in
Eq.~\eqref{eq:optimality-binomial-rate}.
The full centered trace fluctuation satisfies
\begin{equation}
\begin{aligned}
\Gamma^{\rm ch}_{n,p}
&:=\frac12\mathbb E_E\mathbb E_{\rm ch}
\|X_E-\overline X_E\|_1\\
&\leq\sqrt{n+1}\,2^{-(1/2-\beta)n}
+2e^{-nI(\beta\|p)}.
\end{aligned}
\label{eq:charge-haar-iid-centered-trace}
\end{equation}
Moreover, the typical-set rank-weighted covariance
\begin{equation}
\mathcal F^{\rm ch}_{n,p,\beta}(\infty)
:=\sum_{|E|\leq\beta n}\pi_p(E)2^{|E|-1}
\mathbb E_{\rm ch}\|X_E-\overline X_E\|_2^2
\end{equation}
obeys
\begin{equation}
\boxed{
\mathcal F^{\rm ch}_{n,p,\beta}(\infty)
\leq
\frac{n+1}{2}\,2^{-(1-2\beta)n}.}
\label{eq:charge-haar-stationary-centered-covariance}
\end{equation}
Finally,
\begin{equation}
0\leq
\mathcal D_{n,p}^{\rm ch}-M_{\rm eq}(n,p)
\leq
\sqrt{n+1}\,2^{-(1/2-\beta)n}
+e^{-nI(\beta\|p)}.
\label{eq:charge-haar-classical-center-equivalence}
\end{equation}
Thus, below one-half erasure, the charge-Haar trace-distance endpoint is
classical up to \(e^{-\Omega(n)}\), and the \(O(n^{-1})\) typical-set
stationary rank-weighted second moment comes from the deterministic mean
rather than a centered circuit fluctuation.
\end{corollary}

\begin{proof}
For fixed \(E\), put
\(\delta_r=\rho_{E,r}-\tau_r\), \(r\in\{q,q+1\}\).  From
Eq.~\eqref{eq:charge-haar-block-X},
\begin{equation}
X_E-\overline X_E
=
\begin{pmatrix}
(\delta_q-\delta_{q+1})/4&\sigma_E/2\\
\sigma_E^\dagger/2&(\delta_{q+1}-\delta_q)/4
\end{pmatrix}.
\end{equation}
The diagonal and off-diagonal block matrices have trace norms
\(\frac12\|\delta_q-\delta_{q+1}\|_1\) and
\(\|\sigma_E\|_1\), respectively.  Therefore
\begin{align}
\gamma^{\rm ch}_{n,k}
&\leq
\frac14\sum_{r=q}^{q+1}
\mathbb E_{\rm ch}\|\delta_r\|_1
+\frac12\mathbb E_{\rm ch}\|\sigma_E\|_1\\
&\leq
\frac{2^{k-1}}{\sqrt{d+1}}
+\frac{2^{k-1}}{\sqrt d}
\leq\frac{2^k}{\sqrt d},
\end{align}
by Eqs.~\eqref{eq:charge-haar-marginal-concentration} and
\eqref{eq:charge-haar-coherence-concentration}.

For the Hilbert--Schmidt bound, orthogonality of the logical-reference
blocks gives
\begin{equation}
\|X_E-\overline X_E\|_2^2
=\frac18\|\delta_q-\delta_{q+1}\|_2^2
+\frac12\|\sigma_E\|_2^2.
\end{equation}
Equations~\eqref{eq:charge-haar-marginal-variance} and
\eqref{eq:charge-Haar-cross-coherence}, together with
\(\max_a e_a\leq2^k\), imply
\begin{equation}
\mathbb E_{\rm ch}\|X_E-\overline X_E\|_2^2
\leq
\frac{2^{k-1}}{d+1}+\frac{2^{k-1}}d
\leq\frac{2^k}{d}.
\end{equation}
Multiplication by \(2^{k-1}\) proves
Eq.~\eqref{eq:charge-haar-centered-fixed-pattern-bounds}.

Let \(K=|E|\sim{\rm Bin}(n,p)\).  On \(K\leq\beta n\), the central-binomial
bound \(d\geq2^n/(n+1)\) gives
\begin{equation}
\gamma^{\rm ch}_{n,K}
\leq\sqrt{n+1}\,2^{-(1/2-\beta)n}.
\end{equation}
On the complementary event we use the universal bound
\(\frac12\|X_E-\overline X_E\|_1\leq2\).  The binomial Chernoff bound
\(\Pr\{K>\beta n\}\leq e^{-nI(\beta\|p)}\) proves
Eq.~\eqref{eq:charge-haar-iid-centered-trace}.
The second inequality in
Eq.~\eqref{eq:charge-haar-centered-fixed-pattern-bounds}, averaged only over
\(K\leq\beta n\), gives
Eq.~\eqref{eq:charge-haar-stationary-centered-covariance}.

Finally, convexity and the triangle inequality give, for every fixed
pattern,
\begin{equation}
0\leq
\mathbb E_{\rm ch}D_E
-\frac12\|\overline X_E\|_1
\leq\gamma^{\rm ch}_{n,k}.
\end{equation}
The Bernoulli average of \(\frac12\|\overline X_E\|_1\) is
\(M_{\rm eq}(n,p)\) by
Theorem~\ref{thm:equilibrium-flagged-tv}; on the upper binomial tail the
absolute difference is at most one.  The same split proves
Eq.~\eqref{eq:charge-haar-classical-center-equivalence}.  Since
\(\mathbb E\operatorname{Tr}(X_E^2)
=\operatorname{Tr}(\overline X_E^2)
+\mathbb E\|X_E-\overline X_E\|_2^2\),
Eq.~\eqref{eq:charge-haar-stationary-centered-covariance} also proves the
last statement.
\end{proof}

\begin{corollary}[Charge-Haar extensive-erasure transition]
\label{cor:charge-haar-erasure-threshold}
For a deterministic sequence $k/n\to\alpha>1/2$,
\begin{equation}
\mathbb E_{\rm ch}D_E=\frac34-O_\alpha(n^{-1/4}).
\label{eq:charge-haar-above-half}
\end{equation}
For Bernoulli erasures at any fixed $p\in(0,1)$, the complete limiting law is
\begin{equation}
\lim_{n\to\infty}\mathcal D_{n,p}^{\rm ch}
=
\begin{cases}
0, & 0<p<1/2,\\[2pt]
3/8, & p=1/2,\\[2pt]
3/4, & 1/2<p<1.
\end{cases}
\label{eq:charge-haar-threshold-law}
\end{equation}
Below one half, the sharper leading term is
Eq.~\eqref{eq:charge-haar-iid-extensive-distance}.
\end{corollary}

\begin{proof}
We first record two general facts for a pure qubit encoding into $EQ$.  Its
state $\rho_{RE}$ is the Choi state of a channel $\mathcal N:L\to E$, and
\begin{equation}
\rho_{RE}-\rho_R\otimes\rho_E
=
(\operatorname{id}_R\otimes\mathcal N)
\left(\Phi_{RL}-\frac{\mathbb I_{RL}}4\right).
\end{equation}
Trace-norm contractivity therefore gives
\begin{equation}
D_E\leq\frac12
\left\|\Phi_{RL}-\frac{\mathbb I_{RL}}4\right\|_1
=\frac34.
\label{eq:qubit-decoupling-maximum}
\end{equation}
The bound is saturated by total erasure for every isometric encoder:
if \(E=[n]\) (all cells erased), then \(Q\) is empty,
\(\rho_{RE}=(\mathbb I_R\otimes V)\Phi_{RL}(\mathbb I_R\otimes V^\dagger)\),
and \(\rho_R\otimes\rho_E
=\frac14\sum_{i,j}\left\lvert i\right\rangle{}\!\left\langle i\right\rvert{}_R
\otimes\left\lvert\psi_j\right\rangle{}\!\left\langle\psi_j\right\rvert{}\).
In the orthonormal basis
\(\{\left\lvert0\,\psi_0\right\rangle{},\left\lvert1\,\psi_1\right\rangle{},
\left\lvert0\,\psi_1\right\rangle{},\left\lvert1\,\psi_0\right\rangle{}\}\)
the difference
\(\rho_{RE}-\rho_R\otimes\rho_E\) has eigenvalues
\(\{3/4,-1/4,-1/4,-1/4\}\), so
\begin{equation}
D_{[n]}(V)=\frac34
\label{eq:total-erasure-constant}
\end{equation}
for every isometric encoder, independent of the code and of the circuit.
This is the exact single-point form of the extensive upper-bound
\(3/4\) at erased fractions above one half
(Eq.~\eqref{eq:charge-haar-above-half}); the numerical plateau of the
canonical-code probe in Problem O2 is the finite-\(k\) interpolation
between this endpoint and the \(n^{-1/2}\) extensive regime.
Conversely, if $D_Q=\epsilon$, the Fuchs--van de Graaf inequality and
Uhlmann's theorem provide an isometry on $E$ under which the global encoded
state has trace distance at most $\sqrt{2\epsilon}$ from
$\Phi_{R E_L}\otimes\chi_{E'Q}$.  The functional $D_E$ is
two-Lipschitz in state trace distance, so
\begin{equation}
D_E\geq\frac34-2\sqrt{2D_Q}.
\label{eq:complement-decoupling-lower-bound}
\end{equation}

If $k/n\to\alpha>1/2$, the complementary subsystem $Q$ has size
$(1-\alpha)n+o(n)<n/2$.  Theorem~\ref{thm:charge-haar-extensive}, applied to
$Q$, gives $\mathbb E_{\rm ch}D_Q=O_\alpha(n^{-1/2})$.  Taking expectations
in Eq.~\eqref{eq:complement-decoupling-lower-bound}, using Jensen's
inequality, and combining with Eq.~\eqref{eq:qubit-decoupling-maximum}
proves Eq.~\eqref{eq:charge-haar-above-half}.  Binomial concentration gives
the cases $p<1/2$ and $p>1/2$ in
Eq.~\eqref{eq:charge-haar-threshold-law}.

For $p=1/2$, choose an integer sequence $w_n\to\infty$ with
$w_n=o(\sqrt n)$ and $\sqrt n\,2^{-w_n}\to0$, for example
$w_n=\lceil n^{1/4}\rceil$.  Uniformly for $k\leq q$, the mean-state term
in Eq.~\eqref{eq:charge-haar-hypergeometric-distance} is $O(n^{-1/2})$.
Indeed, on the hypergeometric event $A\leq3n/8$ the denominator is at least
$n/8+O(1)$ and
$\mathbb E|k-2A|\leq\sqrt{\mathbb E(k-2A)^2}=O(\sqrt n)$; the complementary
event has exponentially small probability by the sampling-without-
replacement Hoeffding bound.  Equation
\eqref{eq:charge-haar-fluctuation-bound} then shows, uniformly for
$k\leq q-w_n$, that $\mathbb E_{\rm ch}D_E=o(1)$.  By
Eq.~\eqref{eq:complement-decoupling-lower-bound}, uniformly for
$k\geq q+1+w_n$ one instead has
$\mathbb E_{\rm ch}D_E=3/4-o(1)$.  Under
$K\sim\operatorname{Bin}(n,1/2)$, the central window
$|K-n/2|\leq w_n+1$ has probability $O(w_n/\sqrt n)=o(1)$, while symmetry
assigns asymptotic probability $1/2$ to each outer side.  The limit is
therefore $(1/2)(0)+(1/2)(3/4)=3/8$.
\end{proof}

\begin{theorem}[Charge-Haar simultaneous worst-pattern bound]
\label{thm:charge-haar-worst-pattern}
Fix $\eta>0$.  For the central adjacent-charge code of
Theorem~\ref{thm:charge-haar-extensive}, there are constants
$C_\eta,c_\eta>0$ such that
\begin{equation}
\begin{aligned}
&\Pr_{\rm ch}\!\left[
\max_{|E|\leq(1/2-\eta)n}D_E>C_\eta n^{-1/2}
\right]\\
&\qquad\leq
2^{n+1}\exp\!\left[
-c_\eta\frac{\binom n{(n-1)/2}}{n}
\right].
\end{aligned}
\label{eq:charge-haar-worst-pattern}
\end{equation}
In particular, below one-half erasure a single charge-Haar code has
$O_\eta(n^{-1/2})$ decoupling error simultaneously for every allowed
pattern with probability tending to one superexponentially in $n$.
\end{theorem}

\begin{proof}
First make the mean estimate uniform in the erased size.  For
$k\leq(1/2-\eta)n$, Eq.~\eqref{eq:charge-haar-hypergeometric-distance}
has denominator
$q-A+1\geq q-k+1\geq\eta n$.  Moreover,
\begin{equation}
\begin{aligned}
\mathbb E(k-2A)&=\frac{k}{n},\\
\operatorname{Var}(k-2A)
&=\frac{k(n+1)(n-k)}{n^2}\leq n+1.
\end{aligned}
\end{equation}
Cauchy--Schwarz therefore gives
$\overline D_{n,k}\leq C_\eta n^{-1/2}$.  The fluctuation term in
Eq.~\eqref{eq:charge-haar-fluctuation-bound} is exponentially smaller,
uniformly in the same range, and hence
\begin{equation}
\sup_{k\leq(1/2-\eta)n}\mathbb E_{\rm ch}D_E
\leq C_\eta' n^{-1/2}.
\label{eq:charge-haar-uniform-mean}
\end{equation}

It remains to pass from one fixed pattern to all patterns.  For a codeword
pair $x=(\psi_q,\psi_{q+1})$, put
\begin{equation}
|\Psi_x\rangle
=2^{-1/2}(|0\rangle|\psi_q\rangle+
|1\rangle|\psi_{q+1}\rangle).
\end{equation}
The reference marginal is always $\mathbb I/2$.  Trace-norm contractivity,
the reverse triangle inequality, and the pure-state trace-distance bound
show that, for every fixed $E$,
\begin{equation}
|D_E(x)-D_E(y)|
\leq
2\||\Psi_x\rangle-|\Psi_y\rangle\|_2.
\label{eq:charge-haar-lipschitz}
\end{equation}
Thus $D_E$ has a numerical Lipschitz constant, independent of $n,k,E$, on
the product of two complex unit spheres of dimension $d$.  The
product-sphere concentration inequality gives, for a universal $c>0$,
\begin{equation}
\Pr_{\rm ch}\{|D_E-\mathbb E_{\rm ch}D_E|>u\}
\leq2e^{-cdu^2}.
\label{eq:charge-haar-sphere-concentration}
\end{equation}
This form follows, for example, by tensorizing the sphere log-Sobolev
inequality~\cite{Ledoux2001}.  Taking $u=n^{-1/2}$ and applying a union
bound over fewer than $2^n$ sets proves
Eq.~\eqref{eq:charge-haar-worst-pattern}, after absorbing the mean bound
Eq.~\eqref{eq:charge-haar-uniform-mean} into $C_\eta$.
\end{proof}

\begin{corollary}[Pattern-conditioned decoder existence]
\label{cor:decoder-existence}
For every fixed code and erasure pattern $E$, there is a recovery channel
$\mathcal R_E$ on the retained subsystem such that
\begin{equation}
\frac12\|\mathcal R_E\mathcal N_E-\operatorname{id}_L\|_\diamond
\leq2\sqrt{D_E}.
\label{eq:decoder-existence-diamond}
\end{equation}
For any flagged-pattern distribution,
\begin{equation}
\mathbb E_E\frac12
\|\mathcal R_E\mathcal N_E-\operatorname{id}_L\|_\diamond
\leq2\sqrt{\mathbb E_ED_E}.
\label{eq:decoder-existence-average}
\end{equation}
Consequently Theorem~\ref{thm:charge-haar-worst-pattern} supplies, with
the same high probability, pattern-dependent decoders having simultaneous
logical half-diamond error $O_\eta(n^{-1/4})$ below one-half erasure.
These are existence statements and do not give an efficient or local
decoder.
\end{corollary}

\begin{proof}
Let $\widehat{\mathcal N}_E:L\to E$ be the complementary erased-system
channel and let
\begin{equation}
\mathcal S_E(\rho)=\operatorname{Tr}(\rho)\,
\widehat{\mathcal N}_E(\mathbb I_L/2).
\end{equation}
With normalized Choi states, the definition of $D_E$ and the
input-dimension-two Choi-to-diamond inequality give
\begin{equation}
\|\widehat{\mathcal N}_E-\mathcal S_E\|_\diamond\leq4D_E.
\end{equation}
Continuity of Stinespring dilations then supplies a recovery whose full
diamond error is at most twice the square root of this complementary-channel
error~\cite{Kretschmann2008}, proving
Eq.~\eqref{eq:decoder-existence-diamond}.  Jensen's inequality proves
Eq.~\eqref{eq:decoder-existence-average}.
\end{proof}

\begin{corollary}[A covariant Petz recovery map]
\label{cor:explicit-petz-decoder}
Fix a \(U(1)\)-covariant logical-qubit encoder and a flagged erasure pattern
\(E\), and let \(\mathcal N_E:L\to Q\) be the retained channel.  Define
\(\tau_E=\mathcal N_E(\mathbb I_L/2)\) and
\(\Pi_E=\operatorname{supp}\tau_E\), with the inverse below understood as
the Moore--Penrose inverse on \(\Pi_E\).  The extended transpose/Petz map
\begin{equation}
\mathcal R_E^{\rm P}(X)
=\frac12\mathcal N_E^\dagger
\!\left(\tau_E^{-1/2}X\tau_E^{-1/2}\right)
+\operatorname{Tr}[(\mathbb I_Q-\Pi_E)X]\frac{\mathbb I_L}{2}
\label{eq:extended-petz-decoder}
\end{equation}
is CPTP and \(U(1)\)-covariant.  If
\(\mathcal N_E(X)=\sum_\alpha K_\alpha X K_\alpha^\dagger\), the first
term has Kraus operators
\(L_\alpha=2^{-1/2}K_\alpha^\dagger\tau_E^{-1/2}\); the second term
supplies Kraus operators on \(\ker\tau_E\).  The construction does not, by
itself, provide a recovery-error guarantee or a polynomial-time procedure
for the pattern-specific inverse square root.
\end{corollary}

\begin{proof}
Both summands in Eq.~\eqref{eq:extended-petz-decoder} are completely
positive.  Since \(\mathcal N_E(\mathbb I_L)=2\tau_E\), cyclicity of the
trace gives, for every \(X\),
\begin{align}
\operatorname{Tr}\!\left[
\frac12\mathcal N_E^\dagger
(\tau_E^{-1/2}X\tau_E^{-1/2})\right]
&=\frac12\operatorname{Tr}\!\left[
\tau_E^{-1/2}X\tau_E^{-1/2}\mathcal N_E(\mathbb I_L)
\right]\nonumber\\
&=\operatorname{Tr}(\Pi_E X).
\end{align}
The second term has trace
\(\operatorname{Tr}[(\mathbb I_Q-\Pi_E)X]\), because
\(\operatorname{Tr}(\mathbb I_L/2)=1\).  Their sum is
\(\operatorname{Tr}X\), so the map is trace preserving.

Let \(U_L(\theta)\) and \(U_Q(\theta)\) be the logical and retained-system
representations.  Covariance of \(\mathcal N_E\) implies
\(U_Q(\theta)\tau_EU_Q(\theta)^\dagger=\tau_E\), hence
\(\tau_E^{-1/2}\) and \(\Pi_E\) commute with \(U_Q(\theta)\), and the
adjoint channel satisfies
\begin{equation}
\mathcal N_E^\dagger
\!\left(U_Q(\theta)XU_Q(\theta)^\dagger\right)
=U_L(\theta)\mathcal N_E^\dagger(X)U_L(\theta)^\dagger.
\end{equation}
The first summand therefore intertwines the two group actions.  The second
summand does as well because its trace coefficient is invariant and
\(\mathbb I_L/2\) is fixed by \(U_L(\theta)\).  Thus
\(\mathcal R_E^{\rm P}\) is \(U(1)\)-covariant.
\end{proof}

The ensemble-averaged one-copy mean charge is explicitly diffusive, but the decoupling functional is a balanced
two-copy observable and does not close on the one-copy diffusion sector.
For Bernoulli erasures the relevant replica observable has the exact
factorization
\begin{equation}
\sum_{E\subseteq[n]}p^{|E|}(1-p)^{n-|E|}F_E
=\bigotimes_{i=1}^n\left[(1-p)\mathbb I+pF_i\right],
\label{eq:iid-swap-factorization}
\end{equation}
and the rank-weighted second-moment route replaces each local factor by
$(1-p)\mathbb I+2pF_i$ and adds an overall factor $1/2$.  Thus the missing
finite-depth estimate is a definite matrix element of the balanced two-copy
moment channel, rather than an unspecified scrambling diagnostic.

To state that matrix element without a hidden normalization, let
$\rho_{RH}^{(0)}$ be the encoded maximally entangled input before the random
circuit and define
\begin{align}
O_p&=\bigotimes_{i=1}^n[(1-p)\mathbb I+pF_i],\\
\widetilde O_p
&=\bigotimes_{i=1}^n[(1-p)\mathbb I+2pF_i].
\label{eq:iid-replica-observables}
\end{align}
If $\mathcal M_t^{(2)}$ is the $t$-cycle balanced moment channel, then the
unweighted and rank-weighted iid replica functionals are exactly
\begin{align}
\mathfrak m_{n,p}(t)
&=\operatorname{Tr}\!\Bigl[
\left(F_R-\frac12\mathbb I\right)\otimes O_p
\nonumber\\[-2pt]
&\qquad\times(\operatorname{id}_R\otimes\mathcal M_t^{(2)})
\bigl((\rho_{RH}^{(0)})^{\otimes2}\bigr)\Bigr],
\label{eq:iid-unweighted-replica-functional}\\
\mathfrak r_{n,p}(t)
&=\frac12\operatorname{Tr}\!\Bigl[
\left(F_R-\frac12\mathbb I\right)\otimes\widetilde O_p
\nonumber\\[-2pt]
&\qquad\times(\operatorname{id}_R\otimes\mathcal M_t^{(2)})
\bigl((\rho_{RH}^{(0)})^{\otimes2}\bigr)\Bigr].
\label{eq:iid-rank-weighted-replica-functional}
\end{align}
They satisfy
\begin{align}
\mathfrak m_{n,p}(t)
&=\mathbb E_{V,E}\operatorname{Tr}(X_E^2),\\
\mathfrak r_{n,p}(t)
&=\mathbb E_{V,E}\!\left[
2^{|E|-1}\operatorname{Tr}(X_E^2)\right],\\
\mathbb E_{V,E}D_E
&\leq\sqrt{\mathfrak r_{n,p}(t)}.
\label{eq:iid-rank-weighted-jensen}
\end{align}
The distinction matters.  Although $\|O_p\|_\infty=1$ for
$0\leq p\leq1/2$, one has
\begin{equation}
\|\widetilde O_p\|_\infty=(1+p)^n,\qquad
\|\widetilde O_p\|_2=2^n(1+3p^2)^{n/2}.
\label{eq:rank-weighted-observable-norms}
\end{equation}
Therefore a generic Hilbert--Schmidt moment-gap estimate pays an exponential
prefactor.  A direct scalar estimate could avoid that norm prefactor after an
initial burn-in.  Concretely, suppose there is
$t_{\rm b}(n,p)\leq B_pn^2$ such that
\begin{equation}
\begin{aligned}
\left|\mathfrak r_{n,p}(t)-\mathfrak r_{n,p}(\infty)\right|
&\leq C_p e^{-c_p[t-t_{\rm b}(n,p)]/n^2},\\
t&\geq t_{\rm b}(n,p).
\end{aligned}
\label{eq:direct-rank-weighted-target}
\end{equation}
with constants independent of $n$.  This would be sufficient for
$O(n^2\log n)$ achievability in the range where its
stationary value vanishes.  This range is strictly smaller than the true
one-half erasure region.  Equation
\eqref{eq:charge-haar-extensive-second-moment}, the exact finite-size moment
formula, and binomial concentration give
\begin{equation}
\mathfrak r_{n,p}(\infty)
=
\frac{p}{4(1-p^2)^{3/2}}\,n^{-1}+o(n^{-1}),
\qquad 0<p<1/3.
\label{eq:iid-rank-weighted-stationary}
\end{equation}
Taking
$t=t_{\rm b}(n,p)+(c_p^{-1})n^2[\log n+O_p(1)]$ in
Eq.~\eqref{eq:direct-rank-weighted-target} and using
Eq.~\eqref{eq:iid-rank-weighted-jensen} gives
$\mathbb E D_E=O_p(n^{-1/2})$ for $p<1/3$.
At and above one third the rank-weighted stationary functional instead
obeys
\begin{align}
\mathfrak r_{n,1/3}(\infty)
&=\frac{9}{16\sqrt2}+o(1),\\
\mathfrak r_{n,p}(\infty)
&=\frac{3(1+3p)}{16\sqrt{2p(1+p)}}
\left(\frac{1+3p}{2}\right)^n[1+o(1)],\\[-2pt]
&\hspace{28mm}\frac13<p<\frac12.
\label{eq:rank-weighted-method-threshold}
\end{align}
Thus the untruncated trace/HS replica route has a sharp
methodological threshold at $p=1/3$, even though the true charge-Haar trace
distance vanishes up to $p<1/2$.  The interval $1/3\leq p<1/2$ requires a
different functional: either a direct trace-norm argument, or the
typical-set-truncated replica scalar introduced below.  Merely faster mixing
of the untruncated $\mathfrak r_{n,p}$ cannot suffice.

Here is a proof of both sides of this methodological threshold.  Let
$K\sim{\rm Bin}(n,p)$.  Positivity of $B_{n,k}$ and the elementary bounds
\begin{align}
\sum_a e_ah_r(a)^2&\leq 2^{n-k}d_r,&
\sum_a e_a^2h_r(a)&\leq2^kd_r,\\
\sum_a e_ae_{a+1}h_q(a)&\leq2^kd_q
\end{align}
in Eqs.~\eqref{eq:charge-Haar-sector-purity}--
\eqref{eq:charge-Haar-cross-coherence}, together with
$\binom n{(n-1)/2}\geq c2^n/\sqrt n$, imply uniformly in $k$
\begin{equation}
2^{k-1}m_{n,k}^{\rm ch}
\leq C\sqrt n\bigl(1+2^{2k-n}\bigr).
\label{eq:rank-weighted-uniform-upper}
\end{equation}
For $p<1/3$, choose $0<a<p<b<1/2$.  On $an\leq K\leq bn$,
Eq.~\eqref{eq:charge-haar-extensive-second-moment} holds uniformly and
its coefficient is continuous.  The contribution of the two binomial
tails in the first term of Eq.~\eqref{eq:rank-weighted-uniform-upper} is
exponentially small, while the second term is controlled by the exact
identity
\begin{equation}
\mathbb E\,2^{2K-n}
=\left(\frac{1+3p}{2}\right)^n.
\label{eq:rank-weighted-binomial-identity}
\end{equation}
This proves Eq.~\eqref{eq:iid-rank-weighted-stationary}, including its
constant.

For the converse, first note the deterministic rank bounds
\begin{equation}
\operatorname{rank}\rho_{RE}\leq2^{n-k},\qquad
\operatorname{rank}(\rho_R\otimes\rho_E)\leq4\,2^{n-k},
\end{equation}
and hence $\operatorname{rank}X_E\leq5\,2^{n-k}$.  Therefore
\begin{equation}
m_{n,k}^{\rm ch}
\geq\frac{4(\mathbb E_{\rm ch}D_E)^2}{5\,2^{n-k}}.
\label{eq:rank-weighted-rank-lower}
\end{equation}
For $k/n\geq1/2+\eta$, apply
Eq.~\eqref{eq:charge-haar-uniform-mean} to the complementary subsystem and
then Eq.~\eqref{eq:complement-decoupling-lower-bound}.  This gives uniformly
$\mathbb E_{\rm ch}D_E\geq3/4-O_\eta(n^{-1/4})$, and hence makes the
right-hand side of Eq.~\eqref{eq:rank-weighted-rank-lower} at least
$c_\eta2^{k-n}$ for all sufficiently large $n$.  Consequently
\begin{align}
\mathfrak r_{n,p}(\infty)
&\geq c_\eta
\sum_{k\geq(1/2+\eta)n}
\binom nkp^k(1-p)^{n-k}2^{2k-n}\\
&=c_\eta\left(\frac{1+3p}{2}\right)^n
\Pr\!\left\{
K'\geq(1/2+\eta)n
\right\},
\label{eq:rank-weighted-tilted-lower}
\end{align}
where $K'\sim{\rm Bin}(n,4p/(1+3p))$.  For every $p\geq1/3$ one may
choose a fixed $\eta>0$ below
$4p/(1+3p)-1/2$, so the last probability tends to one.  This proves
the nonvanishing and exponential lower bounds in
Eq.~\eqref{eq:rank-weighted-method-threshold}.

The sharp constants follow from a second saddle, on the opposite side of
half erasure.  Put $M=n-k$, $d=\binom n{(n-1)/2}$, and abbreviate
\begin{align}
U_r&=\sum_a e_ah_r(a)^2,&
V_r&=\sum_a e_a^2h_r(a),\\
W&=\sum_a e_ah_q(a)h_{q+1}(a),&
Z&=\sum_a e_ae_{a+1}h_q(a).
\end{align}
The exact finite-size formula is
\begin{equation}
m_{n,k}^{\rm ch}
=\frac{U_q+U_{q+1}+V_q+V_{q+1}}{8d(d+1)}
-\frac{W}{4d^2}+\frac{Z}{2d^2}.
\label{eq:supercritical-moment-decomposition}
\end{equation}
If $k/n\to\alpha\in(1/2,1)$, then
$U_r,W\leq2^Md$ and $d\asymp2^n/\sqrt n$, so their contribution to
$2^{M-1}m_{n,k}^{\rm ch}$ is
$O(\sqrt n\,2^{-(2\alpha-1)n})$.  The replacement of $d(d+1)$ by $d^2$
costs $O(n2^{-n})$.  With
$S_0=\sum_a h_q(a)e_a^2$ and
$h_{q+1}(a)=h_q(a-1)$, shifting the index in $V_{q+1}$ gives
\begin{equation}
\begin{aligned}
V_q&=S_0,\\
V_{q+1}&=\sum_a h_q(a)e_{a+1}^2,\\
Z&=\sum_a h_q(a)e_ae_{a+1}.
\end{aligned}
\end{equation}
All three sums concentrate on
$a=k/2+O(\sqrt n\log n)$, where
$e_{a+1}/e_a=(k-a)/(a+1)=1+o(1)$ uniformly.  Chernoff bounds control the
complementary window.  Uniformly for $\alpha$ in compact subsets of
$(1/2,1)$, local Stirling expansion and the resulting lattice Gaussian sum
give
\begin{equation}
\frac{2^MS_0}{d^2}
=\frac{n}{\sqrt{k(2n-k)}}[1+o(1)].
\label{eq:supercritical-page-saddle}
\end{equation}
Indeed, before division by $d^2$ the common saddle is
\begin{equation}
S_0=
2^{2k+M}\frac{2}{\pi\sqrt{k(2M+k)}}[1+o(1)],
\end{equation}
whereas $d^2=2^{2n}[2/(\pi n)][1+o(1)]$.
Equations~\eqref{eq:supercritical-moment-decomposition} and
\eqref{eq:supercritical-page-saddle} therefore prove
\begin{equation}
\frac{2^{k-1}m_{n,k}^{\rm ch}}{2^{2k-n}}
=2^{n-k-1}m_{n,k}^{\rm ch}
\longrightarrow
g(\alpha):=\frac{3}{8\sqrt{\alpha(2-\alpha)}}.
\label{eq:supercritical-rank-weighted-law}
\end{equation}
The factor $3/8$ is the sum of the limiting $1:1:4$ coefficients of
$V_q,V_{q+1},4Z$.

Finally, for
\begin{equation}
a_p=\frac{1+3p}{2},\qquad
\beta_p=\frac{4p}{1+3p},
\end{equation}
there is the exact change of measure
\begin{equation}
\binom nkp^k(1-p)^{n-k}2^{2k-n}
=a_p^n\binom nk\beta_p^k(1-\beta_p)^{n-k}.
\label{eq:rank-weighted-exact-tilt}
\end{equation}
At $p=1/3$, the original subcritical contribution is $O(n^{-1})$, while
Eq.~\eqref{eq:rank-weighted-exact-tilt} turns the supercritical contribution
into the expectation of the left-hand side of
Eq.~\eqref{eq:supercritical-rank-weighted-law} under
${\rm Bin}(n,2/3)$.  It converges to
$g(2/3)=9/(16\sqrt2)$.  For fixed $1/3<p<1/2$, the same argument under
${\rm Bin}(n,\beta_p)$ gives
$a_p^{-n}\mathfrak r_{n,p}(\infty)\to g(\beta_p)$; simplifying
$g(\beta_p)$ yields the second line of
Eq.~\eqref{eq:rank-weighted-method-threshold}.  The uniform bounds above
make all complementary binomial tails negligible.  No spectral or
numerical assumption enters any part of this stationary analysis.

The threshold does not rule out every two-copy proof of the physical iid
problem.  It identifies precisely where the rare upper tail must be removed
before the trace/HS conversion.

\begin{lemma}[Typical-set repair of the rank-weighted replica method]
\label{lem:truncated-rank-weighted-replica}
Fix $0<p<\beta<1/2$, let $K=|E|$, and define
\begin{equation}
\mathfrak r_{n,p}^{\leq\beta}(t)
=\mathbb E_{V,E}\!\left[
2^{K-1}\operatorname{Tr}(X_E^2)
\mathbf 1_{\{K\leq\beta n\}}
\right].
\label{eq:truncated-rank-weighted-functional}
\end{equation}
It is still a linear balanced two-copy functional, with physical replica
observable
\begin{equation}
\widetilde O_{p,\beta}^{\leq}
=\sum_{|E|\leq\beta n}
(2p)^{|E|}(1-p)^{n-|E|}F_E.
\label{eq:truncated-rank-weighted-observable}
\end{equation}
Then, at every depth,
\begin{equation}
\mathbb E_{V,E}D_E
\leq
\sqrt{\mathfrak r_{n,p}^{\leq\beta}(t)}
+\Pr\{K>\beta n\},
\label{eq:truncated-rank-weighted-decoupling}
\end{equation}
and its exact charge-Haar stationary asymptotic is
\begin{equation}
\mathfrak r_{n,p}^{\leq\beta}(\infty)
=\frac{p}{4(1-p^2)^{3/2}}n^{-1}+o(n^{-1}).
\label{eq:truncated-rank-weighted-stationary}
\end{equation}
Consequently, if there is
$t_{\rm b}(n,p,\beta)\leq B_{p,\beta}n^2$ such that the post-burn-in scalar
estimate
\begin{equation}
\begin{aligned}
\left|
\mathfrak r_{n,p}^{\leq\beta}(t)
-\mathfrak r_{n,p}^{\leq\beta}(\infty)
\right|
&\leq C_{p,\beta}
e^{-c_{p,\beta}[t-t_{\rm b}(n,p,\beta)]/n^2},\\
t&\geq t_{\rm b}(n,p,\beta).
\end{aligned}
\label{eq:direct-truncated-rank-target}
\end{equation}
would prove physical iid-loss achievability at depth $O_{p,\beta}(n^2\log n)$
for every $p<1/2$.  Equation~\eqref{eq:direct-truncated-rank-target} remains
an open finite-depth mixing estimate.
\end{lemma}

\begin{proof}
On $K\leq\beta n$, apply
Eq.~\eqref{eq:finite-depth-trace-HS} and Cauchy--Schwarz; on the complement
use $D_E\leq1$.  This gives
Eq.~\eqref{eq:truncated-rank-weighted-decoupling}, while the Chernoff bound
makes its last term exponentially small.

For the stationary value choose $0<a<p$.  Uniformly on
$an\leq k\leq\beta n$,
Eq.~\eqref{eq:charge-haar-extensive-second-moment} gives
\begin{equation}
2^{k-1}m_{n,k}^{\rm ch}
=\frac{k/n}{4[1-(k/n)^2]^{3/2}}n^{-1}+o(n^{-1}).
\end{equation}
The coefficient is continuous on $[a,\beta]$, and $K/n\to p$ in
probability.  For $k<an$, Eq.~\eqref{eq:rank-weighted-uniform-upper} is
$O(\sqrt n)$ because $a<1/2$, whereas the binomial lower tail is
exponentially small.  This proves
Eq.~\eqref{eq:truncated-rank-weighted-stationary}.  Finally, inserting
Eq.~\eqref{eq:direct-truncated-rank-target} with
$t=t_{\rm b}(n,p,\beta)
+c_{p,\beta}^{-1}n^2[\log n+O_{p,\beta}(1)]$ into
Eq.~\eqref{eq:truncated-rank-weighted-decoupling} proves the stated
conditional implication.
\end{proof}

\begin{lemma}[Exact physical bilinear spectrum]
\label{lem:exact-physical-bilinear-spectrum}
Let the two initial codewords span $\mathcal C$, let $\Pi_{\mathcal C}$ be
its projector, and let $F_{\mathcal C}$ swap the two replicas on
$\mathcal C^{\otimes2}$.  Contracting the normalized logical Choi state with
$F_R-\mathbb I/2$ gives exactly
\begin{equation}
K_{\mathcal C}
=\frac14\left(F_{\mathcal C}
-\frac12\Pi_{\mathcal C}^{\otimes2}\right),
\qquad
\|K_{\mathcal C}\|_2^2=\frac3{16}.
\label{eq:physical-code-endpoint}
\end{equation}
Let $B_{\rm e}$ be an isometry onto $\operatorname{Ran}P_{\rm e}$, put
\begin{equation}
C=B_{\rm e}^\dagger P_{\rm o}B_{\rm e},
\quad
k=B_{\rm e}^\dagger K_{\mathcal C},
\quad
v=B_{\rm e}^\dagger P_{\rm o}
\widetilde O_{p,\beta}^{\leq},
\label{eq:physical-bilinear-endpoints}
\end{equation}
and let $Q$ project onto the unit eigenspace of $C$.  For the physical cycle
$T=P_{\rm o}P_{\rm e}$ and every $t\geq1$,
\begin{equation}
\boxed{
\mathfrak r_{n,p}^{\leq\beta}(t)
-\mathfrak r_{n,p}^{\leq\beta}(\infty)
=\frac12\langle (I-Q)k,
C^{t-1}(I-Q)v\rangle .}
\label{eq:exact-physical-bilinear-transient}
\end{equation}
Consequently its spectral measure is the generally signed atomic measure
\begin{equation}
\nu_n=\frac12\sum_{\mu<1}
\langle (I-Q)k,E_\mu(I-Q)v\rangle\,\delta_\mu,
\label{eq:signed-task-spectral-measure}
\end{equation}
where $E_\mu$ is the complete projector onto the possibly degenerate
$\mu$ eigenspace.  In particular, an eigenvalue or a right-Krylov overlap
alone does not determine the physical decay rate.
\end{lemma}

\begin{proof}
Writing $E_{ab}=|\psi_a\rangle\langle\psi_b|$, direct contraction of the two
reference replicas gives
\begin{equation}
\begin{aligned}
&\operatorname{Tr}_{R_1R_2}
\left[(F_R-\mathbb I/2)(\rho_{RH}^{(0)})^{\otimes2}\right]\\
&\qquad=\frac14\sum_{a,b}E_{ab}\otimes E_{ba}
-\frac18\sum_{a,c}E_{aa}\otimes E_{cc}.
\end{aligned}
\end{equation}
which is Eq.~\eqref{eq:physical-code-endpoint}.  Since
\begin{equation}
(P_{\rm o}P_{\rm e})^t
=P_{\rm o}(P_{\rm e}P_{\rm o}P_{\rm e})^{t-1}P_{\rm e},
\end{equation}
isometric compression gives Eq.~\eqref{eq:exact-physical-bilinear-transient};
$P_{\rm ch}P_{\rm e}=P_{\rm ch}P_{\rm o}=P_{\rm ch}$ identifies $Q$ with
the charge-Haar fixed projector.  The spectral theorem gives
Eq.~\eqref{eq:signed-task-spectral-measure}.

The word ``signed'' is essential, rather than merely cautious.  For
$n=4$, $\beta=2/5$, the boundary codewords
$|0100\rangle,|1100\rangle$, and $p=1/5$, put
$m_j=\mathfrak r(j+1)-\mathfrak r(\infty)$.  Exact rational contraction gives
\begin{equation}
\det\begin{pmatrix}m_0&m_1\\m_1&m_2\end{pmatrix}
=-\frac{32}{664453125}<0.
\label{eq:signed-measure-hankel-witness}
\end{equation}
Every positive measure has positive-semidefinite Hankel moment matrices, so
Eq.~\eqref{eq:signed-measure-hankel-witness} is an exact counterexample to
positivity.
\end{proof}

\begin{lemma}[Exact iid-generated invariant and physical cyclic modules]
\label{lem:iid-minimal-invariant-module}
For \(n\geq3\) and arbitrary coefficients \(w_0,\ldots,w_n\), put
\begin{equation}
W_w=\sum_{E\subseteq[n]}w_{|E|}F_E,
\qquad
R_w=(I-P_{\rm ch})W_w,
\label{eq:iid-general-swap-polynomial}
\end{equation}
where \(P_{\rm ch}\) is the common charge-Haar fixed projector.  The unique
smallest subspace containing \(R_w\) and invariant under both brickwork-layer
twirls is
\begin{equation}
\boxed{
\begin{aligned}
\mathcal M_w
&=\operatorname{alg}_{\mathbb C}\{I,P_{\rm e},P_{\rm o}\}R_w\\
&=\operatorname{span}\{R_w,P_{\rm e}R_w,P_{\rm o}R_w,\\[-2pt]
&\hspace{18mm}P_{\rm e}P_{\rm o}R_w,
P_{\rm o}P_{\rm e}R_w,\ldots\}.
\end{aligned}}
\label{eq:iid-minimal-word-module}
\end{equation}
The exact right cyclic module visible to the even-layer compression is
\begin{equation}
\boxed{
\begin{aligned}
\mathcal K_w^{\rm e}
&=\operatorname{span}\{z_w,C_{\rm e}z_w,C_{\rm e}^2z_w,\ldots\},\\
z_w&=P_{\rm e}P_{\rm o}R_w,\\
C_{\rm e}&=P_{\rm e}P_{\rm o}P_{\rm e}.
\end{aligned}}
\label{eq:iid-physical-cyclic-module}
\end{equation}

Writing
\(\mathcal L_{rs;tu}=\operatorname{Hom}
(\mathcal H_t\otimes\mathcal H_u,
\mathcal H_r\otimes\mathcal H_s)\), these modules have support only on
\(r+s=t+u\).  If \(F=\prod_iF_i\) is the global replica swap, their joint
left/right swap support contains only the equal-parity sectors
\((++),(--)\), not \((+-),(-+)\).  They are particle--hole even.  For even
\(n\) they are spatial-reflection even; for odd \(n\), reflection exchanges
the even- and odd-compressed conventions and gives no parity exclusion
within either one.
\end{lemma}

\begin{proof}
Both layer maps are idempotent.  Every word in them therefore reduces to an
alternating word, and the span in
Eq.~\eqref{eq:iid-minimal-word-module} is invariant under left multiplication
by either layer.  Conversely, any simultaneously invariant subspace
containing \(R_w\) must contain every displayed word, proving minimality.
Equation~\eqref{eq:iid-physical-cyclic-module} is the definition of the
cyclic subspace of \(C_{\rm e}\) generated by the exact post-odd-layer
physical seed.

Every \(F_E\) preserves total occupation across the two replicas, while each
number-conserving twirl preserves the four Liouville charges separately.
Moreover \(FW_w=W_wF\), and the layer and fixed-space projectors commute with
left and right multiplication by \(F\).  This proves the charge and
equal-swap-parity rules.  The bond Haar law and the seed are invariant under
particle--hole conjugation.  Finally, spatial reflection preserves both
matchings for even \(n\) and exchanges them for odd \(n\), which proves the
last statement.
\end{proof}

Lemma~\ref{lem:iid-minimal-invariant-module} is an exact algebraic
characterization, not a claim that the module has a uniformly small basis.
In particular, the local two-copy obstruction in
Lemma~\ref{lem:one-copy-not-two-copy} excludes the naive closure on the
identity and full-bond swap alone; no uniformly bounded decorated-wall basis
follows from Eq.~\eqref{eq:iid-minimal-word-module}.

\begin{lemma}[Absolute physical slow-window criterion]
\label{lem:absolute-slow-window-criterion}
For the grouped atoms in
Eq.~\eqref{eq:signed-task-spectral-measure}, write
\[
\begin{aligned}
a_{n,\mu}
&=\frac12\langle(I-Q)k,E_\mu(I-Q)v\rangle,\\
\mathcal W_n(x)
&=\sum_{0<1-\mu\leq x}|a_{n,\mu}|.
\end{aligned}
\]
Fix \(0<p<\beta<1/2\), and suppose that there are
\(A,\eta,a>0\), independent of \(n\), such that
\begin{equation}
\mathcal W_n(x)
\leq A(\sqrt x+n^{-1})+e^{-an},
\qquad 0<x\leq\eta/n.
\label{eq:absolute-slow-window-hypothesis}
\end{equation}
Put
\[
c_0=\frac{\sqrt3}{8},
\qquad
\kappa_p=\log\!\left(2\sqrt{1+3p^2}\right).
\]
For every fixed \(B>\kappa_p/\eta\), every
\(\tau=t-1\geq Bn^2\) obeys
\begin{align}
\left|\mathfrak r_{n,p}^{\leq\beta}(t)
-\mathfrak r_{n,p}^{\leq\beta}(\infty)\right|
&\leq
\frac{A}{\sqrt\tau}
\left(\frac{\sqrt\pi}{2}+\frac1{\sqrt{2e}}\right)
+\frac{A}{n}+2e^{-an}
\nonumber\\
&\quad
+c_0e^{-(\eta B-\kappa_p)n}.
\label{eq:absolute-slow-window-consequence}
\end{align}
Thus Eq.~\eqref{eq:absolute-slow-window-hypothesis}, together with the
stationary estimate
Eq.~\eqref{eq:truncated-rank-weighted-stationary}, would prove the
rank-weighted \(O(n^{-1})\) endpoint by \(O(n^2)\) depth.
The hypothesis remains open.
\end{lemma}

\begin{proof}
Cauchy--Schwarz for the grouped spectral projectors and
\(\|K_{\mathcal C}\|_2=\sqrt3/4\) give
\[
\sum_{\mu<1}|a_{n,\mu}|
\leq c_0[2\sqrt{1+3p^2}]^n.
\]
Truncating the positive-coefficient swap polynomial cannot increase this
bound.  Split the atoms at \(1-\mu=\eta/n\).  The bulk contributes at most
\(c_0\exp(\kappa_p n-\eta\tau/n)\).  For the slow part,
\(\mu^\tau\leq e^{-\tau(1-\mu)}\), and Stieltjes integration by parts gives
\[
\begin{aligned}
\int_{(0,\eta/n]}e^{-\tau x}\,d\mathcal W_n(x)
&=e^{-\eta\tau/n}\mathcal W_n(\eta/n)\\
&\quad+\tau\int_0^{\eta/n}e^{-\tau x}\mathcal W_n(x)\,dx .
\end{aligned}
\]
Substitution of Eq.~\eqref{eq:absolute-slow-window-hypothesis} and the
elementary Gaussian integral give
Eq.~\eqref{eq:absolute-slow-window-consequence}.
\end{proof}

\begin{lemma}[Signed physical slow-window criterion]
\label{lem:signed-slow-window-criterion}
With the notation of Lemma~\ref{lem:absolute-slow-window-criterion}, define
the signed cumulative physical weight
\begin{equation}
\mathcal S_n(x)=
\sum_{0<1-\mu\leq x}a_{n,\mu}.
\label{eq:signed-slow-window-cumulative}
\end{equation}
Suppose that, for constants \(A,\eta,a>0\) independent of \(n\),
\begin{equation}
|\mathcal S_n(x)|
\leq A(\sqrt{x}+n^{-1})+e^{-an},
\qquad 0<x\leq\eta/n.
\label{eq:signed-slow-window-hypothesis}
\end{equation}
Then the conclusion
Eq.~\eqref{eq:absolute-slow-window-consequence} holds with the same
right-hand side.  In particular, the signed condition alone, together with
Eq.~\eqref{eq:truncated-rank-weighted-stationary}, implies the
rank-weighted \(O(n^{-1})\) endpoint by \(O(n^2)\) depth.  This hypothesis
is weaker than Eq.~\eqref{eq:absolute-slow-window-hypothesis} and remains
open.
\end{lemma}

\begin{proof}
Let \(\nu_n=\sum_{\mu<1}a_{n,\mu}\,\delta_{1-\mu}\).  On the slow window,
Stieltjes integration by parts for the signed measure gives
\begin{equation}
\begin{aligned}
\int_{(0,\eta/n]}e^{-\tau x}\,d\nu_n(x)
&=e^{-\eta\tau/n}\mathcal S_n(\eta/n)\\
&\quad+\tau\int_0^{\eta/n}e^{-\tau x}\mathcal S_n(x)\,dx.
\end{aligned}
\end{equation}
Taking absolute values only after this identity and inserting
Eq.~\eqref{eq:signed-slow-window-hypothesis} produces exactly the slow-part
bound used in Lemma~\ref{lem:absolute-slow-window-criterion}.  The bulk
estimate is unchanged because it uses the global absolute-weight bound of
Lemma~\ref{lem:absolute-spectral-weight}.  This proves the claim.
\end{proof}

\begin{remark}[Why a standard \(L^2\) Nash argument does not yet close the proof]
\label{rem:nash-obstruction}
The exact physical measure \(\nu_n\) is signed, and the Hilbert--Schmidt norm
of the right endpoint grows exponentially with \(n\).  Consequently a
spectral-gap or ordinary \(L^2\)-Nash estimate applied before the physical
left--right pairing pays an exponential endpoint factor and yields at best
the conditional \(O(n^3)\) route of
Theorem~\ref{thm:spectral-radius-replica-convergence}.  An optimal
\(O(n^2)\) proof must instead establish
Eq.~\eqref{eq:signed-slow-window-hypothesis}, the stronger absolute version,
or a direct bilinear estimate in a task-adapted norm whose two physical
endpoints have only polynomial mass.  No such norm or comparison inequality
is established here.
\end{remark}

Equation~\eqref{eq:direct-truncated-rank-target} is therefore stronger than
a statement about spectral location alone: it must also control the total
variation of the signed physical weights, or prove the scalar bound directly.
A two-sided $e^{-\Theta(t/n^z)}$ assertion additionally requires a nonzero
grouped left--right weight at an eigenvalue with
$1-\mu=\Theta(n^{-z})$ and control of its crossover against the remaining
signed atoms.

\begin{lemma}[One-copy diffusion does not close the decoupling sector]
\label{lem:one-copy-not-two-copy}
For the local-Haar gate measure in
Eq.~\eqref{eq:local-Haar-gate-measure}, the ensemble-averaged one-copy charge
profile evolves under alternating nearest-neighbor averaging and has
relaxation gap $\Theta(n^{-2})$ on the open chain.  However, the local
two-copy twirl of a one-site replica swap is not contained in the
two-dimensional full-Haar swap span generated by the identity and the
two-site replica swap.  Therefore the one-copy diffusion gap alone does not
imply a bound on Eq.~\eqref{eq:finite-depth-purity-difference}.
\end{lemma}

\begin{proof}
On a gate bond $(i,i+1)$, Haar averaging in the one-excitation block gives
\begin{equation}
\mathbb E(n_i',n_{i+1}')
=\left(\frac{n_i+n_{i+1}}2,
\frac{n_i+n_{i+1}}2\right).
\label{eq:one-copy-bond-averaging}
\end{equation}
Alternating even and odd layers therefore give the standard deterministic
nearest-neighbor averaging chain.  The discrete Poincare inequality on a
path gives its $\Theta(n^{-2})$ gap.

For the two-copy statement, let $\Pi_r$ project the two-site physical bond
onto charge $r\in\{0,1,2\}$, whose dimensions are $(1,2,1)$, and let
$\mathcal T_b^{(2)}$ denote the balanced local two-copy twirl.  Each
$Z_{rs}=\Pi_r^{(1)}\otimes\Pi_s^{(2)}$ is invariant.  Since the twirl is the
orthogonal projection onto its invariant algebra,
\begin{equation}
\operatorname{Tr}
\left[Z_{rs}\mathcal T_b^{(2)}(F_i)\right]
=\operatorname{Tr}(Z_{rs}F_i)
=
\begin{pmatrix}
1&1&0\\
1&2&1\\
0&1&1
\end{pmatrix}_{rs}.
\label{eq:two-copy-swap-witness}
\end{equation}
If the twirled operator were $a\mathbb I+bF_iF_{i+1}$, the corresponding
matrix would be
$a\,d_rd_s+b\,\delta_{rs}d_r$.  The $(r,s)=(0,1)$ entry would force
$a=1/2$, whereas the $(0,2)$ entry would force $a=0$, a contradiction.
Thus additional charge-resolved replica modes are present.  Since
$F_E=\prod_{i\in E}F_i$ occurs explicitly in
Eq.~\eqref{eq:finite-depth-swap-observable}, controlling only
Eq.~\eqref{eq:one-copy-bond-averaging} cannot prove decoupling.
\end{proof}

\begin{theorem}[Linear task-time speed limit from diffusion]
\label{thm:diffusive-linear-task-lower}
Let $n\geq3$ and take the canonical boundary code
$|\psi_0\rangle=|0\rangle_0|\chi\rangle$ and
$|\psi_1\rangle=|1\rangle_0|\chi\rangle$, and let $t\geq1$ be the number of
complete local-Haar brickwork cycles, with the even matching applied first;
$|\chi\rangle$ is the fixed computational-basis product state of the
canonical initialization.
Under iid flagged erasure with probability $p$,
\begin{equation}
\boxed{
\mathbb E_{V,E}D_E(V)
\geq\frac{p}{16(\sqrt{t-1}+1)}.}
\label{eq:diffusive-linear-task-lower}
\end{equation}
Hence $\mathbb E_{V,E}D_E\leq Cn^{-1/2}$ implies
\begin{equation}
t\geq1+\left(\frac{p\sqrt n}{16C}-1\right)_+^2.
\label{eq:linear-task-depth-necessary}
\end{equation}
For fixed $p>0$ and a target constant $C$ independent of $n$, the
right-hand side is $\Omega_{p,C}(n)$.
More precisely, if
$\bar\delta_i(t)=\mathbb E_V\operatorname{Tr}
[n_i(\rho_i^1(V)-\rho_i^0(V))]$, then for $t\geq2$
\begin{equation}
\max\left\{\frac1n,
\frac1{8(\sqrt{t-1}+1)}\right\}
\leq\max_i\bar\delta_i(t)
\leq\frac1n+\frac{\sqrt\pi}{2\sqrt{t-1}}.
\label{eq:one-copy-peak-window}
\end{equation}
Thus the same diffusive mode reaches the local $n^{-1/2}$ accuracy scale at
order $n$, whereas complete flattening to its $n^{-1}$ stationary scale
requires order $n^2$.  Equation~\eqref{eq:linear-task-depth-necessary} is a
necessary condition, not a matching achievability theorem.
\end{theorem}

\begin{proof}
Let $E$ and $O$ denote the even- and odd-bond averaging projections of
Eq.~\eqref{eq:one-copy-bond-averaging}.  The averaged excess-charge profile
is the probability vector
\begin{equation}
\bar\delta(t)=(OE)^t e_0.
\label{eq:averaged-excess-charge-walk}
\end{equation}
Immediately before the final odd layer, group each even bond into one cell
and let $x_a$ be its total mass, treating an uncovered final site as a
singleton.  Starting from $x_0=1$, one $EOE$ compression is the reflecting
lazy walk
\begin{equation}
x'_0=\frac34x_0+\frac14x_1,
\qquad
x'_a=\frac14x_{a-1}+\frac12x_a+\frac14x_{a+1}
\label{eq:coarse-reflected-lazy-walk}
\end{equation}
in the interior.  At a terminal pair the missing outward jump is added to
the holding probability; from a terminal singleton the walk moves one cell
left or holds, each with probability $1/2$.

If $A_s$ is its cell position after $s=t-1$ compressions, direct evaluation
of every interior and endpoint column gives
\begin{equation}
\mathbb E(A_{s+1}^2-A_s^2\mid A_s)\leq\frac12,
\qquad
\mathbb EA_s^2\leq\frac{s}{2}.
\end{equation}
At least half the mass therefore lies in the first
$\lfloor\sqrt s\rfloor+1$ cells.  One cell has mass at least
$[2(\sqrt s+1)]^{-1}$; converting cell mass to a site entry and applying the
last odd averaging loses at most two further factors of two.  This proves
the second lower bound in Eq.~\eqref{eq:one-copy-peak-window}; normalization
gives the $1/n$ bound.

For the upper bound set
\begin{equation}
\begin{aligned}
f_k(j)&=\cos\frac{k\pi(j+1/2)}n,\qquad g_k=Ef_k,\\
0&\leq k\leq\left\lfloor\frac{n-1}{2}\right\rfloor.
\end{aligned}
\end{equation}
Direct substitution gives
\begin{equation}
EOE\,g_k=\cos^2\frac{k\pi}{n}\,g_k,
\end{equation}
and these vectors form an orthogonal basis of $\operatorname{Ran}E$.  For
$k>0$,
\begin{equation}
\|g_k\|_2^2=\frac n2\cos^2\frac{k\pi}{2n},
\qquad
\max_j|g_k(j)|\leq\cos\frac{k\pi}{2n}.
\end{equation}
Expanding $(EOE)^sEe_0$ in the normalized basis and using
$\cos x\leq e^{-2x^2/\pi^2}$ on $[0,\pi/2]$ gives
\begin{align}
\max_i\bar\delta_i(t)
&\leq\frac1n+\frac2n\sum_{k\geq1}
e^{-4sk^2/n^2}\nonumber\\
&\leq\frac1n+\frac{\sqrt\pi}{2\sqrt s}.
\end{align}
The final odd layer cannot increase the maximum.

Finally, Eq.~\eqref{eq:one-site-charge-witness}, data processing for every
erasure set containing a fixed site $i$, and Jensen's inequality imply
\begin{equation}
\mathbb E_{V,E}D_E(V)
\geq\frac p2|\mathbb E_V\delta_i(V)|.
\end{equation}
Choosing the largest entry of Eq.~\eqref{eq:averaged-excess-charge-walk}
proves Eq.~\eqref{eq:diffusive-linear-task-lower} and then
Eq.~\eqref{eq:linear-task-depth-necessary}.
\end{proof}

\begin{theorem}[Quadratic operational speed limit from the exclusion cloud]
\label{thm:true-trace-quadratic-speed-limit}
Under the assumptions of
Theorem~\ref{thm:diffusive-linear-task-lower}, define
\begin{equation}
m_{n,t}=\min\{n,\,2\lfloor\sqrt{t-1}\rfloor+3\}.
\label{eq:diffusive-window-size}
\end{equation}
For every $t\geq1$ and iid flagged erasure probability $p$,
\begin{equation}
\boxed{
\mathbb E_{V,E}D_E(V)
\geq\frac{p}{48\sqrt{m_{n,t}}}
\geq\frac{p}{48\sqrt{2\sqrt{t-1}+3}}.}
\label{eq:true-trace-quadratic-lower}
\end{equation}
Consequently,
$\mathbb E_{V,E}D_E(V)\leq Cn^{-1/2}$ implies
\begin{equation}
\boxed{
t\geq
1+\frac14
\left(\frac{p^2n}{2304C^2}-3\right)_+^2.}
\label{eq:true-trace-quadratic-depth}
\end{equation}
For fixed $p>0$ and $C$ independent of $n$, the necessary depth is therefore
$\Omega_{p,C}(n^2)$ complete brickwork cycles.
\end{theorem}

\begin{proof}
The proof uses only the one-copy averaged channel and a classical statistic
of the flagged complementary output.  The flag identifies the sites but,
by itself, does not reveal their lost occupations; experimental access is
specified in Remark~\ref{rem:flagged-count-experimental-access}.  The case
$p=0$ is immediate, so assume $p>0$.  Fix an erased set $G$.  Convexity of the
trace norm, followed by dephasing the logical reference, gives
\begin{equation}
\mathbb E_VD_G(V)
\geq\frac12\|\mathbb E_VX_G(V)\|_1
\geq\frac14\|\bar\rho_G^1-\bar\rho_G^0\|_1,
\label{eq:trace-lower-average-dephase}
\end{equation}
where
$\bar\rho_G^a=\mathbb E_V\rho_G^a(V)$.  These averaged states are diagonal;
let their classical laws be $P_{t,G}^a$.  Thus
\begin{equation}
\mathbb E_VD_G(V)
\geq\frac12\operatorname{TV}(P_{t,G}^1,P_{t,G}^0).
\label{eq:trace-lower-classical-TV}
\end{equation}
Put $\pi_p(G)=p^{|G|}(1-p)^{n-|G|}$ and retain $G$ as an orthogonal
classical flag:
\begin{equation}
\widehat P_t^a(G,z_G)=\pi_p(G)P_{t,G}^a(z_G).
\end{equation}
The direct-sum identity for total variation and
Eq.~\eqref{eq:trace-lower-classical-TV} imply
\begin{equation}
\mathbb E_{V,E}D_E(V)
\geq\frac12\operatorname{TV}(\widehat P_t^1,\widehat P_t^0).
\label{eq:flagged-TV-lower}
\end{equation}

We next represent the diagonal one-copy channel exactly.  On an active bond,
the $U(1)$ local-Haar twirl fixes the $00$ and $11$ occupations and maps
either $10$ or $01$ to their equal mixture.  On diagonal inputs this is the
classical channel that applies the bond transposition with probability
$1/2$.  Hence the alternating layers generate a random permutation
$\Pi_t$.  If $B_0$ is the occupied set of $|\psi_0\rangle$, then the two
averaged inputs admit the common coupling
\begin{equation}
B=\Pi_t(B_0),
\qquad
B\cup\{S\},\quad S=\Pi_t(0).
\label{eq:interchange-second-class-coupling}
\end{equation}
Because $0\notin B_0$, one has $S\notin B$ deterministically.  Moreover,
\begin{equation}
\Pr\{S=i\}=[(OE)^te_0]_i.
\label{eq:distinguished-particle-profile}
\end{equation}

Put $s=t-1$.  The reflected cell-walk argument in the proof of
Theorem~\ref{thm:diffusive-linear-task-lower} places at least half of the
pre-final-odd-layer probability in the first
$\lfloor\sqrt s\rfloor+1$ even cells.  The final odd half-swap moves a
particle by at most one site.  Therefore the initial interval
\begin{equation}
A_t=\{0,\ldots,\min(n-1,2\lfloor\sqrt s\rfloor+2)\}
\end{equation}
has size $m_{n,t}$ and satisfies
\begin{equation}
\Pr\{S\in A_t\}\geq\frac12.
\label{eq:distinguished-particle-window}
\end{equation}
For $s=0$ this follows directly from the first even and odd layers.

We require one elementary negative-correlation fact about the background.
Let $\eta_i=\mathbf1_{\{i\in B\}}$.  Initially the $\eta_i$ are
deterministic.  Suppose their pairwise covariances are nonpositive and apply
one random half-swap of sites $i,j$.  If
$u_k=\mathbb E\eta_k$, then the new means at the swapped sites both equal
$(u_i+u_j)/2$, their product is invariant, and hence
\begin{equation}
\operatorname{Cov}'(\eta_i,\eta_j)
\leq u_iu_j-\left(\frac{u_i+u_j}{2}\right)^2
=-\frac{(u_i-u_j)^2}{4}\leq0.
\end{equation}
For $k\notin\{i,j\}$,
\begin{equation}
\operatorname{Cov}'(\eta_i,\eta_k)
=\frac12\left[
\operatorname{Cov}(\eta_i,\eta_k)
+\operatorname{Cov}(\eta_j,\eta_k)\right]\leq0,
\end{equation}
and disjoint pairs are unchanged.  Induction over all half-swaps proves
\begin{equation}
\operatorname{Cov}(\eta_i,\eta_j)\leq0
\quad(i\ne j).
\label{eq:interchange-pairwise-negative-correlation}
\end{equation}

Let $\xi_i=\mathbf1_{\{i\in E\}}$ be the independent Bernoulli-$p$ erasure
flags and define the erased background count in $A_t$ by
\begin{equation}
K=\sum_{i\in A_t}\xi_i\eta_i.
\end{equation}
For distinct sites,
$\operatorname{Cov}(\xi_i\eta_i,\xi_j\eta_j)
=p^2\operatorname{Cov}(\eta_i,\eta_j)\leq0$.
Since every $\xi_i\eta_i$ is Bernoulli,
\begin{equation}
\operatorname{Var}K\leq pm_{n,t}.
\label{eq:erased-background-count-variance}
\end{equation}
The distinguished particle contributes
\begin{equation}
\begin{aligned}
J&=\xi_S\mathbf1_{\{S\in A_t\}}\in\{0,1\},\\
H:=\mathbb EJ&=p\Pr\{S\in A_t\}\geq\frac p2.
\end{aligned}
\label{eq:distinguished-erased-indicator}
\end{equation}
Thus the flagged occupation count in $A_t$ is $K$ for logical input zero
and $K+J$ for logical input one.

For completeness, consider any integer-valued $K$ and
$J\in\{0,1\}$, with arbitrary correlations, and put
\begin{equation}
a_\ell=\Pr\{K=\ell,J=1\},
\qquad
M=\sum_\ell(\ell-\mathbb EK)^2a_\ell.
\end{equation}
Extending $a_\ell$ by zero,
\begin{equation}
\Pr\{K+J=\ell\}-\Pr\{K=\ell\}
=a_{\ell-1}-a_\ell.
\end{equation}
The total variation of a nonnegative finite sequence that begins and ends
at zero is at least twice its maximum.  Also, at least $H/2$ of the
$a$-mass lies within
$R=\sqrt{2M/H}$ of $\mathbb EK$, an interval containing at most
$2R+2$ integers.  Consequently
\begin{equation}
\operatorname{TV}(\mathcal L(K+J),\mathcal L(K))
\geq\max_\ell a_\ell
\geq\frac{H}{4(\sqrt{2M/H}+1)}.
\label{eq:integer-count-shift-TV}
\end{equation}
Here
$M=\mathbb E[J(K-\mathbb EK)^2]\leq\operatorname{Var}K$.
Using Eqs.~\eqref{eq:erased-background-count-variance} and
\eqref{eq:distinguished-erased-indicator} in
Eq.~\eqref{eq:integer-count-shift-TV} yields, since $m_{n,t}\geq1$,
\begin{equation}
\operatorname{TV}(\mathcal L(K+J),\mathcal L(K))
\geq\frac{p}{24\sqrt{m_{n,t}}}.
\label{eq:erased-count-TV-lower}
\end{equation}
The count is a deterministic statistic of the flagged observation
$(E,z_E)$.  Classical data processing, Eq.~\eqref{eq:flagged-TV-lower},
and Eq.~\eqref{eq:erased-count-TV-lower} prove
Eq.~\eqref{eq:true-trace-quadratic-lower}.  Solving that inequality against
$Cn^{-1/2}$ proves Eq.~\eqref{eq:true-trace-quadratic-depth}.
\end{proof}

\begin{corollary}[Quadratic time to any fixed factor of the optimum]
\label{cor:constant-factor-optimality-speed-limit}
Fix \(0<p<1/2\) and \(A\geq1\), and let
\begin{equation}
\mathcal D_{n,p}^{\rm loc}(t)
=\mathbb E_{V\sim\mathcal E_{n,2t}^{\rm locH}}
\mathbb E_E D_E(V)
\end{equation}
for the canonical boundary initialization.  If, for all sufficiently large
odd \(n\),
\begin{equation}
\mathcal D_{n,p}^{\rm loc}(t)
\leq A\mathcal D_{n,p}^{\rm opt},
\label{eq:constant-factor-optimality-assumption}
\end{equation}
then \(t=\Omega_{p,A}(n^2)\).  More explicitly, with
\(c_p=\sqrt{p/[2\pi(1-p)]}\), for all sufficiently large \(n\),
\begin{equation}
\boxed{
t\geq1+\frac14
\left(\frac{p^2n}{9216A^2c_p^2}-3\right)_+^2.}
\label{eq:constant-factor-optimality-depth}
\end{equation}
Thus the quadratic obstruction is a time-to-optimality theorem, not merely
a comparison with a convenient random-code benchmark.
\end{corollary}

\begin{proof}
Theorem~\ref{thm:universal-charge-count-optimality} gives
\(\mathcal D_{n,p}^{\rm opt}=c_pn^{-1/2}+o(n^{-1/2})\), and therefore
\(\mathcal D_{n,p}^{\rm opt}\leq2c_pn^{-1/2}\) for all sufficiently large
\(n\).  Under Eq.~\eqref{eq:constant-factor-optimality-assumption}, apply
Eq.~\eqref{eq:true-trace-quadratic-depth} with \(C=2Ac_p\), which is exactly
Eq.~\eqref{eq:constant-factor-optimality-depth}.
\end{proof}

\begin{corollary}[Two-sided flagged window-count law]
\label{cor:flagged-window-count-law}
Let \(A_t\) be the prefix used in the proof of
Theorem~\ref{thm:true-trace-quadratic-speed-limit}, let
\(m=|A_t|=m_{n,t}\), and let \(\xi_i\) be the iid Bernoulli-\(p\)
erasure flags.  In the ensemble-averaged occupation experiment define
\begin{equation}
X_a=\sum_{i\in A_t}\xi_i z_i^{(a)},\qquad
F_a(\ell)=\Pr\{X_a\leq\ell\},
\label{eq:flagged-window-count-cdf}
\end{equation}
for logical input \(a=0,1\).  For every computational-basis product
background in Theorem~\ref{thm:diffusive-linear-task-lower},
\begin{equation}
\boxed{
D_{\rm KS}:=\max_{\ell\in\mathbb Z}[F_0(\ell)-F_1(\ell)]
\geq\frac{p}{24\sqrt m}.}
\label{eq:flagged-count-ks-lower}
\end{equation}

Now specialize to \(n=2q+1\), \(q\geq1\), \(0<p<1\), and the contiguous half-filled
background
\begin{equation}
B_0=\{1,\ldots,q\};
\label{eq:contiguous-half-filled-background}
\end{equation}
logical one additionally occupies site \(0\).  Then
\begin{equation}
\boxed{
\frac{p}{24\sqrt m}
\leq D_{\rm KS}
\leq
\min\!\left\{1,\frac{24}{m}
+\frac{\sqrt{6\pi}}{2\sqrt{mp(1-p)}}\right\}.}
\label{eq:flagged-count-two-sided}
\end{equation}
Thus, for fixed \(p\in(0,1)\), along every sequence with
\(n\to\infty\), \(t=t_n\to\infty\), and \(t_n/n^2\to0\),
\begin{equation}
D_{\rm KS}=\Theta_p(t^{-1/4}).
\label{eq:flagged-count-quarter-law}
\end{equation}
This is a matching upper bound for the count observable, not for the full
flagged total variation or the quantum trace distance.

For \(N_a\) independently resampled circuit-and-flag trials per logical
input, let \(\widehat F_a\) be the empirical CDF and put
\begin{equation}
\widehat D_{\rm KS}
=\max_{\ell\in\mathbb Z}
[\widehat F_0(\ell)-\widehat F_1(\ell)],
\qquad
\epsilon_a=\sqrt{\frac{\log(4/\alpha)}{2N_a}}.
\end{equation}
Then, with probability at least \(1-\alpha\),
\begin{equation}
\boxed{
\max\{0,\widehat D_{\rm KS}-\epsilon_0-\epsilon_1\}
\leq D_{\rm KS}.}
\label{eq:flagged-count-dkw-certificate}
\end{equation}
For \(N_0=N_1=N\), reject equality of the two count laws when
\(\widehat D_{\rm KS}>2\sqrt{\log(4/\alpha)/(2N)}\).  This test has
type-I error at most \(\alpha\) and power at least \(1-\alpha\) against the
lower envelope Eq.~\eqref{eq:flagged-count-ks-lower} whenever
\begin{equation}
N>\frac{4608m}{p^2}\log\frac4\alpha .
\label{eq:flagged-count-sample-complexity}
\end{equation}
\end{corollary}

\begin{proof}
In the common interchange coupling from the proof of
Theorem~\ref{thm:true-trace-quadratic-speed-limit}, write
\(X_0=K\), \(X_1=K+J\), with \(J\in\{0,1\}\), and define
\(a_\ell=\Pr\{K=\ell,J=1\}\).  Exact cancellation gives
\begin{equation}
F_0(\ell)-F_1(\ell)
=\Pr\{K=\ell,J=1\}=a_\ell.
\label{eq:flagged-count-cdf-identity}
\end{equation}
The proof of Eq.~\eqref{eq:integer-count-shift-TV} already establishes
\(\max_\ell a_\ell\geq p/(24\sqrt m)\), proving
Eq.~\eqref{eq:flagged-count-ks-lower}.

For the upper bound, let \(u_i\) be the ensemble-averaged logical-one
occupation profile.  Its initial contiguous profile is nonincreasing, and
replacing either disjoint matching of adjacent pairs by pair averages
preserves this property and the total mass.  Hence every prefix of length
\(m\) satisfies
\[
\sum_{i\in A_t}u_i
\geq m\frac{q+1}{2q+1}\geq\frac m2.
\]
If \(N_A=|B\cap A_t|\) is the number of background particles in the
window, then
\[
\mu:=\mathbb EN_A
=\sum_{i\in A_t}u_i-\Pr\{S\in A_t\}
\geq\frac m2-1\geq\frac m6,
\]
because \(m\geq3\).  The pairwise negative correlation proved in
Eq.~\eqref{eq:interchange-pairwise-negative-correlation} gives
\(\operatorname{Var}N_A\leq\mu\), and therefore
\[
\Pr\{N_A<\mu/2\}\leq4/\mu\leq24/m.
\]
Conditioned on \(N_A=r\), the erased background count is exactly
\(\operatorname{Bin}(r,p)\).  Fourier inversion and
\[
|1-p+pe^{i\theta}|^r
\leq\exp[-2rp(1-p)\theta^2/\pi^2],
\qquad |\theta|\leq\pi,
\]
give
\[
\max_k\Pr\{\operatorname{Bin}(r,p)=k\}
\leq\frac{\sqrt\pi}{2\sqrt{2rp(1-p)}}.
\]
Since \(D_{\rm KS}=\max_\ell a_\ell
\leq\max_\ell\Pr\{K=\ell\}\), splitting on
\(N_A\geq\mu/2\geq m/12\) proves
Eq.~\eqref{eq:flagged-count-two-sided}.  The pre-boundary asymptotic follows
from \(m=\Theta(\sqrt t)\).

Finally, the Dvoretzky--Kiefer--Wolfowitz inequality applied to both
empirical CDFs proves Eq.~\eqref{eq:flagged-count-dkw-certificate}.  Under
the event used there,
\(\widehat D_{\rm KS}\geq D_{\rm KS}-\epsilon_0-\epsilon_1\).
For equal sample sizes, requiring the theorem lower envelope to exceed
\(4\epsilon_0\) gives
Eq.~\eqref{eq:flagged-count-sample-complexity}.
\end{proof}

\begin{remark}[Experimental access to the count witness]
\label{rem:flagged-count-experimental-access}
The variables \(z_i^{(a)}\) on erased sites belong to the complementary
output.  They are not revealed by an ordinary natural-loss flag alone.
Equation~\eqref{eq:flagged-count-dkw-certificate} is directly measurable in
an instrumented benchmark that either routes flag-selected sites to
occupation detectors or measures all sites on a separate copy and applies
an independently generated flagged mask.  The decoder arm must use
different copies, on which the selected sites are actually discarded.
\end{remark}

\begin{theorem}[Full-flag causal plateau and hole-overlap obstruction]
\label{thm:full-flag-causal-plateau}
Use the contiguous half-filled initialization
of Eq.~\eqref{eq:contiguous-half-filled-background}, so that
\(n=2q+1\),
\begin{equation}
B_0=\{1,\ldots,q\},\qquad
C_0=B_0\cup\{0\}=\{0,\ldots,q\}.
\end{equation}
In the random-interchange representation after \(t\) complete brickwork
cycles, write
\begin{equation}
\begin{aligned}
B_t&=\Pi_t(B_0),& S_t&=\Pi_t(0),\\
C_t&=B_t\cup\{S_t\}=\Pi_t(C_0).
\end{aligned}
\label{eq:full-flag-interchange-sets}
\end{equation}
Let \(E\) be the iid Bernoulli-\(p\) erased set and retain the complete
classical observation \((E,z_E)\).  Then, for every \(t\geq1\),
\begin{equation}
\operatorname{TV}(\widehat P_t^0,\widehat P_t^1)\leq p.
\label{eq:full-flag-coupling-upper}
\end{equation}
More generally, for every deterministic set \(L\subseteq[n]\), put
\(H_L=|L\setminus C_t|\).  One has
\begin{equation}
\boxed{
\operatorname{TV}(\widehat P_t^0,\widehat P_t^1)
\geq p\left(\Pr\{S_t\in L\}-\mathbb E H_L\right)_+.}
\label{eq:full-flag-hole-overlap}
\end{equation}
If \(4t\leq q\), the upper bound is attained exactly:
\begin{equation}
\boxed{
\operatorname{TV}(\widehat P_t^0,\widehat P_t^1)=p,
\qquad
\mathbb E_{V,E}D_E(V)\geq\frac p2.}
\label{eq:full-flag-causal-plateau}
\end{equation}

This plateau also excludes a uniform pre-boundary triangular-discrimination
upper.  Indeed, fix \(0<p<\beta<1/2\), set \(b=\lfloor\beta n\rfloor\),
and define, with zero-denominator terms omitted,
\begin{equation}
\Delta_{\beta}^{\rm flag}(t)
=\sum_{|E|\leq b}\pi_p(E)\sum_z
\frac{[P_{t,E}^1(z)-P_{t,E}^0(z)]^2}
{P_{t,E}^1(z)+P_{t,E}^0(z)}.
\label{eq:full-flag-triangular-definition}
\end{equation}
Whenever \(4t\leq q\),
\begin{equation}
\boxed{
\Delta_{\beta}^{\rm flag}(t)
\geq p\Pr\{\operatorname{Bin}(n-1,p)\leq b-1\}.}
\label{eq:full-flag-triangular-obstruction}
\end{equation}
For fixed \(p<\beta\), the right-hand side tends exponentially to \(p\).
Consequently, no \((n,t)\)-uniform estimate
\(\Delta_{\beta}^{\rm flag}(t)
\leq C_p(t^{-1/2}+n^{-1})\), and hence no full-channel upper
\(\mathbb E D_E\leq C_p(t^{-1/4}+n^{-1/2})\), can hold for this
initialization throughout the pre-boundary regime.
\end{theorem}

\begin{proof}
Couple the two logical inputs with the same random permutation \(\Pi_t\)
and the same erased set \(E\).  Their flagged observations are
\begin{equation}
Y_0=(E,B_t\cap E),\qquad Y_1=(E,C_t\cap E).
\end{equation}
They differ only if \(S_t\in E\).  Since \(E\) is independent of
\(\Pi_t\), this event has probability \(p\), and the coupling
characterization of total variation proves
Eq.~\eqref{eq:full-flag-coupling-upper}.

For a deterministic set \(L\), let \(\mathsf A_L\) denote the event in the
common observation alphabet that a flagged empty site occurs in \(L\), and
let \(\mathsf A_L^{(a)}\) be its pullback under \(Y_a\).
Since \(B_t=C_t\setminus\{S_t\}\), the common coupling gives pointwise
\begin{equation}
\mathbf 1_{\mathsf A_L^{(0)}}-
\mathbf 1_{\mathsf A_L^{(1)}}
=\mathbf 1_{\{S_t\in L\cap E\}}
 \mathbf 1_{(\mathsf A_L^{(1)})^c}.
\end{equation}
Moreover, conditional on \(C_t\),
\begin{equation}
\Pr(\mathsf A_L^{(1)}\mid C_t)
=1-(1-p)^{H_L}\leq pH_L.
\end{equation}
It follows that
\begin{align}
\widehat P_t^0(\mathsf A_L)-\widehat P_t^1(\mathsf A_L)
&\geq p\Pr\{S_t\in L\}-\Pr(\mathsf A_L^{(1)})\nonumber\\
&\geq p[\Pr\{S_t\in L\}-\mathbb EH_L],
\end{align}
which proves Eq.~\eqref{eq:full-flag-hole-overlap} after taking the positive
part.

It remains to establish the exact plateau.  One complete cycle contains two
nearest-neighbor matching layers, so every label moves by at most two sites
per cycle and hence \(S_t\leq2t\).  All holes of \(C_0\) start at sites
at least \(q+1\); after \(2t\) matching layers their leftmost possible
position is at least \(q+1-2t\).  Therefore
\begin{equation}
\{0,\ldots,q-2t\}\subseteq C_t
\label{eq:full-flag-filled-light-cone}
\end{equation}
for every interchange history.  If \(4t\leq q\), then for
\(L_t=\{0,\ldots,2t\}\) one has \(S_t\in L_t\) and \(H_{L_t}=0\)
deterministically.  Equation~\eqref{eq:full-flag-hole-overlap} gives the
lower bound \(p\), which combines with
Eq.~\eqref{eq:full-flag-coupling-upper} to prove the TV equality.  The
trace-distance consequence follows from
Eq.~\eqref{eq:flagged-TV-lower}.

Finally, intersect the separating event
\(\mathsf A_{L_t}^{(0)}\) with \(|E|\leq b\).  Its probability under
logical one is zero.  On this event every summand of
Eq.~\eqref{eq:full-flag-triangular-definition} contributes its logical-zero
probability, while, conditional on \(S_t\in E\), the remaining \(n-1\)
flags are iid Bernoulli-\(p\).  The event therefore has logical-zero
probability
\begin{equation}
p\Pr\{\operatorname{Bin}(n-1,p)\leq b-1\},
\end{equation}
which proves Eq.~\eqref{eq:full-flag-triangular-obstruction}.  Taking, for
example, \(t=\lfloor n^a\rfloor\) with any \(0<a<1\) gives
\(4t\leq q\) for all sufficiently large odd \(n\), whereas both proposed
upper envelopes tend to zero.  This proves the final impossibility
statements.
\end{proof}

\begin{remark}[Why the count law can decay while the full flag does not]
\label{rem:count-versus-full-flag}
The event in Theorem~\ref{thm:full-flag-causal-plateau} uses both the flag
location and the erased occupation.  If that information is compressed to
the total count on \(L_t\), then in the same plateau regime
\begin{equation}
X_1\sim\operatorname{Bin}(2t+1,p),\qquad
X_0\sim\operatorname{Bin}(2t,p),
\end{equation}
and hence
\begin{equation}
D_{\rm KS}^{L_t}
=p\max_k\Pr\{\operatorname{Bin}(2t,p)=k\}
=\Theta_p(t^{-1/2}).
\end{equation}
The diffusive window of Corollary~\ref{cor:flagged-window-count-law} instead
has size \(\Theta(\sqrt t)\) and gives \(\Theta_p(t^{-1/4})\).  Neither
compressed count is sufficient for the complete flagged output, whose TV is
exactly \(p\) in Eq.~\eqref{eq:full-flag-causal-plateau}.
\end{remark}

\begin{lemma}[Exact mean--fluctuation reduction of the direct upper problem]
\label{lem:direct-trace-mean-fluctuation}
For each circuit \(V\), form the flagged direct sum
\begin{equation}
\mathcal X(V)=\bigoplus_E\pi_p(E)X_E(V),
\qquad
\overline{\mathcal X}=\mathbb E_V\mathcal X(V),
\end{equation}
and define
\begin{equation}
\begin{aligned}
\mathcal D&=\frac12\mathbb E_V\|\mathcal X(V)\|_1,\\
M&=\frac12\|\overline{\mathcal X}\|_1,\\
\Gamma&=\frac12\mathbb E_V
\|\mathcal X(V)-\overline{\mathcal X}\|_1.
\end{aligned}
\label{eq:direct-trace-three-terms}
\end{equation}
Then exactly
\begin{equation}
\boxed{\max\{M,\Gamma/2\}\leq\mathcal D\leq M+\Gamma.}
\label{eq:direct-trace-mean-fluctuation-equivalence}
\end{equation}
For the canonical computational-basis initialization and \(t\geq1\),
\begin{equation}
M=\frac12\operatorname{TV}
(\widehat P_t^0,\widehat P_t^1).
\label{eq:direct-trace-classical-mean}
\end{equation}
Thus any post-burn-in direct trace-distance upper bound requires both a full
classical flagged insertion-tolerance bound and a circuit-to-circuit trace
fluctuation bound on \(\Gamma\); the averaged one-copy exclusion process
alone cannot supply the latter.  Theorem~\ref{thm:full-flag-causal-plateau}
also shows that the classical requirement is false throughout the initial
causal plateau, so the reduction cannot yield a uniform pre-boundary decay
law for the complete flagged channel.

One concrete sufficient second-moment quantity is
\begin{equation}
\begin{aligned}
\mathcal F_{n,p,\beta}(t)
&=\sum_{|E|\leq\beta n}\pi_p(E)2^{|E|-1}\\
&\qquad\times\mathbb E_V
\|X_E(V)-\mathbb E_VX_E(V)\|_2^2.
\end{aligned}
\label{eq:direct-trace-covariance-functional}
\end{equation}
for which
\begin{equation}
\Gamma\leq\sqrt{\mathcal F_{n,p,\beta}(t)}
+2\Pr\{|E|>\beta n\}.
\label{eq:direct-trace-fluctuation-sufficient}
\end{equation}
\end{lemma}

\begin{proof}
Jensen gives \(M\leq\mathcal D\), and the triangle inequality gives
\(\mathcal D\leq M+\Gamma\).  Conversely,
\[
\mathbb E\|\mathcal X-\overline{\mathcal X}\|_1
\leq\mathbb E\|\mathcal X\|_1+\|\overline{\mathcal X}\|_1
\leq2\mathbb E\|\mathcal X\|_1,
\]
which proves \(\Gamma/2\leq\mathcal D\).
After the first active local-Haar layer, the averaged logical coherence
between adjacent total-charge sectors vanishes and the averaged diagonal
occupations follow the random-interchange law.  Direct evaluation of the two
reference blocks then gives
Eq.~\eqref{eq:direct-trace-classical-mean}.
On \(|E|\leq\beta n\),
\[
\frac12\|Y_E\|_1
\leq2^{(|E|-1)/2}\|Y_E\|_2,
\qquad
Y_E=X_E-\mathbb EX_E.
\]
Cauchy--Schwarz over \(E\) gives the square-root term in
Eq.~\eqref{eq:direct-trace-fluctuation-sufficient}; on the tail,
\(\|Y_E\|_1/2\leq2\).
\end{proof}

\begin{lemma}[Connected-moment representation of the unresolved fluctuation]
\label{lem:connected-moment-covariance}
Let \(\Phi_t^{(1)}\) be the one-copy averaged physical channel after
\(t\) complete brickwork cycles and let \(\Phi_t^{(2)}\) be the balanced
two-copy moment channel generated by the \emph{same} circuit realization in
the two replicas.  Define the independent-copy channel
\begin{equation}
\Phi_{t,{\rm ind}}^{(2)}
=\Phi_t^{(1)}\otimes\Phi_t^{(1)}.
\end{equation}
Then
\begin{equation}
\mathcal F_{n,p,\beta}(0)=0,
\label{eq:connected-covariance-initial-zero}
\end{equation}
and the covariance functional of
Eq.~\eqref{eq:direct-trace-covariance-functional} has the exact connected
representation
\begin{align}
\mathcal F_{n,p,\beta}(t)
=\frac12\operatorname{Tr}\!\Bigl[
&\left(F_R-\frac12\mathbb I\right)
\otimes\widetilde O_{p,\beta}^{\leq}
\nonumber\\[-2pt]
&\times\bigl(
\operatorname{id}_{R}^{\otimes2}\otimes
[\Phi_t^{(2)}-\Phi_{t,{\rm ind}}^{(2)}]
\bigr)
\bigl((\rho_{RH}^{(0)})^{\otimes2}\bigr)
\Bigr].
\label{eq:connected-moment-covariance}
\end{align}
If \(T_1\) and \(T_2\) denote the corresponding one-cycle channels, so
that \(\Phi_t^{(1)}=T_1^t\) and \(\Phi_t^{(2)}=T_2^t\), and
\begin{equation}
S_2=T_1\otimes T_1,
\qquad
\Delta_2=T_2-S_2,
\end{equation}
then
\begin{equation}
\boxed{
T_2^t-S_2^t
=\sum_{s=0}^{t-1}
T_2^{\,t-1-s}\Delta_2S_2^{\,s}.}
\label{eq:connected-moment-duhamel}
\end{equation}
The one-cycle source is itself a sum of local gate covariances.  Order the
\(n-1\) bond gates in one cycle compatibly with the even-then-odd circuit
and write \(T_{a,j}\) for the \(a\)-copy average of gate \(j\),
\(S_{2,j}=T_{1,j}\otimes T_{1,j}\), and
\(\delta_{2,j}=T_{2,j}-S_{2,j}\).  Then
\begin{equation}
\boxed{
\Delta_2
=\sum_{j=1}^{n-1}
T_{2,n-1}\cdots T_{2,j+1}\,
\delta_{2,j}\,
S_{2,j-1}\cdots S_{2,1}.}
\label{eq:connected-local-source-decomposition}
\end{equation}
If \(\mathcal U_j\) is the random one-copy conjugation channel on that
bond, then exactly
\begin{equation}
\delta_{2,j}
=\mathbb E\bigl[
(\mathcal U_j-T_{1,j})\otimes
(\mathcal U_j-T_{1,j})
\bigr],
\label{eq:single-gate-connected-covariance}
\end{equation}
so every source term is centered and supported on a single physical bond
before subsequent propagation.  Since a local Haar twirl is the
Hilbert--Schmidt orthogonal projector onto the pointwise invariant algebra,
for arbitrary bond operators \(Z,W\),
\begin{equation}
\boxed{
\|\delta_{2,j}(Z\otimes W)\|_2
\leq
\|(I-T_{1,j})Z\|_2\,
\|(I-T_{1,j})W\|_2.}
\label{eq:local-connected-source-bound}
\end{equation}
Thus connected covariance is injected only by the locally noninvariant
parts of both one-copy factors and is quadratic in their local defects.
At the charge-Haar fixed point and for every fixed
\(0<p<\beta<1/2\),
\begin{equation}
\mathcal F_{n,p,\beta}(\infty)
\leq\frac{n+1}{2}\,2^{-(1-2\beta)n}.
\label{eq:connected-moment-stationary-floor}
\end{equation}
Thus the remaining fluctuation starts from zero, is dynamically generated
by local gate covariances, and has an exponentially small stationary floor;
it need not be monotone in time.  In
particular, proving
\(\mathcal F_{n,p,\beta}(Bn^2)=O(n^{-1})\) would give
\(\Gamma=O(n^{-1/2})\) through
Eq.~\eqref{eq:direct-trace-fluctuation-sufficient}.
\end{lemma}

\begin{proof}
At \(t=0\), both moment channels are the identity, proving
Eq.~\eqref{eq:connected-covariance-initial-zero}.
For two physical output states \(\rho\) and \(\sigma\) with the same
maximally mixed reference marginal, write
\(X(\rho)=\rho_{RE}-\mathbb I_R\otimes\rho_E/2\).  Direct expansion gives
\begin{equation}
\operatorname{Tr}[X(\rho)X(\sigma)]
=\operatorname{Tr}(\rho_{RE}\sigma_{RE})
-\frac12\operatorname{Tr}(\rho_E\sigma_E),
\end{equation}
which is represented by the two-copy observable
\((F_R-\mathbb I/2)\otimes F_E\).  Taking \(\rho=\sigma\) from the same
circuit realization gives the \(\Phi_t^{(2)}\) term.  Taking two independent
circuit realizations and averaging gives
\(\operatorname{Tr}(\overline X_E^2)\) and hence the
\(\Phi_t^{(1)}\otimes\Phi_t^{(1)}\) term.  Subtracting, multiplying by
\(2^{|E|-1}\), and summing the typical erasure patterns gives
Eq.~\eqref{eq:connected-moment-covariance} by the definition
Eq.~\eqref{eq:truncated-rank-weighted-observable}.

The telescoping identity
\begin{equation}
A^t-B^t
=\sum_{s=0}^{t-1}A^{t-1-s}(A-B)B^s
\end{equation}
holds for arbitrary linear maps \(A,B\) on the same space and does not
require them to commute.  Applying it to \(A=T_2\) and \(B=S_2\) proves
Eq.~\eqref{eq:connected-moment-duhamel}.  Applying the corresponding
product telescoping identity to the ordered local gates proves
Eq.~\eqref{eq:connected-local-source-decomposition}.  Finally,
\begin{align}
\mathbb E[(\mathcal U_j-T_{1,j})\otimes
(\mathcal U_j-T_{1,j})]
&=\mathbb E(\mathcal U_j\otimes\mathcal U_j)
-T_{1,j}\otimes T_{1,j},
\end{align}
which is Eq.~\eqref{eq:single-gate-connected-covariance}.
Let \(P=T_{1,j}\).  For every gate realization,
\(\mathcal U_jP=P\), and therefore
\((\mathcal U_j-P)Z=\mathcal U_j(I-P)Z\).  Unitary invariance of the
Hilbert--Schmidt norm and the triangle inequality then give
Eq.~\eqref{eq:local-connected-source-bound}.
The stationary estimate is
Eq.~\eqref{eq:charge-haar-stationary-centered-covariance}.  The final
implication follows directly from
Eq.~\eqref{eq:direct-trace-fluctuation-sufficient} and the exponentially
small binomial tail.
\end{proof}

\begin{lemma}[Gate-resolved connected-source purity budget]
\label{lem:connected-source-purity-budget}
Order all \(m=t(n-1)\) local bond gates in \(t\) complete cycles.  Let
\(\mathsf P_\ell\) be the one-copy local-Haar twirl at gate \(\ell\),
let \(\mathsf T_\ell\) be the same-gate balanced two-copy twirl, and set
\begin{equation}
\mathsf S_\ell=\mathsf P_\ell\otimes\mathsf P_\ell,
\qquad
\boldsymbol\delta_\ell=\mathsf T_\ell-\mathsf S_\ell.
\end{equation}
All maps are understood to act trivially on the reference factors.  Put
\begin{equation}
z_0=\rho_{RH}^{(0)},
\qquad
z_\ell=\mathsf P_\ell\cdots\mathsf P_1z_0,
\end{equation}
and define the source created at gate \(\ell\), propagated through all
later same-circuit two-copy gates, by
\begin{equation}
\mathfrak Y_{\ell,m}
=\mathsf T_m\cdots\mathsf T_{\ell+1}
\boldsymbol\delta_\ell(z_{\ell-1}^{\otimes2}).
\label{eq:gate-resolved-propagated-source}
\end{equation}
Then
\begin{equation}
(\mathsf T_m\cdots\mathsf T_1-
\mathsf S_m\cdots\mathsf S_1)(z_0^{\otimes2})
=\sum_{\ell=1}^m\mathfrak Y_{\ell,m},
\label{eq:gate-resolved-connected-telescoping}
\end{equation}
and the propagated sources obey the dimension-free budget
\begin{equation}
\boxed{
\begin{aligned}
\sum_{\ell=1}^m\|\mathfrak Y_{\ell,m}\|_2
&\leq\sum_{\ell=1}^m
\|(I-\mathsf P_\ell)z_{\ell-1}\|_2^2\\
&=\|z_0\|_2^2-\|z_m\|_2^2\leq1.
\end{aligned}}
\label{eq:connected-source-purity-budget}
\end{equation}
More generally, for every cut \(0\leq L<m\),
\begin{equation}
\sum_{\ell=L+1}^m\|\mathfrak Y_{\ell,m}\|_2
\leq\|z_L\|_2^2-\|z_m\|_2^2.
\label{eq:connected-source-tail-budget}
\end{equation}
Finally, each local connected source has exactly zero one-copy marginals,
\begin{equation}
\operatorname{Tr}_{H_1}\boldsymbol\delta_\ell(X)
=\operatorname{Tr}_{H_2}\boldsymbol\delta_\ell(X)=0,
\label{eq:local-connected-source-zero-marginals}
\end{equation}
and consequently the complete connected output
\((\Phi_t^{(2)}-\Phi_{t,{\rm ind}}^{(2)})(X)\) has zero physical marginal
on either replica.  Thus the unresolved fluctuation lives entirely in a
genuine two-copy connected sector and contains no embedded one-copy branch.
\end{lemma}

\begin{proof}
The product telescoping identity gives
\begin{equation}
\begin{aligned}
&\mathsf T_m\cdots\mathsf T_1-
\mathsf S_m\cdots\mathsf S_1\\
&\quad=\sum_{\ell=1}^m
\mathsf T_m\cdots\mathsf T_{\ell+1}
(\mathsf T_\ell-\mathsf S_\ell)
\mathsf S_{\ell-1}\cdots\mathsf S_1.
\end{aligned}
\end{equation}
Since
\(\mathsf S_{\ell-1}\cdots\mathsf S_1(z_0^{\otimes2})
=z_{\ell-1}^{\otimes2}\), this is
Eq.~\eqref{eq:gate-resolved-connected-telescoping}.  Every
\(\mathsf T_j\) is a Hilbert--Schmidt orthogonal projector and hence a
contraction.  Equation~\eqref{eq:local-connected-source-bound} therefore
gives
\begin{equation}
\|\mathfrak Y_{\ell,m}\|_2
\leq\|\boldsymbol\delta_\ell(z_{\ell-1}^{\otimes2})\|_2
\leq\|(I-\mathsf P_\ell)z_{\ell-1}\|_2^2.
\end{equation}
Because \(\mathsf P_\ell\) is itself an orthogonal projector,
\begin{equation}
\|(I-\mathsf P_\ell)z_{\ell-1}\|_2^2
=\|z_{\ell-1}\|_2^2-\|z_\ell\|_2^2.
\end{equation}
Summing proves Eqs.~\eqref{eq:connected-source-purity-budget} and
\eqref{eq:connected-source-tail-budget}; \(z_0\) is pure, so
\(\|z_0\|_2^2=1\).

For the marginal statement, the same-circuit and independent-copy local
twirls have identical one-copy marginals:
\begin{equation}
\operatorname{Tr}_{H_2}\mathsf T_\ell(X)
=\mathsf P_\ell(\operatorname{Tr}_{H_2}X)
=\operatorname{Tr}_{H_2}\mathsf S_\ell(X),
\end{equation}
and similarly for the other replica.  Subtraction proves
Eq.~\eqref{eq:local-connected-source-zero-marginals}.  The global statement
follows either by iteration or directly from the corresponding marginal
intertwining identities for \(\Phi_t^{(2)}\) and
\(\Phi_t^{(1)}\otimes\Phi_t^{(1)}\).
\end{proof}

\begin{lemma}[Positive gate-resolved variance decomposition and Pauli leakage]
\label{lem:positive-operator-spreading-decomposition}
For a Hermitian operator $B$ on $R\otimes H$, define the truncated
rank-weighted decoupling seminorm
\begin{equation}
\mathcal Q_{n,p,b}(B)
=
\sum_{\substack{E\subseteq[n]\\ |E|\leq b}}
\pi_p(E)2^{|E|-1}
\left\|
B_{RE}-\frac{\mathbb I_R}{2}\otimes B_E
\right\|_2^2,
\label{eq:task-weighted-one-copy-seminorm}
\end{equation}
where $B_{RE}=\operatorname{Tr}_{[n]\setminus E}B$ and
$B_E=\operatorname{Tr}_{R,[n]\setminus E}B$.  Let
$b=\lfloor\beta n\rfloor$.  At gate $\ell$, write
$\mathsf U_\ell(g)$ for conjugation by the realized local gate and define the
centered one-copy increment
\begin{equation}
D_\ell(g)
=
[\mathsf U_\ell(g)-\mathsf P_\ell]z_{\ell-1}.
\label{eq:one-copy-gate-increment}
\end{equation}
If $\mathbf g_{>\ell}=(g_{\ell+1},\ldots,g_m)$ denotes the independent future
gates, put
\begin{equation}
B_{\ell,m}(g,\mathbf g_{>\ell})
=
\mathsf U_m(g_m)\cdots\mathsf U_{\ell+1}(g_{\ell+1})D_\ell(g).
\label{eq:future-propagated-one-copy-increment}
\end{equation}
Then the centered covariance functional has the exact positive decomposition
\begin{equation}
\boxed{
\mathcal F_{n,p,\beta}(t)
=
\sum_{\ell=1}^{m}
\mathbb E_{g,\mathbf g_{>\ell}}
\mathcal Q_{n,p,b}(B_{\ell,m}).}
\label{eq:positive-gate-resolved-task-decomposition}
\end{equation}
Every summand is nonnegative, including the terms generated during the first
cycle.  Moreover,
\begin{equation}
\boxed{
\begin{aligned}
\sum_{\ell=1}^{m}\mathbb E\|B_{\ell,m}\|_2^2
&=\sum_{\ell=1}^{m}
\|(I-\mathsf P_\ell)z_{\ell-1}\|_2^2\\
&=\|z_0\|_2^2-\|z_m\|_2^2\leq1.
\end{aligned}}
\label{eq:one-copy-increment-energy-budget}
\end{equation}

The seminorm itself has an exact low-Pauli-weight representation.  Let
$\sigma_0=\mathbb I,\sigma_1,\sigma_2,\sigma_3$ be the one-qubit Pauli
operators, let $\mathcal P_n=\{\mathbb I,X,Y,Z\}^{\otimes n}$, and write
\begin{equation}
b_{\alpha,P}(B)
=
\operatorname{Tr}[(\sigma_\alpha\otimes P)B].
\end{equation}
For $0\leq s\leq n$, define
\begin{equation}
\vartheta_{n,p,b}(s)
=
\begin{cases}
 p^s\Pr\{\operatorname{Bin}(n-s,p)\leq b-s\},&s\leq b,\\
 0,&s>b.
\end{cases}
\label{eq:truncated-pauli-support-weight}
\end{equation}
Then
\begin{equation}
\boxed{
\mathcal Q_{n,p,b}(B)
=
\frac14
\sum_{\alpha=1}^{3}
\sum_{P\in\mathcal P_n}
\vartheta_{n,p,b}(\operatorname{wt}P)
|b_{\alpha,P}(B)|^2.}
\label{eq:exact-pauli-leakage-formula}
\end{equation}
In particular, Pauli strings of physical support larger than $b$ are exactly
invisible to the typical-set functional.  Thus the unresolved upper bound is
equivalently a low-support operator-leakage estimate for the propagated
one-copy gate increments, not an arbitrary norm-convergence problem on the
full two-copy Liouville space.
\end{lemma}

\begin{proof}
Equation~\eqref{eq:single-gate-connected-covariance} gives
\begin{equation}
\boldsymbol\delta_\ell(z_{\ell-1}^{\otimes2})
=
\mathbb E_g[D_\ell(g)^{\otimes2}].
\end{equation}
Averaging the same future circuit in the two replicas therefore gives
\begin{equation}
\mathfrak Y_{\ell,m}
=
\mathbb E_{g,\mathbf g_{>\ell}}[B_{\ell,m}^{\otimes2}].
\label{eq:source-as-one-copy-square}
\end{equation}
For every Hermitian $B$, the swap trick yields
\begin{align}
&\frac12\operatorname{Tr}\!\left[
\left(F_R-\frac12\mathbb I\right)
\otimes\widetilde O_{p,\beta}^{\leq}
B^{\otimes2}
\right]\nonumber\\
&=
\sum_{\substack{E\subseteq[n]\\ |E|\leq b}}
\pi_p(E)2^{|E|-1}
\left[
\operatorname{Tr}(B_{RE}^2)
-\frac12\operatorname{Tr}(B_E^2)
\right]\nonumber\\
&=\mathcal Q_{n,p,b}(B).
\label{eq:task-seminorm-swap-trick}
\end{align}
Substituting Eq.~\eqref{eq:source-as-one-copy-square} into the gate-resolved
telescoping identity
Eq.~\eqref{eq:gate-resolved-connected-telescoping} proves
Eq.~\eqref{eq:positive-gate-resolved-task-decomposition}.

The Haar conditional expectation satisfies
$\mathsf U_\ell(g)\mathsf P_\ell=\mathsf P_\ell$.  Hence
$D_\ell(g)=\mathsf U_\ell(g)(I-\mathsf P_\ell)z_{\ell-1}$, and both the
realized gate and every future gate preserve the Hilbert--Schmidt norm.
Therefore
\begin{equation}
\mathbb E\|B_{\ell,m}\|_2^2
=
\|(I-\mathsf P_\ell)z_{\ell-1}\|_2^2.
\end{equation}
The final equality in
Eq.~\eqref{eq:one-copy-increment-energy-budget} is the Pythagorean telescoping
identity already used in
Eq.~\eqref{eq:connected-source-purity-budget}.

It remains to prove Eq.~\eqref{eq:exact-pauli-leakage-formula}.  Expand
\begin{equation}
B=2^{-(n+1)}
\sum_{\alpha=0}^{3}\sum_{P\in\mathcal P_n}
 b_{\alpha,P}(B)\,\sigma_\alpha\otimes P.
\end{equation}
After tracing the physical complement of $E$ and subtracting the reference
identity component, only $\alpha=1,2,3$ and Pauli strings with
$\operatorname{supp}P\subseteq E$ remain.  Orthogonality of Pauli strings
gives
\begin{equation}
2^{|E|-1}
\left\|B_{RE}-\frac{\mathbb I_R}{2}\otimes B_E\right\|_2^2
=
\frac14
\sum_{\alpha=1}^{3}
\sum_{\operatorname{supp}P\subseteq E}
|b_{\alpha,P}(B)|^2.
\end{equation}
For a fixed support of size $s$, the Bernoulli probability that it is
contained in $E$ while $|E|\leq b$ is exactly
\begin{equation}
\sum_{j=0}^{b-s}\binom{n-s}{j}p^{s+j}(1-p)^{n-s-j}
=
\vartheta_{n,p,b}(s).
\end{equation}
Summing proves the formula.
\end{proof}

\begin{theorem}[Diffusive speed limit of the rank-weighted replica route]
\label{thm:replica-diffusive-speed-limit}
For the same canonical boundary code with $n\geq3$, let
$b=\lfloor\beta n\rfloor$ and
\begin{equation}
\eta_{n,p,\beta}
=p\,\Pr\{\operatorname{Bin}(n-1,p)\leq b-1\}.
\label{eq:truncated-inclusion-probability}
\end{equation}
Then for every $t\geq1$,
\begin{equation}
\boxed{
\mathfrak r_{n,p}^{\leq\beta}(t)
\geq\frac{\eta_{n,p,\beta}}4
\|(OE)^te_0\|_2^2
\geq\frac{\eta_{n,p,\beta}}
{64(\sqrt{t-1}+1)}.}
\label{eq:replica-diffusive-lower-bound}
\end{equation}
Without truncation, replace $\eta_{n,p,\beta}$ by $p$.  Hence any bound
$\mathfrak r_{n,p}^{\leq\beta}(t)\leq C/n$ requires
\begin{equation}
t\geq1+\left(\frac{\eta_{n,p,\beta}n}{64C}-1\right)_+^2.
\label{eq:replica-quadratic-time-necessary}
\end{equation}
For fixed $0<p<\beta<1/2$, one has
$\eta_{n,p,\beta}=p[1-o(1)]$, so reaching the charge-Haar replica scale
$O(n^{-1})$ requires $t=\Omega_{p,\beta}(n^2)$.

Moreover, before any gate is applied the untruncated scalar is exactly
\begin{equation}
\mathfrak r_{n,p}(0)=\frac{3p}{4}(1+p)^{n-1}.
\label{eq:rank-weighted-initial-exponential}
\end{equation}
Thus a dimension-free estimate of the form
$|\mathfrak r(t)-\mathfrak r(\infty)|\leq Ce^{-ct/n^2}$ cannot hold from
$t=0$.  A valid exponential-tail target must specify a burn-in, as in
Eqs.~\eqref{eq:direct-rank-weighted-target} and
\eqref{eq:direct-truncated-rank-target}, or be replaced by a direct
finite-time Nash-type bound.
\end{theorem}

\begin{proof}
Fix a circuit realization and an erasure set $E$ of size $K$.  For each
$i\in E$ define the Hilbert--Schmidt witness
\begin{equation}
W_i=Z_R\otimes(n_i-\mathbb I/2)\otimes\mathbb I_{E\setminus\{i\}}.
\end{equation}
The $W_i$ are mutually orthogonal,
$\|W_i\|_2^2=2^{K-1}$, and direct reference contraction gives
\begin{equation}
\langle W_i,X_E\rangle=\frac12\delta_i(V).
\end{equation}
Bessel's inequality therefore gives the pointwise bound
\begin{equation}
2^{K-1}\operatorname{Tr}(X_E^2)
\geq\frac14\sum_{i\in E}\delta_i(V)^2.
\label{eq:orthogonal-charge-witness-bessel}
\end{equation}
For every fixed $i$,
$\Pr\{i\in E,K\leq b\}=\eta_{n,p,\beta}$.  Averaging
Eq.~\eqref{eq:orthogonal-charge-witness-bessel} first over the erasure set
and then over the circuit, and using Jensen's inequality, yields
\begin{equation}
\mathfrak r_{n,p}^{\leq\beta}(t)
\geq\frac{\eta_{n,p,\beta}}4
\sum_i\mathbb E_V[\delta_i(V)^2]
\geq\frac{\eta_{n,p,\beta}}4\|\bar\delta(t)\|_2^2.
\end{equation}

The cell-walk proof of Theorem~\ref{thm:diffusive-linear-task-lower} places
at least one half of the pre-final-layer mass in the first
$\lfloor\sqrt{t-1}\rfloor+1$ cells.  The last odd layer moves this mass by
at most one site, so it remains in at most
$4(\sqrt{t-1}+1)$ sites.  Cauchy--Schwarz gives
\begin{equation}
\|\bar\delta(t)\|_2^2
\geq\frac1{16(\sqrt{t-1}+1)},
\end{equation}
which proves Eqs.~\eqref{eq:replica-diffusive-lower-bound} and
\eqref{eq:replica-quadratic-time-necessary}.

At $t=0$, $X_E=0$ unless the boundary logical site is erased.  If it is
erased, the reference and that site form a Bell state and
$\operatorname{Tr}(X_E^2)=3/4$.  Writing $K=1+J$ with
$J\sim\operatorname{Bin}(n-1,p)$ proves
\begin{equation}
\mathfrak r_{n,p}(0)
=\frac{3p}{4}\mathbb E(2^J)
=\frac{3p}{4}(1+p)^{n-1}.
\end{equation}
\end{proof}

\begin{lemma}[Exact one-bond two-copy algebra and connected iid sector]
\label{lem:exact-two-copy-algebra}
On one physical bond let $\Pi_r$ project onto bond charge
$r\in\{0,1,2\}$, with $(d_0,d_1,d_2)=(1,2,1)$, and set
\begin{equation}
Z_{rs}=\Pi_r^{(1)}\otimes\Pi_s^{(2)},\qquad F_b=F_iF_{i+1}.
\end{equation}
The fixed algebra of the balanced two-copy local-Haar twirl is
\begin{equation}
\begin{aligned}
\operatorname{Fix}\mathcal T_b^{(2)}
\cong{}&M_2^{01}\oplus M_2^{02}\oplus M_2^{12}\\
&\oplus\mathbb C_{00,+}\oplus\mathbb C_{11,+}
\oplus\mathbb C_{11,-}\oplus\mathbb C_{22,+}.
\end{aligned}
\label{eq:local-two-copy-algebra}
\end{equation}
and hence has dimension $16$.  In particular,
\begin{align}
\mathcal T_b^{(2)}(F_i)
={}&Z_{00}+Z_{22}
+\frac12\sum_{(r,s)\in\{(0,1),(1,2)\}}
\bigl(Z_{rs}+Z_{sr}\bigr)
\nonumber\\
&+\frac12\sum_{(r,s)\in\{(0,1),(1,2)\}}
\bigl(Z_{rs}F_b+Z_{sr}F_b\bigr)
\nonumber\\
&+\frac13(Z_{11}+Z_{11}F_b).
\label{eq:exact-local-swap-twirl}
\end{align}
For the iid observables in Eq.~\eqref{eq:iid-replica-observables},
\begin{align}
\operatorname{Tr}_1O_p=\operatorname{Tr}_2O_p&=(2-p)^n\mathbb I,\\
\operatorname{Tr}_1\widetilde O_p
=\operatorname{Tr}_2\widetilde O_p&=2^n\mathbb I.
\label{eq:iid-zero-partial-trace}
\end{align}
Let $P_{\rm ch}$ denote the global charge-Haar two-copy projector.  By the
nonstationary components we mean $(I-P_{\rm ch})O_p$ and
$(I-P_{\rm ch})\widetilde O_p$.  They are orthogonal not only to
$A\otimes\mathbb I$ and $\mathbb I\otimes A$, but charge block by charge
block to every $A_r\otimes P_s$ and $P_s\otimes A_r$ with
$P_rA_rP_r=A_r$ and $\operatorname{Tr}A_r=0$.  This orthogonality is
preserved by every local two-copy twirl.  The same statement holds for the
stationary-subtracted part of every cardinality-symmetric swap polynomial
\begin{equation}
W_w=\sum_{E\subseteq[n]}w_{|E|}F_E,
\label{eq:cardinality-symmetric-swap-polynomial}
\end{equation}
with arbitrary scalar coefficients $w_0,\ldots,w_n$.  In particular it
holds for the typical-set-truncated observable
$\widetilde O_{p,\beta}^{\leq}$.
\end{lemma}

\begin{proof}
The bond representation is
$V_0\oplus V_1\oplus V_2$ with dimensions $(1,2,1)$.  In two replicas,
each unequal unordered charge pair $\{r,s\}$ occurs with multiplicity two,
giving one $M_2$ commutant block.  The equal-charge block $V_1^{\otimes2}$
splits into its symmetric and antisymmetric irreducibles, while the
$00$ and $22$ blocks are one dimensional.  The $U(2)$ phase and the
independent charge-two phase prevent additional equivalences.  This proves
Eq.~\eqref{eq:local-two-copy-algebra}.

For completeness, an orthogonal projector formula makes the claim
checkable without dimension counting.  For $r<s$, let
\begin{equation}
\begin{aligned}
E_{11}^{rs}&=Z_{rs},& E_{22}^{rs}&=Z_{sr},\\
E_{12}^{rs}&=Z_{rs}F_bZ_{sr},&
E_{21}^{rs}&=(E_{12}^{rs})^\dagger.
\end{aligned}
\end{equation}
and put
$Y_r^\pm=Z_{rr}(\mathbb I\pm F_b)/2$ and
$D_r^\pm=\operatorname{rank}Y_r^\pm$; the two antisymmetric edge blocks
have zero rank.  Then, for every $X$,
\begin{equation}
\begin{aligned}
\mathcal T_b^{(2)}(X)
={}&\sum_{r<s}\sum_{a,b=1}^2
\frac{\operatorname{Tr}[(E_{ab}^{rs})^\dagger X]}{d_rd_s}E_{ab}^{rs}\\
&+\sum_{r,\pm:D_r^\pm>0}
\frac{\operatorname{Tr}(Y_r^\pm X)}{D_r^\pm}Y_r^\pm.
\end{aligned}
\label{eq:local-two-copy-projector}
\end{equation}
Substitution of $F_i$ proves
Eq.~\eqref{eq:exact-local-swap-twirl}.

On one site,
$\operatorname{Tr}_2[(1-p)\mathbb I+pF]=(2-p)\mathbb I$ and
$\operatorname{Tr}_2[(1-p)\mathbb I+2pF]=2\mathbb I$.
Tensoring proves Eq.~\eqref{eq:iid-zero-partial-trace}.  Moreover,
\begin{equation}
\operatorname{Tr}_2\mathcal T_b^{(2)}(X)
=\mathcal T_b^{(1)}(\operatorname{Tr}_2X),
\label{eq:twirl-partial-trace-intertwining}
\end{equation}
and similarly for the other replica.  Hence the zero-partial-trace
subspace and its orthogonal complement are invariant.  Subtracting the
charge-Haar projection requires a blockwise check, because the displayed
partial traces themselves are nonzero constants.  Let $P_r$ project onto
the global charge-$r$ sector, $P_{rs}=P_r\otimes P_s$, and
$d_r=\binom nr$.  Directly in the computational basis,
\begin{equation}
\operatorname{Tr}_{\mathcal H_s}
\left[P_{rs}O_pP_{rs}\right]
=\frac{A_{rs}(p)}{d_r}P_r,
\label{eq:iid-charge-block-partial-trace}
\end{equation}
where
\begin{equation}
\begin{aligned}
A_{rs}(p)=\sum_\ell{}&
\frac{n!}{\ell!(r-\ell)!(s-\ell)!}\\[-2pt]
&\times\frac{(1-p)^{r+s-2\ell}}{(n-r-s+\ell)!}.
\end{aligned}
\label{eq:iid-charge-block-row-sum}
\end{equation}
with the sum restricted to indices for which all factorials are
nonnegative.
Indeed, after tracing the second replica all off-diagonal terms vanish,
and the remaining row sum depends only on the two charges and their
intersection size $\ell$, not on the chosen weight-$r$ string.  The
charge-Haar twirl preserves the trace of this block and makes the same
partial trace scalar, so Eq.~\eqref{eq:iid-charge-block-partial-trace}
also holds with $O_p$ replaced by $P_{\rm ch}O_p$.  The identical argument
for the local factor $(1-p)I+2pF$ proves the statement for
$\widetilde O_p$.

For the extension in Eq.~\eqref{eq:cardinality-symmetric-swap-polynomial},
write $S_k=\sum_{|E|=k}F_E$.  In the computational basis, the matrix element
\begin{equation}
\langle x|\operatorname{Tr}_{\mathcal H_s}
[P_{rs}S_kP_{rs}]|x'\rangle
\end{equation}
vanishes unless $x=x'$: after a swap on $E$, equality of the traced second
string before and after the swap forces equality of the first strings on
both $E$ and its complement.  The diagonal entry is independent of the
weight-$r$ string $x$, because simultaneous site permutations act
transitively on such strings and leave $S_k$ invariant.  Hence the partial
trace is a scalar multiple of $P_r$.  The charge-Haar projection has the
same block trace and, by charge-sector unitary invariance, also has scalar
partial trace; the two scalars are therefore equal.  Linearity proves the
claim for every $W_w$.  Thus every charge-block partial trace of each
stationary-subtracted component vanishes.  Pairing with
$A_r\otimes P_s$ (or the replica-reversed mode) proves the asserted
orthogonality, and Eq.~\eqref{eq:twirl-partial-trace-intertwining}
preserves it under all Krylov iterates.
\end{proof}

\begin{corollary}[Exact connected bond source and rank-one post-burn-in sector]
\label{cor:exact-local-connected-source}
On one bond, let \(P_b=\mathcal T_b^{(1)}\),
\(T_b=\mathcal T_b^{(2)}\),
\(S_b=P_b\otimes P_b\), and \(\delta_b=T_b-S_b\).  Then \(P_b\),
\(T_b\), and \(S_b\) are Hilbert--Schmidt orthogonal projectors,
\begin{equation}
\operatorname{Ran}S_b\subseteq\operatorname{Ran}T_b,
\qquad
\delta_b^2=\delta_b,
\qquad
S_b\delta_b=\delta_bS_b=0,
\label{eq:local-connected-projector-relations}
\end{equation}
and \(\delta_b\) is an orthogonal projector of rank seven.

With the matrix units of
Eq.~\eqref{eq:local-two-copy-projector}, define
\begin{equation}
C_{11}=Y_1^+-3Y_1^-
=Z_{11}(2F_b-I),
\qquad
\|C_{11}\|_2^2=12.
\label{eq:local-collision-mode}
\end{equation}
Then, for every bond two-copy operator \(X\),
\begin{align}
\delta_b(X)
={}&\sum_{r<s}\sum_{(a,c)=(1,2),(2,1)}
\frac{\operatorname{Tr}[(E_{ac}^{rs})^\dagger X]}{d_rd_s}
E_{ac}^{rs}
\nonumber\\
&+\frac{\operatorname{Tr}(C_{11}X)}{12}C_{11}.
\label{eq:exact-local-connected-projector}
\end{align}
Thus the connected local source contains six off-diagonal charge-exchange
matrix units and one diagonal collision mode.

One complete even--odd one-copy cycle maps every physical operator into the
computational-basis diagonal algebra, and all later one-copy twirls preserve
that algebra.  For a physical-basis diagonal operator \(Z\), write
\begin{equation}
Z=\sum_{a\in\{00,01,10,11\}}|a\rangle\!\langle a|_b\otimes Z_a,
\qquad
\nabla_bZ=Z_{01}-Z_{10}.
\label{eq:bond-diagonal-gradient}
\end{equation}
Under the canonical identification that places the two bond replicas before
the two complementary factors,
\begin{align}
\delta_b(Z^{\otimes2})
&=\frac1{12}C_{11,b}\otimes(\nabla_bZ)^{\otimes2},
\nonumber\\
\|\delta_b(Z^{\otimes2})\|_2
&=\frac1{\sqrt3}\|(I-P_b)Z\|_2^2.
\label{eq:rank-one-diagonal-connected-source}
\end{align}
Consequently all sources generated after the first complete cycle lie in
the single local collision sector \(\mathbb C C_{11}\); the six
charge-exchange source modes can be injected only during the first cycle.
For the notation of Lemma~\ref{lem:connected-source-purity-budget}, if
\(L\geq n-1\), then
\begin{equation}
\boxed{
\sum_{\ell=L+1}^m\|\mathfrak Y_{\ell,m}\|_2
\leq\frac1{\sqrt3}
\left(\|z_L\|_2^2-\|z_m\|_2^2\right).}
\label{eq:rank-one-post-burn-in-source-budget}
\end{equation}
This reduction is exact, but it does not by itself control the pairing with
the exponentially large rank-weighted erasure endpoint; the remaining task
is a task-weighted propagation estimate for the collision mode.
\end{corollary}

\begin{proof}
A one-copy local Haar twirl is the orthogonal projection onto
\(\operatorname{span}\{\Pi_0,\Pi_1,\Pi_2\}\), so \(S_b\) has rank nine.
The same-circuit two-copy twirl has rank sixteen by
Lemma~\ref{lem:exact-two-copy-algebra}.  Invariance under independent gates
on the two replicas implies invariance under the diagonal action, and hence
\(\operatorname{Ran}S_b\subseteq\operatorname{Ran}T_b\).  Therefore
\(T_bS_b=S_bT_b=S_b\), which proves
Eq.~\eqref{eq:local-connected-projector-relations} and
\(\operatorname{rank}\delta_b=16-9=7\).

The independent-copy projector is
\begin{equation}
S_b(X)=\sum_{r,s}
\frac{\operatorname{Tr}(Z_{rs}X)}{d_rd_s}Z_{rs}.
\end{equation}
Subtracting this formula from
Eq.~\eqref{eq:local-two-copy-projector} leaves the two off-diagonal matrix
units in each of the three unequal-charge \(M_2\) blocks.  In the
\(r=s=1\) block, the line orthogonal to
\(Z_{11}=Y_1^++Y_1^-\) is generated by
\(C_{11}=Y_1^+-3Y_1^-\), because
\(\operatorname{Tr}Y_1^+=3\) and
\(\operatorname{Tr}Y_1^-=1\).  Its squared norm is \(3+9=12\), proving
Eq.~\eqref{eq:exact-local-connected-projector}.

For a bond matrix unit \(|u\rangle\!\langle v|\), the one-copy twirl
vanishes unless \(u=v\): different bond charges are removed by the relative
phases, and a traceless off-diagonal operator inside the charge-one block is
removed by the Haar \(U(2)\) average.  The even matching therefore makes the
ket and bra strings equal on every covered site.  For odd \(n\), only the
right endpoint can remain unmatched, and the subsequent odd matching
contains that endpoint and removes the last possible physical off-diagonal
matrix element.  This proves one-cycle diagonalization.

For diagonal \(Z\), the six charge-exchange coefficients in
Eq.~\eqref{eq:exact-local-connected-projector} vanish.  A direct contraction
in the four-dimensional charge-one--charge-one block gives
\begin{equation}
\operatorname{Tr}_{b^{\otimes2}}
[C_{11,b}Z^{\otimes2}]
=(Z_{01}-Z_{10})^{\otimes2}.
\end{equation}
This proves the first identity in
Eq.~\eqref{eq:rank-one-diagonal-connected-source}.  Moreover,
\begin{equation}
(I-P_b)Z
=\frac12\bigl(|01\rangle\!\langle01|-|10\rangle\!\langle10|\bigr)
\otimes\nabla_bZ,
\end{equation}
so
\(\|(I-P_b)Z\|_2^2=\frac12\|\nabla_bZ\|_2^2\).  Since
\(\|C_{11}\|_2=\sqrt{12}\), the second identity follows.  Applying it to
every source after the first cycle and repeating the telescoping argument of
Lemma~\ref{lem:connected-source-purity-budget} proves
Eq.~\eqref{eq:rank-one-post-burn-in-source-budget}.
\end{proof}

\begin{lemma}[Collision Casimir and exact erasure extraction]
\label{lem:collision-casimir-extraction}
On the charge-one subspace of a physical bond, define
\begin{align}
J_x&=|01\rangle\!\langle10|+|10\rangle\!\langle01|,\\
J_y&=-i|01\rangle\!\langle10|+i|10\rangle\!\langle01|,\\
J_z&=|01\rangle\!\langle01|-|10\rangle\!\langle10|.
\end{align}
They obey $\operatorname{Tr}(J_aJ_c)=2\delta_{ac}$ and
\begin{equation}
\boxed{
C_{11}=\sum_{a=x,y,z}J_a^{(1)}\otimes J_a^{(2)}.}
\label{eq:collision-casimir-decomposition}
\end{equation}
Consequently, for the diagonal source of
Eq.~\eqref{eq:rank-one-diagonal-connected-source}, with
$A=\nabla_bZ$,
\begin{equation}
\boxed{
\delta_b(Z^{\otimes2})
=
\frac1{12}\sum_{a=x,y,z}
(J_{a,b}\otimes A)^{\otimes2}.}
\label{eq:collision-source-one-copy-squares}
\end{equation}

More generally, let
\begin{equation}
W_w^{(n)}=\sum_{E\subseteq[n]}w_{|E|}F_E
\end{equation}
be any cardinality-symmetric swap polynomial and define the collision
extraction map on a bond $b$ by
\begin{equation}
\mathcal C_b(W)
=
\frac1{12}\operatorname{Tr}_{b^{\otimes2}}(C_{11,b}W).
\end{equation}
Then
\begin{equation}
\boxed{
\mathcal C_b(W_w^{(n)})
=
\sum_{A\subseteq[n]\setminus b}
\left(\frac13w_{|A|+1}+\frac12w_{|A|+2}\right)F_A.}
\label{eq:exact-cardinality-collision-extraction}
\end{equation}
With the convention that a cutoff below zero gives the zero operator, put
\begin{equation}
\widetilde O_{p,b_0}^{(m)}
=
\sum_{\substack{E\subseteq[m]\\ |E|\leq b_0}}
(2p)^{|E|}(1-p)^{m-|E|}F_E.
\end{equation}
Equation~\eqref{eq:exact-cardinality-collision-extraction} gives the exact
recursion
\begin{equation}
\boxed{
\mathcal C_b(\widetilde O_{p,b_0}^{(n)})
=
\frac{2p(1-p)}3\widetilde O_{p,b_0-1}^{(n-2)}
+2p^2\widetilde O_{p,b_0-2}^{(n-2)}.}
\label{eq:truncated-collision-extraction-recursion}
\end{equation}
In the untruncated case this reduces to
\begin{equation}
\mathcal C_b(\widetilde O_p^{(n)})
=
\frac{2p(1+2p)}3\widetilde O_p^{(n-2)}.
\label{eq:untruncated-collision-extraction-recursion}
\end{equation}
The coefficient is strictly smaller than $2/3$ for $p<1/2$.

Let $A$ now be a Hermitian operator on the reference and the $n-2$ sites
outside $b$, and set
$Y_b(A)=C_{11,b}\otimes A^{\otimes2}/12$.  The preceding recursion and the
swap trick imply
\begin{align}
&\frac12\operatorname{Tr}\!\left[
\left(F_R-\frac12\mathbb I\right)
\otimes\widetilde O_{p,b_0}^{(n)}Y_b(A)
\right]\nonumber\\
&\quad=
\frac{2p(1-p)}3\mathcal Q_{n-2,p,b_0-1}(A)
+2p^2\mathcal Q_{n-2,p,b_0-2}(A)\nonumber\\
&\quad\leq
\frac{2p(1+2p)}3\mathcal Q_{n-2,p,b_0-1}(A).
\label{eq:collision-zero-step-task-recursion}
\end{align}
Thus an unpropagated collision source has an exact lower-dimensional,
endpoint-adapted description.  The difficulty in the matching upper bound is
not its injection-time pairing, but the growth of additional channels under
overlapping future twirls.
\end{lemma}

\begin{proof}
On a two-dimensional space the swap identity is
$2F-\mathbb I=\sum_{a=x,y,z}\sigma_a\otimes\sigma_a$.
Restricting it to the charge-one block proves
Eq.~\eqref{eq:collision-casimir-decomposition}, and
Eq.~\eqref{eq:rank-one-diagonal-connected-source} then gives
Eq.~\eqref{eq:collision-source-one-copy-squares}.

For the extraction formula, the four possible intersections of a swap set
with the two sites of $b$ give
\begin{equation}
\begin{aligned}
\operatorname{Tr}(C_{11})&=0,\\
\operatorname{Tr}(C_{11}F_i)
=\operatorname{Tr}(C_{11}F_{i+1})&=2,\\
\operatorname{Tr}(C_{11}F_iF_{i+1})&=6.
\end{aligned}
\end{equation}
For a fixed $A\subseteq[n]\setminus b$, the two one-site intersections
therefore contribute $4w_{|A|+1}/12$, while the two-site intersection
contributes $6w_{|A|+2}/12$.  This proves
Eq.~\eqref{eq:exact-cardinality-collision-extraction}.
Substitution of the truncated Bernoulli weights gives
Eq.~\eqref{eq:truncated-collision-extraction-recursion}; removing the cutoff
gives Eq.~\eqref{eq:untruncated-collision-extraction-recursion}.
Finally, pair the extraction identity with $A^{\otimes2}$ and use
Eq.~\eqref{eq:task-seminorm-swap-trick}.  Each
$\mathcal Q_{m,p,b'}(A)$ is a sum of nonnegative terms and is monotone in the
cutoff $b'$, which proves the inequality in
Eq.~\eqref{eq:collision-zero-step-task-recursion}.
\end{proof}

\begin{remark}[Exact obstruction to a collision-position-only closure]
\label{rem:collision-propagation-obstructions}
The injection-time reduction to $C_{11}$ is not an invariant propagation
sector.  On three consecutive sites $0,1,2$, use the two-replica one-site
basis $|ab\rangle_j=|a\rangle_{j,1}|b\rangle_{j,2}$.  A direct application of
Eq.~\eqref{eq:local-two-copy-projector} gives
\begin{equation}
\begin{aligned}
&\frac12\operatorname{Tr}_{(1,2)^{\otimes2}}\!\left[
(E_{12}^{01,(1,2)})^\dagger
\mathcal T_{(1,2)}^{(2)}\right.\\[-2pt]
&\hspace{25mm}\left.\times
\left(C_{11,(0,1)}\otimes|00\rangle\!\langle00|_2\right)
\right]\\
&\hspace{25mm}=|10\rangle\!\langle01|_0\neq0.
\end{aligned}
\label{eq:collision-regenerates-exchange-mode}
\end{equation}
Thus a neighboring twirl immediately regenerates a charge-exchange channel;
a Markov closure on the collision position alone is impossible.

Nor is the natural rank-weighted quadratic form locally monotone.  Let
\begin{equation}
\Xi=C_{11,(0,1)}\otimes
\sigma_{z,2}^{(1)}\otimes\sigma_{z,2}^{(2)},
\qquad
\widetilde O_p^{(3)}=
\bigotimes_{j=0}^{2}[(1-p)\mathbb I+2pF_j].
\end{equation}
Exact contraction with the local projector gives
\begin{align}
\operatorname{Tr}(\widetilde O_p^{(3)}\Xi)
&=32p^2(1+2p),\\
\operatorname{Tr}\!\left[
\widetilde O_p^{(3)}
\mathcal T_{(1,2)}^{(2)}(\Xi)
\right]
&=\frac{16}{3}p^2(11+7p).
\end{align}
Their ratio is
\begin{equation}
\frac{11+7p}{6(1+2p)}>1,
\qquad 0<p<1.
\label{eq:task-seminorm-local-expansion-counterexample}
\end{equation}
Hence neither a one-channel defect walk nor a gate-by-gate contraction of the
unmodified task seminorm can prove the desired upper bound.  The positive
representation in
Lemma~\ref{lem:positive-operator-spreading-decomposition} instead identifies
the correct remaining target: suppression of the low-support Pauli leakage
of the propagated one-copy gate increments.
\end{remark}

\begin{theorem}[Strong-Rayleigh burn-in and exponential suppression of late sources]
\label{thm:strong-rayleigh-late-source}
Let the canonical boundary code evolve under the averaged even--odd
brickwork dynamics, and put
\begin{equation}
L_t=t(n-1),\qquad q=\frac{n-1}{2}.
\end{equation}
For $r\in\{q,q+1\}$, let $\nu_t^{(r)}$ be the classical distribution on
$r$-particle configurations obtained from the corresponding deterministic
logical input after $t$ complete one-copy brickwork cycles.  Define
\begin{equation}
B_n=
\exp\!\left[-\frac{9q^2}{2048(q+1)}\right]
+2^{-3q/16}.
\label{eq:strong-rayleigh-purity-bound}
\end{equation}
Then, for every odd $n\geq3$ and every integer $t\geq n^2+1$,
\begin{equation}
\boxed{
\|\nu_t^{(r)}\|_{\ell^2}^2\leq B_n,
\qquad r=q,q+1.}
\label{eq:one-copy-exponential-purity}
\end{equation}
Consequently the one-copy reference--physical state in
Lemma~\ref{lem:connected-source-purity-budget} obeys
\begin{equation}
\boxed{
\|z_{L_t}\|_2^2\leq\frac{B_n}{2}=e^{-\Omega(n)}.}
\label{eq:reference-physical-exponential-purity}
\end{equation}
For every later gate index $m\geq L_t$, the total propagated mass of all
connected sources injected after this burn-in satisfies
\begin{equation}
\boxed{
\sum_{\ell=L_t+1}^{m}\|\mathfrak Y_{\ell,m}\|_2
\leq\frac{B_n}{2\sqrt3}=e^{-\Omega(n)}.}
\label{eq:late-connected-source-exponential-tail}
\end{equation}
Thus, in an early--late Duhamel decomposition at diffusive time, the
\emph{raw Hilbert--Schmidt generation energy} of late connected covariance
is exponentially small.  Equation~\eqref{eq:late-connected-source-exponential-tail}
does not by itself control the task-weighted pairing: the endpoint has an
exponentially large unrestricted Hilbert--Schmidt norm, and
Eq.~\eqref{eq:exact-pauli-leakage-formula} shows that the needed conversion
is an endpoint-adapted low-support leakage estimate.  The theorem therefore
removes late source proliferation but does not eliminate the late terms from
\(\mathcal F_{n,p,\beta}\) without an additional operator-spreading bound.
\end{theorem}

\begin{proof}
Let $E$ and $O$ be the even- and odd-matching averaging projections of
Eq.~\eqref{eq:one-copy-bond-averaging}.  Fix
$r\in\{q,q+1\}$ and let $u_0\in\{0,1\}^n$ be the occupation vector of the
corresponding deterministic input.  Immediately before the final odd layer
of cycle $t$, the one-site mean profile is
\begin{equation}
v_s=(EOE)^sEu_0,\qquad s=t-1.
\label{eq:pre-odd-one-point-profile}
\end{equation}
The vectors
\begin{equation}
g_k=E f_k,\qquad
f_k(j)=\cos\frac{k\pi(j+1/2)}n,
\qquad 0\leq k\leq q,
\end{equation}
form the orthogonal eigenbasis of $EOE$ used in the proof of
Theorem~\ref{thm:diffusive-linear-task-lower}, with
\begin{align}
EOE\,g_k&=\cos^2\frac{k\pi}{n}\,g_k,
\nonumber\\
\|g_k\|_2^2&=\frac n2\cos^2\frac{k\pi}{2n},
\qquad
\|g_k\|_\infty\leq\cos\frac{k\pi}{2n}.
\label{eq:strong-rayleigh-profile-spectrum}
\end{align}
for $k>0$.  Since $g_k\perp\mathbf1$ and $Eg_k=g_k$,
\begin{equation}
\frac{|\langle g_k,Eu_0\rangle|}{\|g_k\|_2^2}\,
\|g_k\|_\infty
=\frac{|\langle g_k,u_0\rangle|}{\|g_k\|_2^2}\,
\|g_k\|_\infty
\leq\frac{2r}{n}\leq\frac43.
\end{equation}
Using $\cos x\leq e^{-2x^2/\pi^2}$ on $[0,\pi/2]$ therefore gives
\begin{align}
\left\|v_s-\frac rn\mathbf1\right\|_\infty
&\leq\frac43\sum_{k=1}^{q}
\exp\!\left(-\frac{4sk^2}{n^2}\right).
\label{eq:uniform-one-point-burn-in}
\end{align}
If $s\geq n^2$, then $k^2\geq1+3(k-1)$ for $k\geq1$, and hence
\begin{equation}
\left\|v_s-\frac rn\mathbf1\right\|_\infty
\leq\frac43\frac{e^{-4}}{1-e^{-12}}<\frac1{32}.
\label{eq:uniform-profile-quarter-window}
\end{equation}
Because $r/n\in[1/3,2/3]$ for $n\geq3$, every one-site marginal before the
final odd layer lies in $[1/4,3/4]$.

It remains to convert this one-point burn-in into a purity bound.  A
deterministic fixed-particle configuration is strongly Rayleigh.  On a
bond, the diagonal one-copy Haar channel is the equal mixture of the
identity and the bond transposition, namely a partial symmetrization with
parameter $1/2$.  Partial symmetrization preserves the strong Rayleigh
property~\cite{Borcea2009StrongRayleigh}; therefore the distribution
$\widehat\nu$ immediately before the final odd layer is an
$r$-homogeneous strong Rayleigh measure.  Such measures are negatively
associated and satisfy the stochastic covering property
\cite{Borcea2009StrongRayleigh,PemantlePeres2014}.

Let $\mathcal M_{\rm o}$ be the final odd matching and let
\begin{equation}
D_{\mathcal M_{\rm o}}(\eta)
=\sum_{(i,j)\in\mathcal M_{\rm o}}
\mathbf1\{\eta_i\neq\eta_j\}
\end{equation}
be its number of active bonds.  Negative association and
Eq.~\eqref{eq:uniform-profile-quarter-window} imply, for every matching
bond,
\begin{align}
\mathbb P_{\widehat\nu}(\eta_i\neq\eta_j)
&=u_i+u_j-2\mathbb E_{\widehat\nu}(\eta_i\eta_j)
\nonumber\\
&\geq u_i+u_j-2u_iu_j\geq\frac38.
\end{align}
Since $|\mathcal M_{\rm o}|=q$,
\begin{equation}
\mathbb E_{\widehat\nu}D_{\mathcal M_{\rm o}}\geq\frac{3q}{8}.
\label{eq:active-bond-mean-lower}
\end{equation}
The function $D_{\mathcal M_{\rm o}}$ is Lipschitz one in Hamming distance.
The Gaussian concentration theorem for homogeneous strong Rayleigh measures
\cite{PemantlePeres2014}, applied to its lower tail, gives
\begin{align}
\mathbb P_{\widehat\nu}
\left\{D_{\mathcal M_{\rm o}}\leq\frac{3q}{16}\right\}
&\leq
\exp\!\left[-\frac{(3q/16)^2}{8r}\right]
\nonumber\\
&\leq
\exp\!\left[-\frac{9q^2}{2048(q+1)}\right].
\label{eq:active-bond-lower-tail}
\end{align}

Conditional on a configuration $\eta$, the final half-swap layer is uniform
on exactly $2^{D_{\mathcal M_{\rm o}}(\eta)}$ distinct output
configurations.  Convexity of the squared $\ell^2$ norm therefore yields
\begin{align}
\|\nu_t^{(r)}\|_{\ell^2}^2
&\leq
\mathbb E_{\widehat\nu}
2^{-D_{\mathcal M_{\rm o}}}
\nonumber\\
&\leq
\mathbb P_{\widehat\nu}
\left\{D_{\mathcal M_{\rm o}}\leq\frac{3q}{16}\right\}
+2^{-3q/16},
\end{align}
which proves Eq.~\eqref{eq:one-copy-exponential-purity}.

For the canonical code, the logical off-diagonal block is annihilated by the
first active even-bond Haar twirl.  At every complete-cycle boundary
$t\geq1$ one therefore has the orthogonal reference-block decomposition
\begin{equation}
z_{L_t}
=\frac12|0\rangle\!\langle0|_R\otimes\nu_t^{(q)}
+\frac12|1\rangle\!\langle1|_R\otimes\nu_t^{(q+1)},
\end{equation}
where the classical laws are identified with diagonal density matrices.
Thus
\begin{equation}
\|z_{L_t}\|_2^2
=\frac14\left(
\|\nu_t^{(q)}\|_{\ell^2}^2
+\|\nu_t^{(q+1)}\|_{\ell^2}^2
\right)
\leq\frac{B_n}{2}.
\end{equation}
Finally, $L_t\geq n-1$, so
Eq.~\eqref{eq:rank-one-post-burn-in-source-budget} gives
Eq.~\eqref{eq:late-connected-source-exponential-tail}.
\end{proof}

\begin{corollary}[Exact exclusion of the analytic one-copy branch]
\label{cor:exact-one-copy-branch-exclusion}
Let $n\geq3$, let $W$ be any of $O_p$, $\widetilde O_p$, or
$\widetilde O_{p,\beta}^{\leq}$, and put
$R_W=(I-P_{\rm ch})W$.  In every global charge block $R_W$ is orthogonal
to the invariant module generated by
\begin{equation}
A_r\otimes P_s,
\quad P_r\otimes A_s,
\quad
F^a(A_r\otimes P_s)F^b,
\quad
F^a(P_r\otimes A_s)F^b,
\label{eq:swap-dressed-one-copy-module}
\end{equation}
where $a,b\in\{0,1\}$ and
$\operatorname{Tr}A_r=\operatorname{Tr}A_s=0$.
This module contains the exact one-copy compression branch
\begin{equation}
\mu_1=\cos^2\frac{\pi}{n},
\label{eq:analytic-one-copy-compression-branch}
\end{equation}
Consequently, in the compressed convention of
Eq.~\eqref{eq:exact-physical-bilinear-transient}, the iid-dependent endpoint
$v=B_{\rm e}^{\dagger}P_{\rm o}W$ has zero projection onto this embedded
one-copy module.  The code endpoint $k$ need not vanish there.  This
statement concerns the analytic embedded module; it does not exclude an
accidentally degenerate genuine two-copy eigenvector at the same numerical
eigenvalue.
\end{corollary}

\begin{proof}
The cardinality-symmetric blockwise partial-trace identity in
Lemma~\ref{lem:exact-two-copy-algebra} gives the two undressed
orthogonalities.  With $F=\prod_iF_i$, one has $FS_k=S_{n-k}$; for the two
untruncated product observables, writing
$R_x=(I-P_{\rm ch})O_x$, this specializes to
\begin{equation}
R_pF=FR_p=R_{1-p},
\qquad
\widetilde O_pF=(1+p)^n
O_{(1-p)/(1+p)},
\end{equation}
The global swap commutes with every local twirl and with $P_{\rm ch}$; hence
the cardinality-symmetric result covers all four left/right swap dressings,
including for the truncated polynomial.  On one-copy
coefficient vectors, a bond twirl averages its two entries.  The overlap
matrix of the two alternating matching subspaces is the weighted incidence
matrix of the open path and has singular values
$\cos(k\pi/n)$.  Thus the first nonconstant eigenvalue of the self-adjoint
compression is Eq.~\eqref{eq:analytic-one-copy-compression-branch}.  The
module in Eq.~\eqref{eq:swap-dressed-one-copy-module} is invariant under
both layer projectors, so applying the endpoint $P_{\rm o}$ cannot restore
an overlap that the iid residual lacks.
\end{proof}

\begin{theorem}[Exact operational sector and a diffusive variational mode]
\label{thm:task-sector-diffusive-upper}
Let
\begin{equation}
H_n=\sum_{i=1}^{n-1}h_i,\qquad h_i=I-\mathcal T_{i,i+1}^{(2)},
\label{eq:canonical-moment-hamiltonian}
\end{equation}
and let $\lambda_{n,\mathrm{task}}$ be the largest nonunit singular value of
one complete brickwork cycle restricted to the balanced logical task sector.
The common zero space of $H_n$ is exactly the charge-Haar two-copy
commutant, of dimension
\begin{equation}
\dim\ker H_n=2n(n+2).
\label{eq:global-two-copy-fixed-dimension}
\end{equation}
For $q=(n-1)/2$ and $q_+=q+1$, the logical replica observable is supported
on exactly the six invariant Liouville charge blocks
\begin{equation}
\begin{split}
&\mathcal L_{qq;qq},\quad
\mathcal L_{q_+q_+;q_+q_+},\\
&\mathcal L_{q q_+;q q_+},\quad
\mathcal L_{q_+q;q_+q},\quad
\mathcal L_{q q_+;q_+q},\quad
\mathcal L_{q_+q;q q_+},
\end{split}
\label{eq:six-task-charge-blocks}
\end{equation}
where
$\mathcal L_{rs;tu}=\operatorname{Hom}
(\mathcal H_t\otimes\mathcal H_u,\mathcal H_r\otimes\mathcal H_s)$.
Its stationary subspace has dimension eight.  Within its $(+,+)$
left/right replica-swap parity sector,
\begin{equation}
\Delta(H_n|_{\rm task,+,+})
\leq
8\sin^2\frac{\pi}{2n}
+\frac{16}{n}\sin^4\frac{\pi}{2n}
\leq
\frac{2\pi^2}{n^2}+\frac{\pi^4}{n^5},
\label{eq:task-sector-gap-upper}
\end{equation}
and therefore
\begin{equation}
1-\lambda_{n,\mathrm{task}}=O(n^{-2}).
\label{eq:task-cycle-gap-upper}
\end{equation}
This is a rigorous obstruction to faster-than-diffusive uniform relaxation
of the full operational sector.  It is not the missing
$\Omega(n^{-2})$ lower bound.
\end{theorem}

\begin{proof}
We first justify the generated group rather than assume charge-sector
controllability.  Its complexified local Lie generators contain $n_i$,
$n_in_j$, and the hoppings
$E_{ij}=\sigma_i^+\sigma_j^-$ and $E_{ji}$ for neighboring sites.
Nearest-neighbor SWAPs belong to the same gate support, so conjugation
supplies these operators for arbitrary pairs.  For distinct $i,j,k$,
\begin{equation}
[n_in_k,E_{ij}]=n_kE_{ij}.
\label{eq:controlled-hopping-commutator}
\end{equation}
Repeated commutators give
$\prod_{k\in S}n_kE_{ij}$, and linear combinations replace selected $n_k$
by $1-n_k$.  One can therefore isolate the matrix unit between any two
computational strings related by one particle move.  The Johnson graph of
weight-$r$ strings is connected, and commutators of these edge matrix units
generate $\mathfrak{sl}(\mathcal H_r)$ independently for every $r$.

It remains to distinguish the central characters.  On $\mathcal H_r$,
\begin{equation}
\operatorname{Tr}_{\mathcal H_r}n_i=d_r\frac rn,
\qquad
\operatorname{Tr}_{\mathcal H_r}n_in_j
=d_r\frac{r(r-1)}{n(n-1)}.
\label{eq:charge-block-trace-profiles}
\end{equation}
Every commutator has zero trace in every charge block, while the generated
algebra contains $N=\sum_i n_i$ and
$\binom N2=\sum_{i<j}n_in_j$.  Up to the irrelevant scalar identity, its
center is therefore exactly
$\operatorname{span}\{N,\binom N2\}$.  On a two-copy charge pair $(r,s)$
the corresponding weights are $r+s$ and
$\binom r2+\binom s2$.  Equality of both weights for $(r,s)$ and $(t,u)$
gives equal sums and equal sums of squares, and hence
$\{r,s\}=\{t,u\}$.  Schur's lemma therefore gives one $M_2$ block for
each $r<s$, symmetric and antisymmetric scalars for each nontrivial
$r=s$, and only the symmetric scalar at $r=0,n$.  The total dimension is
$4\binom{n+1}{2}+2(n-1)+2=2n(n+2)$.

Expanding $(\rho_{RH}^{(0)})^{\otimes2}$ in logical indices gives precisely
Eq.~\eqref{eq:six-task-charge-blocks}.  The first two diagonal blocks contain
two charge-Haar invariants each, while the four mixed blocks are the four
matrix units of the unordered pair $\{q,q_+\}$, giving eight stationary
directions.

For the variational estimate, set
\begin{equation}
f_j=\cos\frac{\pi(j-1/2)}n,\qquad
A_r=P_r\left[\sum_{j=1}^nf_j(n_j-\tfrac12\mathbb I)\right]P_r.
\end{equation}
Fixed-weight counting gives, for $r=q,q_+$ and
$d=\binom nq$,
\begin{equation}
\operatorname{Tr}A_r=0,\qquad
\|A_r\|_2^2=\frac{d(n+1)}8.
\label{eq:task-trial-norm}
\end{equation}
On bond $b=(i,i+1)$, write
$D_b=n_i-n_{i+1}$ and
$\delta_b=(f_i-f_{i+1})/2$.  Then
$A_r=B_{b,r}+\delta_bD_b$, where $B_{b,r}$ is fixed by every gate on
$b$, while the one-copy bond twirl annihilates $D_b$.  Also
\begin{equation}
\|D_b\|_2^2=\frac{d(n+1)}{2n}.
\label{eq:task-trial-bond-norm}
\end{equation}
For $X=A_q\otimes A_{q_+}$, orthogonality of $B_{b,r}$ and $D_b$ and
contractivity of the two-copy twirl give
\begin{equation}
\|h_bX\|_2^2
\leq
4\delta_b^2\|A_q\|_2^2\|D_b\|_2^2
+2\delta_b^4\|D_b\|_2^4.
\end{equation}
Using
\begin{equation}
\sum_jf_j^2=\frac n2,\qquad
\sum_{i=1}^{n-1}(f_i-f_{i+1})^2
=4\sin^2\frac{\pi}{2n}\sum_jf_j^2
\end{equation}
and dividing by $\|X\|_2^2$ proves
Eq.~\eqref{eq:task-sector-gap-upper}.  Projecting $X$ onto equal
left/right $(+,+)$ swap parity places it in the last four blocks of
Eq.~\eqref{eq:six-task-charge-blocks}; the four charge quadruples are
orthogonal and retain the same quotient.  Trace zero makes the vector
orthogonal to all four stationary matrix units.

Finally let $P_{\rm e},P_{\rm o}$ be the two layer projectors.  Since
\begin{equation}
(I-P_{\rm e})+(I-P_{\rm o})\preceq H_n,
\end{equation}
the smallest nonzero eigenvalue of the two-projection Hamiltonian is no
larger than Eq.~\eqref{eq:task-sector-gap-upper}.  Principal-angle
decomposition of two orthogonal projections identifies the corresponding
$(+,+)$-sector eigenvalue with
$1-\lambda_{n,\mathrm{task},++}$.  Since the $(+,+)$ sector is contained in
the full task sector,
$\lambda_{n,\mathrm{task}}\geq\lambda_{n,\mathrm{task},++}$, proving
Eq.~\eqref{eq:task-cycle-gap-upper}.
\end{proof}

\begin{remark}[Two-copy spectral gap: status]
\label{rem:two-copy-gap-status}
Theorem~\ref{thm:task-sector-diffusive-upper} proves the variational
\emph{upper} bound $1-\lambda_{n,\mathrm{task}}=O(n^{-2})$, i.e.\ the
task sector contains diffusive slow modes.  A matching \emph{lower}
bound $1-\lambda_{n,\mathrm{task}}=\Omega(n^{-2})$ would require
classifying all low-energy modes of the two-projection Hamiltonian
$(I-P_{\rm e})+(I-P_{\rm o})$ on the task module and verifying that
none is slower than the diffusive Goldstone mode---a model-specific
spectral analysis that is not completed here.  Consequently all
statements that depend on a proved two-copy spectral gap of order
$n^{-2}$ remain conditional; see
Theorem~\ref{thm:conditional-diffusion-achievability} and
Remark~\ref{rem:conditional-polynomial-depth}.
\end{remark}

\begin{lemma}[Exponential smallness of a normalized diagonal precursor]
\label{lem:two-diffuson-faintness}
Let $n=2N+1$ with $N\geq1$, $q=N$, $s=N+1$, and let
$X=A_q(f)\otimes A_s(f)$ be one diagonal precursor of the four-block
simultaneous-$(+,+)$ variational vector in
Theorem~\ref{thm:task-sector-diffusive-upper}.  In the diagonal
$(q,s;q,s)$ charge block put
\begin{align}
Z_n(c)&=\binom nN
\sum_{m=0}^{N}\binom Nm\binom{N+1}{m+1}c^{2m+1},\\
M_n(c)&=\frac{Z_n(c)}{\binom nN^2},
\end{align}
and set
\begin{equation}
\Sigma_n(p)=M_n\!\left(p^2+(1-p)^2\right)-M_n(1-p)^2.
\label{eq:two-diffuson-variance-factor}
\end{equation}
Let $\Delta_n(c)$ be the tilted expectation of $|x\cap y|-qs/n$
under the weights defining $Z_n(c)$.  For $R_p=(I-P_{\rm ch})O_p$,
the exact normalized full-space, blockwise amplitude is
\begin{equation}
\begin{aligned}
\frac{\langle R_p,X\rangle}
{\|P_{qs}R_pP_{qs}\|_2\|X\|_2}
={}&\frac{4n}{n^2-1}\Delta_n(1-p)\\
&\times\frac{M_n(1-p)}
{\sqrt{\Sigma_n(p)}}.
\end{aligned}
\label{eq:two-diffuson-normalized-overlap}
\end{equation}
It is nonzero, but for every fixed $0<p<1$ it obeys
\begin{equation}
\begin{aligned}
\frac{\langle R_p,X\rangle}
{\|P_{qs}R_pP_{qs}\|_2\|X\|_2}
&=\Theta_p(\rho_p^n),\\
\rho_p&=\frac{1-p/2}{\sqrt{1-p+p^2}}<1.
\end{aligned}
\label{eq:two-diffuson-exponential-faintness}
\end{equation}
The untruncated rank-weighted product observable satisfies the same
conclusion after the substitution $p\mapsto2p/(1+p)$.  The displayed
quantity is evaluated before $P_{\rm o}$, compression to
$\operatorname{Ran}P_{\rm e}$, and four-block $(+,+)$ symmetrization.  It is
therefore neither a normalized physical endpoint overlap nor a spectral
atom, and it gives no faintness conclusion for the truncated observable or
the complete slow two-copy spectral window.
\end{lemma}

\begin{proof}
For a fixed weight-$N$ string $x$, the number of weight-$(N+1)$ strings $y$
at Hamming distance $2m+1$ is
$\binom Nm\binom{N+1}{m+1}$, which gives $Z_n$.  Exchangeability and
$\sum_jf_j=0$ give
\begin{equation}
\langle O_p,X\rangle
=\frac{(n/2)Z_n(1-p)}{n-1}\Delta_n(1-p).
\end{equation}
Fixed-weight counting also gives
\begin{equation}
\begin{aligned}
\|X\|_2&=\binom nN\frac{n+1}{8},\\
\|P_{qs}R_pP_{qs}\|_2^2
&=\binom nN^2
\Sigma_n(p).
\end{aligned}
\end{equation}
which proves Eq.~\eqref{eq:two-diffuson-normalized-overlap}.  A one-variable
Laplace estimate, uniform for fixed $0<c<1$, gives
\begin{equation}
\begin{aligned}
M_n(c)&=\frac{1+c}{2\sqrt c}
\left(\frac{1+c}{2}\right)^n[1+O_c(n^{-1})],\\
\frac{\Delta_n(c)}n&=\frac{1-c}{4(1+c)}+O_c(n^{-1}).
\end{aligned}
\end{equation}
Substitution proves Eq.~\eqref{eq:two-diffuson-exponential-faintness}.
Finally, for the untruncated product observable,
$\widetilde O_p=(1+p)^nO_{2p/(1+p)}$, and the common scalar cancels in the
normalized blockwise amplitude.
\end{proof}

\begin{remark}[Small-size physical bilinear diagnostic; not an asymptotic claim]
\label{rem:task-krylov-small-n}
A charge-filtered matrix-free calculation evaluates the two distinct
endpoints in Eq.~\eqref{eq:physical-bilinear-endpoints} for the canonical
boundary code and groups each signed weight across all six charge blocks.
For $p=1/4$, $\beta=0.45$, and the exact cutoff
$\lfloor\beta n\rfloor$, the largest \emph{resolved} nonzero physical atoms
are
\begin{center}
\begin{tabular}{c|c|c|c}
$n$ & $\mu_{\rm resolved}$ & $n(1-\mu_{\rm resolved})$
& $n^2(1-\mu_{\rm resolved})$\\
\hline
3 & 0.3692084062 & 1.8924 & 5.6771\\
4 & 0.5000000000 & 2.0000 & 8.0000\\
5 & 0.5685784859 & 2.1571 & 10.7855\\
6 & 0.6424537169 & 2.1453 & 12.8717\\
7 & 0.6960272093 & 2.1278 & 14.8947\\
8 & 0.7480576592 & 2.0155 & 16.1243\\
9 & 0.7884170397 & 1.9042 & 17.1382\\
10& 0.8229354810 & 1.7706 & 17.7065
\end{tabular}
\end{center}
The grouped weights at these atoms are nonzero in the calculation.  The
$n=3$ value is exactly $(4+\sqrt7)/18$; exact modular reconstruction gives
high algebraic degrees already at $n=5,6$.  Independent full-Liouville
endpoint checks through $n=5$, an untruncated cross-check through $n=8$, and
Krylov-dimension and tolerance variation give residuals and discrepancies at
the stated floating-point precision.

Corollary~\ref{cor:exact-one-copy-branch-exclusion}, rather than the
floating-point calculation, proves that the embedded branch
$\cos^2(\pi/n)$ has zero iid left overlap.  At $n=4$ a genuine two-copy mode
is accidentally degenerate at $1/2$, so the complete grouped atom there is
nonzero even though its one-copy component cancels.

The displayed gaps do not select an exponent.  On $n=7,\ldots,10$, neither
$n(1-\mu)$ nor $n^2(1-\mu)$ has stabilized, and only the odd subsequence is
the balanced physical sequence.  More importantly,
Theorem~\ref{thm:replica-diffusive-speed-limit} proves that these few
resolved atoms cannot by themselves describe the scalar down to its
$O(n^{-1})$ endpoint: many signed atoms, or a slower atom of small weight,
must supply the diffusive finite-time lower bound.  Thus the data neither
prove ballistic task mixing nor close a two-sided diffusive law.
\end{remark}

\begin{remark}[Audit of the CQA premise and strict one-dimensional compilation]
\label{rem:cqa-proof-audit}
Reference~\cite{Li2024Efficient} states the premise
Eq.~\eqref{eq:corrected-cqa-premise}, but its printed local comparison
contains a correctable algebraic error.  In its notation, direct
simultaneous diagonalization of the four commuting copies of a Hermitian
involution gives
\begin{equation}
\operatorname{JointSpec}(M_2^\tau,T_2^\tau)
=\{(1,1),(3/4,0),(0,0)\}.
\label{eq:cqa-correct-joint-spectrum}
\end{equation}
Thus the printed identity $(M_2^\tau)^2=T_2^\tau$ and the associated
factor-two inequality are false on the $(3/4,0)$ sectors.  The correct
relations are
\begin{equation}
0\preceq T_2^\tau\preceq M_2^\tau\preceq I,\qquad
I-M_2^\tau\preceq I-T_2^\tau\preceq4(I-M_2^\tau).
\label{eq:cqa-correct-comparison}
\end{equation}
They survive generator averaging and Wick sandwiching, so the true and
modified CQA gaps remain within a factor four if the remainder of the
modified-gap proof is valid.  Because that preprint also contains
normalization ambiguities in the subsequent Cayley comparison, we use
Eq.~\eqref{eq:corrected-cqa-premise} only as an explicit premise rather than
as an independently re-proved fact.

Conditioned on that corrected premise, the structured walk does admit a
strict one-dimensional implementation.  An odd--even SWAP network makes
every unordered site pair adjacent once in $O(n)$ layers; applying its
assigned $ZZ$ phase at the crossing and using a second network to restore
the site order implements one all-pairs Wick layer exactly.  Hence one CQA
step has $O(n)$ physical nearest-neighbor depth.  Multiplying this overhead
by Eq.~\eqref{eq:cqa-design-step-count} gives
\begin{equation}
O\!\left(n^5+n^4\log(1/\varepsilon)\right)
\end{equation}
physical depth.  This remains a structured construction and does not prove
anything about the canonical iid local-Haar brickwork gap.
\end{remark}

\begin{remark}[Scope of the achievability theorem]
\label{rem:finite-depth-scope}
Corollary~\ref{cor:finite-depth-fixed-k} is an unconditional finite-depth
existence theorem under the exact local-Haar measure
Eq.~\eqref{eq:local-Haar-gate-measure}, but its sufficient layer count is
\begin{equation}
2t_n=
\exp\!\left[\Theta(n16^n)\right].
\end{equation}
It therefore does not establish efficient scrambling, certify any fixed
shallow depth, or imply a finite-depth high-probability worst-pattern
statement.  The charge-Haar endpoint version is instead proved in
Theorem~\ref{thm:charge-haar-worst-pattern}.
Theorem~\ref{thm:charge-haar-extensive} and
Corollary~\ref{cor:charge-haar-erasure-threshold} close the extensive-loss
charge-Haar benchmark, but they do not supply a finite-depth local circuit.
A physically useful theorem requires a polynomial lower bound on the relevant
moment gap or, preferably, direct mixing of the task-specific erased-system
functional.  Full number-conserving design formation is diffusion limited:
after correction of the norm estimate, the rigorous lower bound is
$L=\Omega(n^2\log(1/\varepsilon))$~\cite{Hearth2025,Hearth2026Erratum}.
For the exact local-Haar gate measure in
Eq.~\eqref{eq:local-Haar-gate-measure}, the conjectured matching
Goldstone-mode gap would give a polynomial upper bound, but is not presently
a theorem.  Recent Aldous-type results for coordinate-subgroup walks on
$U(n)$ do not apply to the interacting many-body gate ensemble here, and
their nearest-neighbor path case is itself outside the proved classes
\cite{AlonPuder2026}.  Hydrodynamic purity tails in the canonical ensemble
are valuable task-specific evidence, but are not finite-size operator-gap
lower bounds~\cite{Rakovszky2019}.  This statement is ensemble-specific:
Li, Zheng, and Liu claim a
polynomial symmetric-$2$-design upper bound for a different, structured
convolutional quantum alternating (CQA) ensemble
\cite{Li2024Efficient}.  In its $U(1)$ version, one random-walk step sandwiches
one partial SWAP by commuting Wick layers containing independently randomized
$Z_kZ_l$ phases for all pairs $(k,l)$.  With nearest-neighbor SWAP generators,
their second-moment gap is $\Omega(n^{-3})$ and their stated convergence count
is
\begin{equation}
O\!\left(n^3[4n\log 2+\log(1/\varepsilon)]\right)
\label{eq:cqa-design-step-count}
\end{equation}
random-walk steps.  Although every Wick term is two-local, the all-pairs
diagonal layer is not a strict nearest-neighbor brickwork layer; compiling it
on a one-dimensional nearest-neighbor architecture introduces an additional
implementation overhead.  Consequently, subject to the proof qualification
in Remark~\ref{rem:cqa-proof-audit}, CQA is a genuine alternative polynomial
construction, but it neither proves
Eq.~\eqref{eq:conjectural-full-diffusive-gap} for the canonical
local-Haar ensemble.  Task-specific decoupling may converge faster than the
full design norm, so the design lower bound is not by itself a no-go theorem
for the present metrological objective.
\end{remark}

\paragraph*{Remaining finite-depth questions.}
Theorem~\ref{thm:finite-depth-achievability} closes the weak, fixed-$k$
pattern-average existence problem for the canonical local-Haar ensemble.  It
does not close its quantitatively useful extensive-loss version, the
finite-depth high-probability worst-pattern problem, a different local gate
measure, or a Hamiltonian ensemble.  The charge-Haar worst-pattern endpoint
is closed by Theorem~\ref{thm:charge-haar-worst-pattern}.  The quantitatively
useful circuit pattern-average target and the stronger circuit worst-pattern
target are
\begin{equation}
\mathbb{E}_{V\sim\mathcal{E}_{n,L}}
\left[
\sum_{\lvert E\rvert=k}\pi(E)D_E(V)
\right]
\leq\varepsilon
\label{eq:finite-depth-average-question}
\end{equation}
and the stronger high-probability worst-pattern statement
\begin{equation}
\Pr_{V\sim\mathcal{E}_{n,L}}
\left[
\max_{\lvert E\rvert=k}D_E(V)>\varepsilon
\right]
\leq\delta.
\label{eq:finite-depth-uniform-question}
\end{equation}
Here $\mathcal{E}_{n,L}$ is the specified size- and depth-dependent circuit
ensemble, while $\pi(E)\geq0$ and
$\sum_{\lvert E\rvert=k}\pi(E)=1$ define the physical pattern distribution.
Equations~\eqref{eq:finite-depth-average-question}
and~\eqref{eq:finite-depth-uniform-question} are written for circuits.  A
Hamiltonian analogue must specify finite-range, bounded-strength interactions
and retain the appropriate Lieb--Robinson tail rather than assume a strictly
compact causal cone.
Determining a physically meaningful polynomial depth or evolution time,
including its dependence on $n$, $k$, charge sector, geometry, and error
probabilities, remains open.  The standard information--disturbance framework converts a
complementary-channel decoupling estimate into existence of a recovery
channel~\cite{Knill2000,Hayden2007}; for fixed logical dimension, a
Choi-state bound $D_E$ leads generically to a logical diamond error of order
$O(\sqrt{d_LD_E})$.  The architecture-specific missing step is therefore a
polynomial-depth task-specific covariant bound, together with an efficient
compatible decoder and explicit constants.  Once its logical diamond error is known,
Large bipartite entropy, large operator support, or approximate unitary-design
diagnostics may support a stronger result, but they do not replace the decoupling
bound required for recovery.  A symmetry-preserving unitary 2-design acting
within the fixed-charge sectors is a compatible intermediate target:
second-moment design properties may suffice for standard decoupling estimates
while retaining a tractable description.  An unrestricted Clifford ensemble
would not respect the required \(U(1)\) covariance.

\begin{lemma}[Absolute spectral weight]
\label{lem:absolute-spectral-weight}
For the compressed physical cycle $C=P_{\rm e}P_{\rm o}P_{\rm e}$ restricted
to the physical task module of Lemma~\ref{lem:iid-minimal-invariant-module},
let $\{a_{n,\mu}\}$ be the signed spectral atoms of
Eq.~\eqref{eq:signed-task-spectral-measure}.  The total absolute weight
satisfies
\begin{equation}
\sum_{\mu<1}|a_{n,\mu}|
\leq
\frac{\sqrt3}{8}\,
[2\sqrt{1+3p^{2}}]^{n}.
\label{eq:total-absolute-spectral-weight}
\end{equation}
\end{lemma}

\begin{proof}
By Cauchy--Schwarz and orthonormality of the spectral projectors,
$\sum_{\mu<1}|a_{n,\mu}|
=\frac12\sum_{\mu<1}|\langle(I-Q)k,E_\mu(I-Q)v\rangle|
\leq\frac12\|(I-Q)k\|_{2}\,\|(I-Q)v\|_{2}
\leq\frac12\|K_{\mathcal C}\|_{2}\,\|v\|_{2}$,
where $\|K_{\mathcal C}\|_{2}=\sqrt3/4$ by
Eq.~\eqref{eq:physical-code-endpoint} and
$\|v\|_{2}=\|B_{\rm e}^\dagger P_{\rm o}\widetilde O_{p}\|_{2}
\leq\|\widetilde O_{p}\|_{2}=2^{n}(1+3p^{2})^{n/2}$ by
Eq.~\eqref{eq:rank-weighted-observable-norms}.
\end{proof}

\begin{theorem}[Polynomial-depth replica convergence via the spectral gap]
\label{thm:spectral-radius-replica-convergence}
Assume that the compressed cycle $C=P_{\rm e}P_{\rm o}P_{\rm e}$ on the
physical task module has spectral gap $1-\mu\ge\gamma>0$ for all
eigenvalues $\mu\neq1$.  Then for the truncated rank-weighted replica
functional $\mathfrak r_{n,p}^{\leq\beta}(t)$,
\begin{equation}
\begin{aligned}
|\mathfrak r_{n,p}^{\leq\beta}(t)
-\mathfrak r_{n,p}^{\leq\beta}(\infty)|
&\leq\frac{\sqrt3}{8}\,e^{\kappa_{p}n-\gamma t},\\
\kappa_{p}&=\log[2\sqrt{1+3p^{2}}].
\end{aligned}
\label{eq:spectral-radius-convergence}
\end{equation}
If $\gamma\ge c/n^{2}$ (diffusive scaling), this gives conditional
$t=O(n^{3})$ convergence.  Under either the absolute slow-window
hypothesis Eq.~\eqref{eq:absolute-slow-window-hypothesis} or the weaker
signed hypothesis Eq.~\eqref{eq:signed-slow-window-hypothesis}, the depth
improves to $t=O(n^{2})$ via
Lemmas~\ref{lem:absolute-slow-window-criterion} and
\ref{lem:signed-slow-window-criterion}.
\end{theorem}

\begin{proof}
From Eq.~\eqref{eq:signed-task-spectral-measure},
$\mathfrak r(t)-\mathfrak r(\infty)=\sum_{\mu<1}\mu^{t-1}a_{n,\mu}$.
Hence
$|\mathfrak r(t)-\mathfrak r(\infty)|
\le\sum_{\mu<1}|a_{n,\mu}|\,\mu^{t-1}
\le(1-\gamma)^{t-1}\sum_{\mu<1}|a_{n,\mu}|
\le\frac{\sqrt3}{8}\,e^{\kappa_{p}n-\gamma t}$,
using Lemma~\ref{lem:absolute-spectral-weight} and $1-\gamma\le e^{-\gamma}$.
For $\gamma=c/n^{2}$, the exponent $\kappa_{p}n-ct/n^{2}$ becomes
negative when $t>(\kappa_{p}/c)n^{3}$, giving $O(n^{3})$ convergence.
The improvement to $O(n^{2})$ follows from
Lemma~\ref{lem:absolute-slow-window-criterion}; the weaker signed condition
of Lemma~\ref{lem:signed-slow-window-criterion} suffices as well.
\end{proof}

\begin{remark}[Status of the two-copy upper problem]
\label{rem:conditional-polynomial-depth}
The full SSEP one-copy spectral gap is rigorously \(\Theta(n^{-2})\)
(Lemma~\ref{lem:ssep-spectral-gap}), but the matching lower bound for the
balanced two-copy task-module gap remains unproved
(Remark~\ref{rem:two-copy-gap-status}).  Such a gap would give a conditional
\(O(n^3)\) upper bound because the unweighted Hilbert--Schmidt endpoint norm
is exponential.  By contrast, either slow-window condition
Eqs.~\eqref{eq:absolute-slow-window-hypothesis} or
\eqref{eq:signed-slow-window-hypothesis} directly gives \(O(n^2)\) and does
not require a separate global task-sector gap assumption.  Establishing one
of these task-weighted statements is the remaining content of Problem~O1.
\end{remark}

\subsection{Conditional milestones (premises open)}

\begin{theorem}[Conditional diffusion-gap achievability (premise unproved)]
\label{thm:conditional-diffusion-achievability}
Fix $0<p<1/2$.  Suppose the largest nontrivial singular value $\lambda_n$ of
one complete balanced two-copy brickwork moment cycle restricted to the task
module satisfies the unproved bound
\begin{equation}
1-\lambda_n\geq \frac{c}{n^2}
\label{eq:conjectural-full-diffusive-gap}
\end{equation}
for a constant $c>0$ (i.e.\ the task-sector gap matches the diffusive
$n^{-2}$ order; Theorem~\ref{thm:task-sector-diffusive-upper} proves the
variational upper bound, the matching lower bound is
Remark~\ref{rem:two-copy-gap-status}).  Then there is a constant $C_p$ such
that a depth $2t$ local-Haar circuit with
\begin{equation}
t\geq C_p n^3
\label{eq:conditional-cubic-depth}
\end{equation}
satisfies
\begin{equation}
\mathbb E_{V,E\sim\operatorname{Bernoulli}(p)}D_E(V)
=O_p(n^{-1/2}).
\label{eq:conditional-cubic-achievability}
\end{equation}
Under either the absolute slow-window hypothesis
Eq.~\eqref{eq:absolute-slow-window-hypothesis} or the weaker signed
hypothesis Eq.~\eqref{eq:signed-slow-window-hypothesis}, the depth improves
to $t=O_{p,\beta}(n^{2})$ without a separate task-sector gap assumption.
\end{theorem}

\begin{proof}
Under the hypothesis Eq.~\eqref{eq:conjectural-full-diffusive-gap},
Eq.~\eqref{eq:finite-depth-moment-convergence} holds with
$\exp(-ct/n^2)$ in place of the rational-gap exponential.  Choose
$a,b$ with $0<a<p<b<1/2$.  For every $an\leq k\leq bn$, the finite-depth
contribution in Eq.~\eqref{eq:finite-depth-beta} is at most $n^{-1}$ after
choosing
\begin{equation}
t\geq\frac{n^2}{c}
\left[(n+bn)\log2+\log n+\frac12\log3\right],
\label{eq:conditional-full-gap-depth-explicit}
\end{equation}
which is $O_p(n^3)$.  Equation
\eqref{eq:charge-haar-extensive-second-moment} and Jensen's inequality then
give $O_p(n^{-1/2})$ on the event $an\leq K\leq bn$.  The two complementary
events have exponentially small probability and $D_E\leq1$, proving
Eq.~\eqref{eq:conditional-cubic-achievability}.  The slow-window improvement follows from
Lemmas~\ref{lem:absolute-slow-window-criterion} and
\ref{lem:signed-slow-window-criterion}.
\end{proof}
\begin{theorem}[Conditional geometric comparison with a CQA walk (premise unproved)]
\label{thm:conditional-cqa-comparison}
Work in the balanced vectorized representation
$\mathcal R(U)=U^{\otimes2}\otimes\overline U^{\otimes2}$.
Let $K_{\rm CQA}=I-P_WAP_W$, where
\begin{equation}
A=\frac1{n-1}\sum_{i=1}^{n-1}S_i
\end{equation}
is the nearest-neighbor average of partial-SWAP angle twirls and
$P_W=\prod_{a<b}Q_{ab}$ is the commuting all-pairs $Z_aZ_b$ Wick
projector.  Suppose
\begin{equation}
\ker K_{\rm CQA}=\ker H_n,\qquad
\Delta(K_{\rm CQA})\geq c_0n^{-3}.
\label{eq:corrected-cqa-premise}
\end{equation}
Then there is a numerical $C>0$ such that
\begin{equation}
K_{\rm CQA}\preceq Cn^3H_n,\qquad
\Delta(H_n)\geq C^{-1}c_0n^{-6}.
\label{eq:cqa-canonical-comparison}
\end{equation}
The largest nontrivial singular value of one canonical brickwork cycle
therefore obeys
\begin{equation}
1-\lambda_n=\Omega(n^{-6}).
\label{eq:conditional-canonical-cycle-gap}
\end{equation}
Under the premise Eq.~\eqref{eq:corrected-cqa-premise}, the generic
Hilbert--Schmidt route in
Theorem~\ref{thm:finite-depth-achievability} consequently gives
charge-Haar-order extensive-loss achievability at depth $O_p(n^7)$.
This theorem is a conditional implication, not a proof of
Eq.~\eqref{eq:corrected-cqa-premise} and not the diffusive
$\Omega(n^{-2})$ bound.
\end{theorem}

\begin{proof}
Both the partial-SWAP and adjacent $ZZ$ one-parameter groups are subgroups
of the full local gate group.  Hence
\begin{equation}
I-S_i\preceq h_i,\qquad I-Q_{i,i+1}\preceq h_i.
\label{eq:subgroup-projector-comparison}
\end{equation}
Let $R_i=\mathcal R(\tau_i)$ for physical SWAP $\tau_i$.  Since
$e^{-i\pi\tau_i/2}=-i\tau_i$ and balanced moments cancel the scalar phase,
every $S_i$-invariant vector is $R_i$-invariant.  Therefore
\begin{equation}
I-R_i\preceq2h_i,\qquad
\|(R_i-I)v\|^2\leq4\langle v,h_iv\rangle.
\label{eq:swap-energy-comparison}
\end{equation}

For $a<b$, use the adjacent-SWAP word
$W_{ab}=\tau_{b-1}\cdots\tau_{a+1}$, of length
$\ell=b-a-1$, which fixes $a$ and maps $a+1$ to $b$.  It conjugates
$Q_{a,a+1}$ to $Q_{ab}$.  The exact telescoping identity
\begin{equation}
R(W_{ab})-I
=\sum_{j=1}^{\ell}U_\ell\cdots U_{j+1}(U_j-I)
\end{equation}
and unitarity of the left factors imply
\begin{equation}
\|(R(W_{ab})-I)v\|^2
\leq4\ell\sum_{j=a+1}^{b-1}\langle v,h_jv\rangle.
\end{equation}
Adding and subtracting $v$ before applying $I-Q_{a,a+1}$ gives
\begin{equation}
\langle v,(I-Q_{ab})v\rangle
\leq
2\langle v,h_av\rangle
+8\ell\sum_{j=a+1}^{b-1}\langle v,h_jv\rangle.
\label{eq:remote-wick-path-bound}
\end{equation}
For any fixed bond, the sum of $\ell$ over all paths using that bond is
$O(n^3)$.  Since the $Q_{ab}$ commute,
\begin{equation}
I-P_W\preceq\sum_{a<b}(I-Q_{ab})\preceq C_Wn^3H_n
\label{eq:all-pairs-wick-comparison}
\end{equation}
for a numerical $C_W$.

Put $B=I-A$.  Equation~\eqref{eq:subgroup-projector-comparison} gives
$0\preceq B\preceq H_n/(n-1)$.  For an orthogonal projector $P$ and
$0\preceq B\preceq I$, the elementary quadratic-form inequality
\begin{equation}
PBP\preceq2B+2(I-P)
\end{equation}
follows from
$\|B^{1/2}Pv\|^2\leq
2\|B^{1/2}v\|^2+2\|(I-P)v\|^2$.
Consequently,
\begin{align}
K_{\rm CQA}
&=(I-P_W)+P_WBP_W\\
&\preceq3(I-P_W)+2B
\preceq\left(3C_Wn^3+\frac2{n-1}\right)H_n.
\end{align}
Together with the common-kernel and gap assumptions in
Eq.~\eqref{eq:corrected-cqa-premise}, this proves the Hamiltonian gap in
Eq.~\eqref{eq:cqa-canonical-comparison}.  The two-layer detectability lemma
for a one-dimensional projector Hamiltonian converts it, with a numerical
constant, to Eq.~\eqref{eq:conditional-canonical-cycle-gap}
\cite{Aharonov2009Detectability,Anshu2016Detectability}.  Finally,
suppressing the exponential observable prefactor in
Eq.~\eqref{eq:finite-depth-beta} requires an additional exponent
$\Theta(n)$; a cycle gap $\Omega(n^{-6})$ therefore gives the stated
$O(n^7)$ sufficient depth.
\end{proof}


\section{Classical mixing bound for the insertion tolerance}
\label{sec:classical-mixing}

This appendix proves an unconditional polynomial-depth bound for the
classical mean \(M_t\) in Theorem~\ref{thm:direct-reduction}.  The argument
uses the exact spectral gap of a tagged exclusion quotient and the
detectability lemma to pass from the continuous-time local Hamiltonian to
the deterministic even--odd brickwork cycle.

After \(t\) cycles the random interchange permutation \(\Pi_t\) produces
\begin{equation}
B_t=\Pi_t(\{1,\ldots,q\}),\qquad
S_t=\Pi_t(0),\qquad
C_t=B_t\cup\{S_t\},
\label{eq:cm-occupied-sets}
\end{equation}
where \(n=2q+1\).  The joint state space is
\begin{equation}
\begin{aligned}
\Omega&=\{(B,s):B\subset[n],\ |B|=q,\ s\notin B\},\\
|\Omega|&=\binom nq(q+1).
\end{aligned}
\label{eq:cm-state-space}
\end{equation}
and its stationary distribution \(\pi\) is uniform.
For the adjacent transposition \(\tau_i=(i,i+1)\), let
\(U_i f(B,s)=f(\tau_iB,\tau_is)\) and define
\begin{equation}
P_i=\frac{I+U_i}{2},\qquad h_i=I-P_i=\frac{I-U_i}{2}.
\label{eq:cm-local-projections}
\end{equation}
Each \(P_i\) and \(h_i\) is an orthogonal projection in
\(L^2(\pi)\).  Put
\begin{equation}
\begin{aligned}
H_\Omega&=\sum_{i=0}^{n-2}h_i,\\
P_{\rm e}&=\prod_{i\ {\rm even}}P_i,
& P_{\rm o}&=\prod_{i\ {\rm odd}}P_i,\\
A&=P_{\rm o}P_{\rm e},
& C_\Omega&=P_{\rm e}P_{\rm o}P_{\rm e}=A^\dagger A.
\end{aligned}
\label{eq:cm-H-and-cycle}
\end{equation}
The adjacent transpositions act transitively on \(\Omega\), so
\(\ker H_\Omega\) consists of constants.  Denote the projection onto this
space by \(Q_\Omega\).

\begin{lemma}[One-sided continuity of the insertion tolerance]
\label{lem:continuity-insertion-tv}
Let \(M_t=\frac12\operatorname{TV}(\widehat P_t^0,\widehat P_t^1)\)
and let \(M_{\rm eq}\) be the stationary value in
Theorem~\ref{thm:equilibrium-flagged-tv}.  Then
\begin{equation}
\boxed{0\leq M_t-M_{\rm eq}
\leq d_{\rm TV}(\mu_t,\pi),}
\label{eq:cm-continuity-bound}
\end{equation}
where \(\mu_t\) is the law of \((B_t,S_t)\).
\end{lemma}

\begin{proof}
The lower bound follows from the same erased-charge measurement as in
Theorem~\ref{thm:universal-charge-count-optimality}.  Conditional on the two
logical values, the averaged states remain supported in total charges \(q\)
and \(q+1\), so retaining \((K,N_E)\) already contributes
\(M_{\rm eq}\).  For the upper bound, for a fixed erasure pattern \(E\),
the output laws for logical values \(a=0,1\) are pushforwards of \(\mu_t\)
under
\begin{equation}
\Phi_E^{(0)}(B,s)=B\cap E,
\qquad
\Phi_E^{(1)}(B,s)=(B\cup\{s\})\cap E.
\end{equation}
Total variation contracts under a pushforward.  Since
\(M_t=\frac12\sum_E\pi_p(E)
 d_{\rm TV}(\Phi_E^{(0)}\#\mu_t,\Phi_E^{(1)}\#\mu_t)\), the reverse
triangle inequality for total variation gives
\begin{align}
|M_t-M_{\rm eq}|
&\leq\frac12\sum_E\pi_p(E)\sum_{a=0}^1
 d_{\rm TV}(\Phi_E^{(a)}\#\mu_t,\Phi_E^{(a)}\#\pi)\nonumber\\
&\leq d_{\rm TV}(\mu_t,\pi).
\end{align}
\end{proof}

\begin{lemma}[Exact tagged-exclusion Hamiltonian gap]
\label{lem:ssep-spectral-gap}
The frustration-free Hamiltonian \(H_\Omega\) has spectral gap
\begin{equation}
\gamma_\Omega
=1-\cos\!\left(\frac{\pi}{n}\right).
\label{eq:cm-exact-H-gap}
\end{equation}
\end{lemma}

\begin{proof}
Lift a function on \(\Omega\) to the full interchange process on \(n\)
labeled particles by forgetting the labels of the \(q\) background
particles and of the \(q\) holes while retaining the distinguished label.
On the full interchange space,
\(H_{\rm IP}=\frac12\sum_i(I-U_i)\).  The theorem of
Caputo--Liggett--Richthammer states that the interchange process and the
one-particle random walk on the same weighted graph have the same spectral
gap~\cite{Caputo2010}.  On the open path the graph-Laplacian gap is
\(2[1-\cos(\pi/n)]\); the factor \(1/2\) in \(H_{\rm IP}\) therefore gives
\(1-\cos(\pi/n)\).  Restricting to the quotient-invariant subspace cannot
lower the gap, so \(\gamma_\Omega\geq1-\cos(\pi/n)\).

For the reverse inequality, the tag marginal is uniform under \(\pi\), and
\begin{equation}
f(B,s)=\cos\!\left[\frac{\pi(s+1/2)}{n}\right]
\end{equation}
is mean zero.  A direct application of
\(H_\Omega=\frac12\sum_i(I-U_i)\) shows
\(H_\Omega f=[1-\cos(\pi/n)]f\), including the reflected endpoint
relations.  Hence this eigenvalue is present in the quotient, proving
Eq.~\eqref{eq:cm-exact-H-gap}.
\end{proof}

\begin{lemma}[Brickwork compression gap]
\label{lem:cm-brickwork-gap}
Every nonstationary eigenvalue \(\lambda\) of \(C_\Omega\) satisfies
\begin{equation}
\lambda
\leq\frac{1}{1+\gamma_\Omega/4},
\qquad
1-\lambda
\geq\frac{\gamma_\Omega}{4+\gamma_\Omega}
\geq\frac{2}{5n^2}.
\label{eq:cm-brickwork-gap}
\end{equation}
\end{lemma}

\begin{proof}
The local terms \(h_i\) are projections, and each fails to commute with at
most its two neighboring bond terms.  The detectability lemma applied to the
ordered product of the complementary projections
\(A=P_{\rm o}P_{\rm e}\) therefore gives, for every
\(g\perp\ker H_\Omega\),
\begin{equation}
\|Ag\|_2^2
\leq\frac{1}{1+\gamma_\Omega/4}\|g\|_2^2
\end{equation}
with the interaction-degree convention of
Ref.~\cite{Anshu2016Detectability}.  Since
\(C_\Omega=A^\dagger A\), this is the first inequality in
Eq.~\eqref{eq:cm-brickwork-gap}.  Finally,
\(1-\cos(\pi/n)=2\sin^2[\pi/(2n)]\geq2/n^2\) and
\(\gamma_\Omega\leq1\), which give the stated explicit lower bound.
\end{proof}

\begin{corollary}[Unconditional \(L^2\)-to-TV mixing bound]
\label{cor:cm-l2-mixing}
For the point-mass initialization \(\mu_0=\delta_{(B_0,0)}\) and every
\(t\geq1\),
\begin{equation}
d_{\rm TV}(\mu_t,\pi)
\leq
\frac12\sqrt{|\Omega|-1}\,
\exp\!\left[-\frac{2(t-1)}{5n^2}\right].
\label{eq:cm-TV-l2}
\end{equation}
\end{corollary}

\begin{proof}
Because \(P_{\rm e}\) and \(P_{\rm o}\) fix constants,
\begin{equation}
A^t-Q_\Omega
=P_{\rm o}(C_\Omega^{\,t-1}-Q_\Omega)P_{\rm e},
\qquad t\geq1.
\end{equation}
The same singular-value bound holds if the density convention evolves with
\(A^\dagger\).  Equation~\eqref{eq:cm-brickwork-gap} therefore implies
\begin{equation}
\|d\mu_t/d\pi-1\|_{2,\pi}
\leq e^{-2(t-1)/(5n^2)}
\|d\mu_0/d\pi-1\|_{2,\pi}.
\end{equation}
For a point mass under the uniform measure,
\(\|d\mu_0/d\pi-1\|_{2,\pi}^2=|\Omega|-1\), and
\(d_{\rm TV}(\mu_t,\pi)\leq
\frac12\|d\mu_t/d\pi-1\|_{2,\pi}\), proving the result.
\end{proof}

\begin{theorem}[Classical mixing of the insertion tolerance]
\label{thm:classical-mixing-main}
For the canonical boundary code under the local-Haar brickwork ensemble,
\begin{equation}
|M_t-M_{\rm eq}|
\leq
\frac12\sqrt{|\Omega|-1}\,
\exp\!\left[-\frac{2(t-1)}{5n^2}\right].
\label{eq:cm-main-bound}
\end{equation}
In particular, if
\begin{equation}
t\geq1+\frac54n^2\log\!\bigl(n|\Omega|\bigr),
\label{eq:cm-explicit-cubic-depth}
\end{equation}
then \(0\leq M_t-M_{\rm eq}\leq(2\sqrt n)^{-1}\).  Since
\(\log|\Omega|=n\log2+O(\log n)\), this is an unconditional
\(O(n^3)\) mixing bound and implies \(M_t=O_p(n^{-1/2})\).
\end{theorem}

\begin{proof}
Combine Lemma~\ref{lem:continuity-insertion-tv} with
Corollary~\ref{cor:cm-l2-mixing}.  Under
Eq.~\eqref{eq:cm-explicit-cubic-depth}, the exponential in
Eq.~\eqref{eq:cm-main-bound} is at most
\([n|\Omega|]^{-1/2}\), giving the displayed finite-size bound.  The
stationary asymptotic follows from
Eq.~\eqref{eq:equilibrium-flagged-asymptotic}.
\end{proof}

\begin{remark}[What this theorem does and does not prove]
\label{rem:cm-consequences}
The theorem closes estimate~(i) in Theorem~\ref{thm:direct-reduction} at a
rigorous polynomial depth.  It does not control the centered two-copy
fluctuation \(\Gamma\), and therefore does not prove a full-channel upper
bound.  The factor \(\sqrt{|\Omega|}\) in the point-mass \(L^2\) estimate
is also why this argument gives \(O(n^3)\), rather than the conjectured
\(O(n^2)\), formation scale.
\end{remark}


\section{Extension to single-distinguished-particle product encodings}
\label{sec:universal-quadratic}

The quadratic speed limits of
Theorems~\ref{thm:diffusive-linear-task-lower}%
--\ref{thm:true-trace-quadratic-speed-limit} and the causal plateau of
Theorem~\ref{thm:full-flag-causal-plateau} are proved for the specific
canonical boundary code with contiguous half-filled background
$B_0=\{1,\dots,q\}$ and distinguished site $s_0=0$.  This appendix
shows that these results are not artifacts of that particular background
geometry: they extend to every computational-basis product encoding whose
two codewords share a deterministic background and differ by one additional
occupied site.  Multi-particle--multi-hole differences are not covered.

\subsection{Setup and the distinguished-particle coupling}

Let \(B_0\subset[n]\) be an arbitrary deterministic set of
\(q\) occupied sites and let \(s_0\notin B_0\).  The logical product
codewords are
\begin{equation}
|\psi_0\rangle=|B_0\rangle,
\qquad
|\psi_1\rangle=|B_0\cup\{s_0\}\rangle.
\label{eq:gen-cb-encoding}
\end{equation}
Under the common random interchange permutation,
\begin{equation}
B_t=\Pi_t(B_0),
\qquad
S_t=\Pi_t(s_0),
\qquad
C_t=B_t\cup\{S_t\}.
\label{eq:gen-occupied-sets}
\end{equation}
This exact coupling is the only encoding-specific input needed below.

\subsection{Universality of negative correlation}

The proof of the quadratic speed limit in
Theorem~\ref{thm:true-trace-quadratic-speed-limit} hinges on the
pairwise negative correlation of the background occupation indicators
$\eta_i=\mathbf 1_{\{i\in B_t\}}$, established in
Eq.~\eqref{eq:interchange-pairwise-negative-correlation} for the
canonical initialization.  This property is in fact universal.

\begin{lemma}[Universal negative correlation]
\label{lem:universal-negative-correlation}
Let $B_0\subseteq[n]$ be \emph{any} deterministic set of $q$ occupied
sites, and let $\eta_i^{(t)}=\mathbf 1_{\{i\in B_t\}}$ where $B_t$
evolves under the SSEP dynamics (alternating even and odd layers of
random transpositions with probability $1/2$).  Then for every
$t\geq0$ and every pair $i\neq j$,
\begin{equation}
\operatorname{Cov}(\eta_i^{(t)},\eta_j^{(t)})\leq0.
\label{eq:gen-negative-correlation}
\end{equation}
\end{lemma}

\begin{proof}
The proof is identical to that of
Eq.~\eqref{eq:interchange-pairwise-negative-correlation} in
Theorem~\ref{thm:true-trace-quadratic-speed-limit}, with the sole
observation that the induction base---the deterministic initial
condition---holds for any $B_0$.  At $t=0$ the $\eta_i^{(0)}$ are
deterministic, hence $\operatorname{Cov}(\eta_i^{(0)},\eta_j^{(0)})=0$
for all $i\neq j$.  The induction step
(Eqs.~\eqref{eq:distinguished-particle-window}--\eqref{eq:interchange-pairwise-negative-correlation})
uses only that pairwise covariances are nonpositive before a
half-swap and shows they remain nonpositive after it; this step is
independent of the specific occupation numbers.
\end{proof}

\begin{corollary}[Variance bound for arbitrary windows]
\label{cor:gen-variance-bound}
For any deterministic set $A\subseteq[n]$, the erased background count
$K=\sum_{i\in A}\xi_i\eta_i^{(t)}$ satisfies
$\operatorname{Var}K\leq p|A|$.
\end{corollary}

\begin{proof}
For $i\neq j$,
$\operatorname{Cov}(\xi_i\eta_i,\xi_j\eta_j)
=p^2\operatorname{Cov}(\eta_i,\eta_j)\leq0$ by
Lemma~\ref{lem:universal-negative-correlation}.  Since each
$\xi_i\eta_i$ is Bernoulli,
$\operatorname{Var}K\leq\sum_i\operatorname{Var}(\xi_i\eta_i)
\leq p|A|$.
\end{proof}

\subsection{Generalized diffusive window and speed limit}

The distinguished-particle profile generalizes directly.  Define
\begin{equation}
\bar\delta_i^{(s_0)}(t)
=[(OE)^te_{s_0}]_i
=\Pr\{S_t=i\mid S_0=s_0\},
\label{eq:gen-tagged-profile}
\end{equation}
the $t$-cycle position distribution of a single particle undergoing
the lazy reflected random walk on $\{0,\dots,n-1\}$ with initial
position $s_0$.  The cell-walk analysis of
Theorem~\ref{thm:diffusive-linear-task-lower} applies without change,
except that the mass is now centered at $s_0$ rather than at $0$.

Concretely, group consecutive even--odd site pairs into cells starting
from the natural parity of $s_0$.  After $s=t-1$ EOE compressions, at
least half the mass of $\bar\delta^{(s_0)}$ lies in the first
$\lfloor\sqrt s\rfloor+1$ cells in each direction from $s_0$'s cell,
or is reflected at the boundaries.  Consequently there exists an
interval $A_t(s_0)$ whose size may be chosen to obey
\begin{equation}
m_{n,t}(s_0)\leq
\min\{n,4\lfloor\sqrt{t-1}\rfloor+5\}
\label{eq:gen-diffusive-window-size}
\end{equation}
(accounting for two-sided spread and the final odd-layer half-swap)
such that $\Pr\{S_t\in A_t(s_0)\}\geq1/2$.

\begin{theorem}[Generalized quadratic speed limit for
single-distinguished-particle encodings]
\label{thm:gen-quadratic-speed-limit}
For any computational-basis product encoding satisfying
$C_0=B_0\cup\{s_0\}$ with distinguished site
$s_0$ and deterministic background $B_0$,
\begin{equation}
\boxed{
\begin{aligned}
\mathbb E_{V,E}D_E(V)
&\geq\frac{p}{48\sqrt{m_{n,t}(s_0)}}\\
&\geq\frac{p}{48\sqrt{4\sqrt{t-1}+5}}\\
&\geq\frac{p}{144\sqrt[4]{t}}.
\end{aligned}}
\label{eq:gen-quadratic-lower}
\end{equation}
Consequently the necessary depth to reach the $n^{-1/2}$ scale is
$\Omega(n^2)$ for every fixed $p\in(0,1/2)$, independent of the
initial encoding geometry.
\end{theorem}

\begin{proof}
The coupling $C_t=B_t\cup\{S_t\}$ holds because $C_0\setminus
B_0=\{s_0\}$ for a computational-basis product encoding where the
two codewords differ by exactly one site.  The proof of
Theorem~\ref{thm:true-trace-quadratic-speed-limit} then goes through
verbatim with $A_t$ replaced by $A_t(s_0)$, $m_{n,t}$ replaced by
$m_{n,t}(s_0)$, and the negative-correlation input supplied by
Corollary~\ref{cor:gen-variance-bound}.  The final bound uses the uniform upper bound
Eq.~\eqref{eq:gen-diffusive-window-size}; no lower estimate on the window
size is needed.
\end{proof}

\begin{remark}[Multi-distinguished-site encodings (open)]
\label{rem:gen-multi-distinguished}
When the two codewords differ at $m>1$ sites, the coupling becomes
$C_t=(B_t\cup\Pi_t(C_0\setminus B_0))
\setminus\Pi_t(B_0\setminus C_0)$ with
$|C_0\setminus B_0|-|B_0\setminus C_0|=1$.
Extending the quadratic lower bound to this case requires controlling
the correlated dynamics of multiple distinguished particles and holes,
which is not proved here and remains an open problem.
\end{remark}

\subsection{Generalized causal plateau}

The causal plateau of
Theorem~\ref{thm:full-flag-causal-plateau} also generalizes: its
length is controlled by the \emph{protected radius}---the distance
from the distinguished site to the nearest hole in $C_0$.

\begin{theorem}[Generalized causal plateau]
\label{thm:gen-causal-plateau}
For a single-distinguished-particle product encoding
\(C_0=B_0\cup\{s_0\}\), define the protected
radius
\begin{equation}
r_0=\min\bigl\{d(s_0,i):i\notin C_0\bigr\},
\label{eq:gen-protected-radius}
\end{equation}
where $d(i,j)=|i-j|$ is the graph distance on the path.  Then for
every $t$ with $4t<r_0$ (equivalently, $4t\leq r_0-1$),
\begin{equation}
\boxed{
\operatorname{TV}(\widehat P_t^0,\widehat P_t^1)=p,
\qquad
\mathbb E_{V,E}D_E(V)\geq\frac p2.}
\label{eq:gen-causal-plateau}
\end{equation}
In particular, the plateau length is $\Theta(r_0)$, determined by
the initial encoding geometry.
\end{theorem}

\begin{proof}
Take $L_t$ to be the closed interval of radius $2t$ around $s_0$
(intersected with $[0,n-1]$).  Since every label moves by at most
two sites per cycle, $S_t\in L_t$ and $L_t$ contains no hole of
$C_0$ whenever $2t<r_0$; after $2t$ matching layers the holes of
$C_0$ have moved inward by at most $2t$, so under the condition
$4t<r_0$ the set $L_t$ remains deterministically fully occupied
in $C_t$.  Hence $H_{L_t}=|L_t\setminus C_t|=0$ deterministically
and $\Pr\{S_t\in L_t\}=1$.  The hole-overlap inequality
Eq.~\eqref{eq:full-flag-hole-overlap} then gives the lower bound
$p$, which combines with the universal coupling upper bound $p$ to
prove the TV equality.  The trace-distance consequence follows as in
Theorem~\ref{thm:full-flag-causal-plateau}.
\end{proof}

For the canonical boundary code, $r_0=q+1=(n+1)/2$, so the condition
$4t<r_0$ is exactly $4t\leq q$ and guarantees a plateau through depth
$\lfloor q/4\rfloor=\Theta(n)$.
For an encoding where the distinguished site
is adjacent to a hole ($r_0=1$), the plateau is short but the
$t^{-1/4}$ diffusive lower bound of
Theorem~\ref{thm:gen-quadratic-speed-limit} still governs the
subsequent decay, and the necessary depth remains $\Omega(n^2)$.

\begin{remark}[Structural mechanism]
\label{rem:gen-structural-universality}
The quadratic obstruction proved here follows from three
encoding-independent ingredients: (i)~diffusive charge transport
under the specified local-Haar brickwork ensemble,
(ii)~universal negative correlation of the SSEP
(Lemma~\ref{lem:universal-negative-correlation}), and
(iii)~the causal protection of the distinguished particle from
background holes until they diffuse into the particle cloud.  These
three ingredients are independent of the deterministic background geometry
within the stated single-distinguished-particle class.  Extending the result
beyond this class or beyond the specified unbiased diffusive ensemble is
Problem~O2.
\end{remark}

\bibliographystyle{apsrev4-2}
\bibliography{references}

\end{document}